\documentclass[11pt]{article}
\usepackage[utf8]{inputenc}
\usepackage{fullpage}
\usepackage{geometry}
\usepackage{mathtools}
\usepackage{blkarray}
\usepackage{amsmath}
\usepackage{tikz}
\usetikzlibrary{quantikz}
\usepackage{graphicx}
\usetikzlibrary{decorations.pathmorphing}
\usepackage{amsmath}
\usepackage{bbm}
\usepackage{amssymb}
\usepackage{hyperref}
\usepackage{comment}
\usepackage{authblk}
\usepackage{float}
\usepackage{bigints}
\usepackage{relsize}
\usepackage{physics}
\usepackage{booktabs}
\usepackage[math]{cellspace}

\usetikzlibrary{decorations.markings,decorations.pathreplacing}
\usepackage{amsthm}
\usepackage{bm}

\DeclarePairedDelimiterX{\infdivx}[2]{(}{)}{%
  #1\;\delimsize\|\;#2%
}

\newcommand{\eq}[1]{(\ref{eq:#1})}
\newcommand{\eqn}[1]{(\ref{eqn:#1})}
\newcommand{\rem}[1]{\hyperref[rem:#1]{Remark~\ref*{rem:#1}}}
\newcommand{\thm}[1]{\hyperref[thm:#1]{Theorem~\ref*{thm:#1}}}
\newcommand{\asm}[1]{\hyperref[asm:#1]{Assumption~\ref*{asm:#1}}}
\newcommand{\cor}[1]{\hyperref[cor:#1]{Corollary~\ref*{cor:#1}}}
\newcommand{\defn}[1]{\hyperref[def:#1]{Definition~\ref*{def:#1}}}
\newcommand{\lem}[1]{\hyperref[lem:#1]{Lemma~\ref*{lem:#1}}}
\newcommand{\prop}[1]{\hyperref[prop:#1]{Proposition~\ref*{prop:#1}}}
\newcommand{\fig}[1]{\hyperref[fig:#1]{Figure~\ref*{fig:#1}}}
\newcommand{\tab}[1]{\hyperref[tab:#1]{Table~\ref*{tab:#1}}}
\newcommand{\algo}[1]{\hyperref[algo:#1]{Algorithm~\ref*{algo:#1}}}
\renewcommand{\sec}[1]{\hyperref[sec:#1]{Section~\ref*{sec:#1}}}
\newcommand{\append}[1]{\hyperref[append:#1]{Appendix~\ref*{append:#1}}}

\newtheorem{theorem}{Theorem}[section]
\newtheorem{assumption}[theorem]{Assumption}
\newtheorem{corollary}[theorem]{Corollary}
\newtheorem{lemma}[theorem]{Lemma}
\newtheorem{definition}[theorem]{Definition}

\newtheorem{conjecture}[theorem]{Conjecture}
\newtheorem{remark}[theorem]{Remark}
\newtheorem{proposition}[theorem]{Proposition}
\newtheorem{example}[theorem]{Example}
\newtheorem{fact}[theorem]{Fact}

\usepackage{cleveref}

\DeclareMathOperator{\Enc}{Enc}
\DeclareMathOperator{\poly}{poly}

\newcommand{\R}{\mathbb{R}}

\newcommand{\bbC}{\mathbb{C}}
\newcommand{\bbE}{\mathbb{E}}
\newcommand{\cA}{\mathcal{A}}
\newcommand{\cE}{\mathcal{E}}
\newcommand{\cF}{\mathcal{F}}
\newcommand{\cG}{\mathcal{G}}

\newcommand{\cK}{\mathcal{K}}
\newcommand{\cO}{\mathcal{O}}
\newcommand{\cP}{\mathcal{P}}

\newcommand{\cV}{\mathcal{V}}

\def\>{\rangle}
\def\<{\langle}

\newcommand{\nc}{\newcommand}
\nc\benum{\begin{enumerate}}
\nc\eenum{\end{enumerate}}
\nc\bit{\begin{itemize}}
\nc\eit{\end{itemize}}
\def\be#1\ee{\begin{equation}#1\end{equation}}
\def\ba#1\ea{\begin{align}#1\end{align}}
\def\bas#1\eas{\begin{align*}#1\end{align*}}

\newcommand{\eps}{\epsilon}

\title{Exponential speedups for quantum walks in \\ random hierarchical graphs}
\author[1]{Shankar Balasubramanian}
\author[1, 2]{Tongyang Li}
\author[1]{Aram W. Harrow}

\affil[1]{Center for Theoretical Physics, Massachusetts Institute of Technology, Cambridge, MA 02139, USA}
\affil[2]{Center on Frontiers of Computing Studies and School of Computer Science, Peking University, Beijing 100871, China}

\begin{document}

\maketitle

\begin{abstract}
There are few known exponential speedups for quantum algorithms and these tend to fall into even fewer families.  One speedup that has mostly resisted generalization is the use of quantum walks to traverse the welded-tree graph, due to Childs, Cleve, Deotto, Farhi, Gutmann, and Spielman.  We show how to generalize this to a large class of hierarchical graphs in which the vertices are grouped into ``supervertices'' which are arranged according to a $d$-dimensional lattice.  Supervertices can have different sizes, and edges between supervertices correspond to random connections between their constituent vertices.

The hitting times of quantum walks on these graphs are related to the localization properties of zero modes in certain disordered tight binding Hamiltonians.  The speedups range from superpolynomial to exponential, depending on the underlying dimension and the random graph model.
We also provide concrete realizations of these hierarchical graphs, and introduce a general method for constructing graphs with efficient quantum traversal times using graph sparsification.
\end{abstract}

\tableofcontents


\section{Introduction}
\paragraph{Motivation.}
Superpolynomial quantum speedups are an ultimate motivation for building quantum computers and it is a pressing research priority to broaden our currently short list of such speedups.  Notable examples of superpolynomial quantum-classical separation have mainly been discovered in algebraic and number theoretic problems, such as factoring/discrete logarithm~\cite{Shor99}, the hidden subgroup problem~\cite{BonehLipton,EHK04,Kup05}, Pell's equation and the principal ideal problem~\cite{Hal07}, the unit group problem~\cite{EHKS14}, etc.

A striking example of exponential quantum-classical separation without
apparent group structure is the welded tree problem proposed by
Childs, Cleve, Deotto, Farhi, Gutmann, and Spielman~\cite{CCDFGS03},
building on earlier work on the boolean
hypercube~\cite{MR02,Kempe05}. The welded tree is a graph consisting
of two depth-$n$ binary trees attached together by an alternating
cycle on the leaves of the two trees, with the two roots of the two
binary trees referred to as the entrance and the exit,
respectively. It can be shown that starting from the entrance, a
quantum walk can hit the exit with high probability in time $\poly(n)$
(the state-of-the-art quantum algorithm takes $O(n)$ queries and
$O(n^2)$ time~\cite{jeffery2022walk}), whereas classical algorithms
require $2^{\Omega(n)}$ queries for the same task under a reasonable
oracle model.  A promising feature of this result is that it does not
resemble most earlier speedups which rely on a group structure and
that the algorithm for it cannot be parallelized~\cite{CM20}.
Unfortunately the construction is artificial and the argument is
rigid; small changes to the graph, such as introducing short cycles,
will break the argument, and it has mostly resisted generalization,
although a version relying on group structure
exists~\cite{krovi2007quantum}.

Beyond the welded tree, a general understanding of superpolynomial and exponential quantum speedup of hitting times on graphs is relatively limited. Given the wide applicability of quantum algorithms on graphs as well as the general interest in superpolynomial quantum speedup, in this paper we discuss how to realize a wide family of graphs capable of achieving superpolynomial and exponential separation of hitting times between quantum and classical computing.  Our contribution is not only to prove a superpolynomial and/or exponential speedup on this class of graphs, but also to provide insight into when to expect a speedup for such problems.

\paragraph{Quantum walks in continuous time.}
Our paper studies the algorithmic power of quantum walks.  Many different forms of quantum walks have been studied, such as discrete-time walks, continuous-time walks, coined walks, the Szegedy walk, etc. \cite{ambainis2003quantum, venegas2012quantum}.  In this paper, we will deal with continuous-time walks on undirected graphs. Given an undirected graph (possibly with self-loops) with adjacency matrix $A$, define the Hamiltonian (also called the hopping Hamiltonian or the graph Hamiltonian) to be $H=-A$. A continuous-time quantum walk is simply the time evolution of $H$, i.e.~$e^{-iHt}$.  Our algorithms will consist of starting on a given vertex, running this walk for a time $t$ drawn from an appropriate distribution, and then measuring.  Such algorithms have been studied in various contexts in the past, such as spatial search on a lattice \cite{ChildsGoldstone, ChildsGoldstoneDirac}, quantum algorithms determining whether all elements of a list are distinct \cite{ambainis2007quantum}, triangle finding in graphs \cite{le2014improved}, and computing an output of a boolean expression corresponding to a tree of NAND gates acting on a binary string labelling the leaves \cite{farhi2007quantum}.

\paragraph{Hierarchical graph model.}
We prove a superpolynomial quantum speedup for quantum walks hitting a certain exit node on a wide range of random graph ensembles.  A graph $\cG$ drawn from this ensemble is called a \emph{hierarchical graph}, and is built from a graph $G=(V,E)$, a degree $D$, and size function $s_v$ for vertices and $e_{v,w}$ (sometimes denoted $e_{vw}$) for edges.  Each $v\in V$ is called a ``supervertex'' and corresponds to a set $S_v$ of vertices, with $|S_v|=s_v$.  For each superedge $(v,w)\in E$ we create a set $\cE_{v,w}$ consisting of $e_{v,w}$ random edges between $S_v$ and $S_w$ subject to additional constraints which we discuss in the next section.  Putting these together we obtain $\cG \triangleq (\cV \triangleq \cup_v S_v, \cE \triangleq \cup_{v,w} \cE_{v,w})$.  A simple example of such a graph, in the case where $G$ is a line graph, is in \Cref{fig:1D-graph-example}.

\begin{figure}[t]
    \centering
    \includegraphics[scale=0.28]{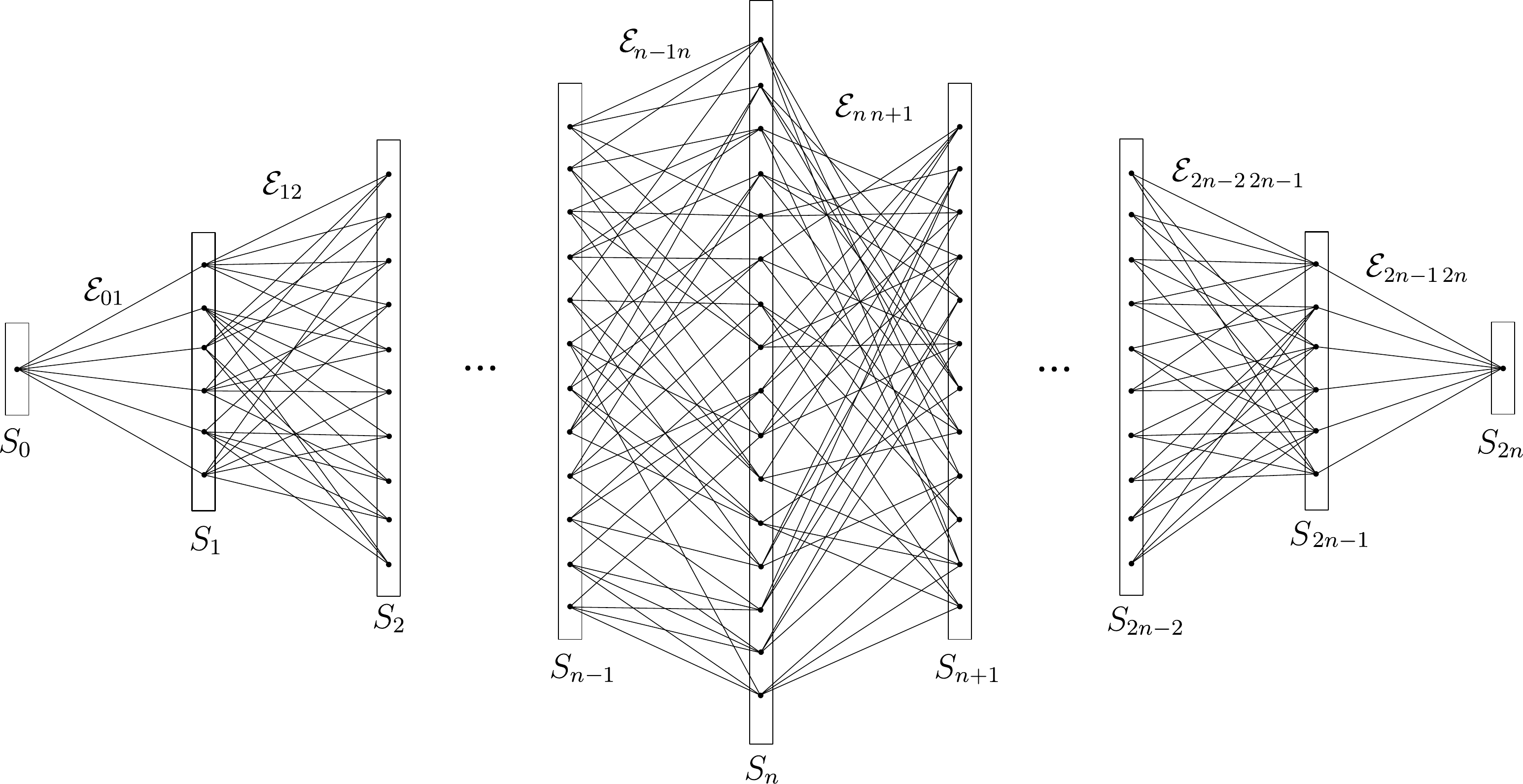}
    \caption{An example of a one-dimensional random hierarchical graph with degree $D = 6$.}
    \label{fig:1D-graph-example}
\end{figure}

\paragraph{One-dimensional (1D) supergraphs.}
A special case of our graph model bears resemblance to the welded tree from \cite{CCDFGS03}: $(V,E)$ is a line graph of length $2n+2$ with vertex sizes $s_0 = 1,s_1 = 2,s_2 = 4,s_3 = 8,\ldots, s_{n-1} = 2^{n-1},s_n = s_{n+1} = 2^n$, $s_{n+2} = 2^{n-1},\ldots, s_{2n} = 2,s_{2n+1} = 1$ and edge weights all equal to 2 except for the middle edge (between supervertices $n,n+1$ which each have size $2^n$) which has weight 1.  In \cite{CCDFGS03}, the quantum algorithm worked only if the $s_i / s_{i+1}$ changed in a single location.  This result could be plausibly extended up to $O(1)$ locations, meaning that the sequence $s_0,s_1,\dots$ should alternate a constant number of times between exponentially growing or shrinking.  Moreover, the classical lower bound in  \cite{CCDFGS03} relied on the graph lacking short cycles, which forced the graph structure to resemble binary trees.  These rigid constraints has caused the speedup in \cite{CCDFGS03}  to thus far resist generalization.

In our 1D random graph models, we can analyze the performance of quantum walks for any sequence of supervertex sizes $s_0,s_1,\ldots$ where $|\log(s_i/s_{i+1})|$ is $O(1)$.   We consider the case of random connectivity between supervertices, but our analysis can even accommodate some structure.  In particular, we show that so long as there is not too much local structure that reveals a classical walker's position within the graph, classical algorithms cannot significantly outperform a nonbacktracking random walk.  Random connectivity has this property, but many forms of more structured connectivity can as well.  An example of a 1D hierarchical graph with random connectivity is illustrated in \fig{1D-graph-example}.  In particular, while the example in \cite{CCDFGS03} is somewhat isolated, our class of random graph ensembles casts a much wider net, encompassing most hierarchical graphs built from a line graph $G$.  Later we will further generalize to hierarchical graphs built from $d$-dimensional lattices, and then to hierarchical graphs constructed from graph sparsification.

For 1D hierarchical graphs ($d=1$), we show the following result:
\begin{theorem}[Informal; see also \thm{main-1D},\cor{1D-LB}, and \cor{general-1D-LB}]
  Let $G$ be a line of $n$ vertices $\{1,\ldots,n\}$ and construct $\cG$ from $G$ as well as a degree $D=O(1)$, some vertex sizes $s_1,\ldots,s_n$ and edge sizes $e_1,\ldots,e_{n-1}$.
  (Let $e_i \triangleq e_{i-1,i}$.)  Suppose $s_1=s_n=1$ and consider the problem of reaching the supervertex $n$ starting from supervertex 1.  Assume that constants $c_1,c_2$ exist such that $c_1 \leq s_i / s_{i+1} \leq c_2$ for all $i$.  With high probability over $\cG$ we have the following:

  \bit
\item Any classical algorithm that starts at supervertex 1 and reaches supervertex $n$ with non-negligible probability needs to take time $(\max_i s_i)^{\Omega(1)}$.
\item A quantum walk starting at supervertex 1 can reach supervertex $n$ in time $T$ and with probability $p$, with $T$ and $1/p$ both bounded by $\sim \exp(\max_{k, \ell}\abs{\sum_{i=k}^{\ell} \log e_{2i} - \log e_{2i+1}})\cdot \poly(n)$.
  \eit
\end{theorem}

To understand this theorem, let us first consider some special cases.  In the welded tree we have $\max s_i = 2^n$ corresponding to an $\exp(O(n))$ classical runtime, while the quantum runtime is $\poly(n)$.  A generalization of this would have the $s_i$ rise and fall but also have some random fluctuations.  For example, we might choose $s_i/s_{i+1}$ to be a random variable with expectation $\alpha>1$ for $i<n/2$ and expectation $1/\alpha$ for $i>n/2$.  In this case, $\max s_i$ would be close to $\alpha^{n/2}$ which is exponentially large in $n$, while the quantum runtime would be $\exp(O(\sqrt{n}))$.  In this more generic case, we have a superpolynomial quantum speedup.  Thus we see that the exponential speedup of the welded-tree graph is based on a fine-tuned choice of parameters, and superpolynomial speedups are more common for 1D hierarchical graphs.

A crucial reason why the hitting time for the quantum walk on these graphs is fast is that the quantum dynamics is constrained in a way that the effective dimension explored by the walk is small. We can understand the effective dimension in terms of the Krylov subspace of the walk.  Given a starting vector $\ket\psi$, the Krylov subspace $\cK$ is defined to be the span of $\ket\psi, H\ket\psi, H^2\ket\psi, H^3\ket\psi, \ldots$.   If $\dim\cK$ is much smaller than the dimension of the ambient Hilbert space, then the dynamics is sufficiently restricted and the hitting time can be efficient for a quantum walk.  For hierarchical graphs, we will argue that Krylov subspaces are small when the starting state is the uniform superposition over one of the supervertices.
However, even when the dimension of the Krylov subspace is small, there may be other effects within the Krylov subspace that could inhibit fast dynamics of the quantum walk.  An example of this is if eigenstates of the Hamiltonian in the Krylov subspace are spatially localized with exponentially decaying tails.

\paragraph{Eigenstate localization.} To study this class of random
graph models within the Krylov subspace, we need to address the
possibility of eigenstate localization-- the simplest example of this
is \emph{Anderson localization}.  Suppose that one is given a
Hermitian matrix (Hamiltonian) describing the adjacency matrix of a
one-dimensional line (equivalently, a free particle on a line).  The
eigenstates are plane waves and are therefore delocalized, implying
that they have non-negligible amplitudes at both the beginning and the
end of the chain. If one adds random diagonal terms to the
Hamiltonian, the Hamiltonian is no longer translationally invariant;
in one dimension, this is manifested in the eigenstates localizing
with exponential tails.  This was first argued by Anderson in 1958
\cite{Anderson}.  While this may appear to be a no-go result for an
efficient quantum walk based algorithm, the class of random graph
models we study can be mapped onto a disordered Hamiltonian with
bond/hopping disorder.  In this case, for part of the spectrum, only a weaker form of localization occurs; the tails decay as $\exp(-O(\sqrt{n}))$, instead of $\exp(-O(n))$.  As we have seen, this is enough for superpolynomial quantum speedup.
To show this we rely on a result due to Kotowski and
Vir{\'a}g~\cite{KV17} regarding the Dyson singularity for random
hopping {S}chr\"{o}dinger operators to prove that such Hamiltonians
have a portion of their spectrum delocalized; see \sec{Dyson} for details, or \sec{Haminv} for an alternate proof.
We relate the localization properties of these
delocalized states to the hitting time (\lem{evolution}), therefore
arriving at the advertised theorem.

\paragraph{Higher-dimensional supergraphs.}

Delocalization is believed to be more common in higher dimensions, intuitively because there are more directions for a particle to escape.  As a result, we expect these sorts of hierarchical graphs to become even easier to traverse if $G$ has higher connectivity than a line graph.  An example we analyze is a $d$-dimensional lattice.
The situation is analogous to continuous-time quantum walks for spatial search, where higher-dimensional lattices are easier to traverse~\cite{ChildsGoldstone,ChildsGoldstoneDirac,ChildsGe}.

The simplest example of a $d$-dimensional lattice is the cubic lattice.  Here a subtlety arises that prevents us from obtaining a simple rigorous answer.  While for $d=1$, there is a single zero mode which we can use for our analysis, in a $d\geq 2$-dimensional cubic lattice with random hopping, there is a large subspace of zero modes, and interference within this subspace could lead to behavior that our current methods cannot predict.  We leave the analysis of these cases of higher-dimensional zero-eigenvalue subspaces to future work.

Instead, we are able to analyze a simple variant of the cubic lattice
called the \emph{Lieb lattice}.  It can be obtained by starting with the
cubic lattice on $[N]^d$ and adding a new vertex at the midpoint of
each edge.  We further introduce a requirement on the hopping
amplitudes to ensure that there is only a single zero mode that
has support on both the entrance and exit nodes.  In this case, a
large-dimensional subspace of zero modes still exists, but this subspace does not
have any overlap on the entrance and exit nodes.  As in the 1D case,
the performance of the quantum walk will depend on particular choices
of supervertex sizes, or equivalently, of hopping amplitudes.  While
our techniques can be applied to any choice of hopping matrix elements, for concreteness we
focus on one representative case called the biased Gaussian Free Field
(BGFF).

The BGFF is a natural generalization of our 1D models in the following
sense.  In our 1D models the sizes of supervertices follow a biased
random walk, first rising and then falling.  The continuous-space
version of a random walk is Brownian motion with a drift, and the
generalization of this to higher dimensions is the BGFF.  In each
case, the drift (or bias) term governs the typical shape of the graph, and any
specific graph we sample from this distribution will be described by
some random deviation around this.  Similarly to the 1D random walk, the
Gaussian free field allows adjacent supervertices in our $d$-dimension
lattice to have similar sizes, while allowing large-scale variations over
the entire graph.  See \sec{high-dim} for a precise definition of this
model.

Analyzing this graph model yields the following result.  Full details
are given in \sec{highdim}.
\begin{theorem}[Informal; see also \thm{2D-main}, \thm{highD-main}, and \cor{GFF-LB}]
  Let $G$ be a $d$-dimensional cubic lattice of $N^d$ sites with additional sites decorating the edges, forming a Lieb lattice.  Construct $\mathcal{G}$ from $G$ with degree $D = O(d)$, and define vertex sizes labelled by $s_{(m_1,m_2, \ldots, m_d)}$ where $(m_1,m_2, \ldots, m_d)$ indicates a lattice coordinate.  Suppose $s_{(0,0,\ldots,0)} = s_{(N,N,\ldots,N)} = 1$ and consider the problem of reaching the supervertex at $(N,N,\ldots,N)$ from the supervertex at $(0,0,\ldots,0)$.  Assume that ratios of adjacent vertex sizes $c_1 \leq s_u/s_v \leq c_2$ for $u \sim v$ and ratios of adjacent edge sizes are sampled from a biased Gaussian free field where the drift is positive near $(0,0,\ldots,0)$ and negative near $(N,N,\ldots,N)$.  With high probability over $\mathcal{G}$ we have the following:
  \begin{itemize}
  \item Any classical algorithm that starts at supervertex $(0,0,\ldots,0)$ and reaches supervertex $(N,N,\ldots,N)$ with non-negligible probability needs to take time $(\max_i s_i)^{\Omega(1)}$.
  \item A quantum walk starting at supervertex $(0,0,\ldots,0)$ can reach supervertex $(N,N,\ldots,N)$ in time $T$ and dimension $d \leq k_1 \log N$ for some constant $k_1$ and with probability $p$, with $T$ and $1/p$ bounded by $O\big(\poly\big(N^d \exp\kern-0.8mm\big(\sqrt{d\, 2^d \log N}\big)\big)\big)$.
  \item A quantum walk starting at supervertex $(0,0,\ldots,0)$ can reach supervertex $(N,N,\ldots,N)$ in time $T$ and dimension $k_2 \log N \leq d \leq k_3 N$ for constants $k_2$ and $k_3$ and with probability $p$, with $T$ and $1/p$ bounded by $O\big(\poly\big(N^d \exp\kern-0.8mm\big(d\big)\big)\big)$.
  \end{itemize}
\end{theorem}

To compare the quantum runtime and the classical runtime, we need the expected value of $\max_i s_i$.  Assuming a positive drift near the entrance supervertex and a negative drift near the exit supervertex, the size of a supervertex at a point in the middle of the lattice scales exponentially with the shortest path from the entrance to the desired point; this gives $\max_i s_i = \Omega(\exp(d N))$ for the $d$-dimensional Lieb lattice.  In comparison, for $d = O(N)$, the quantum algorithm takes time $\exp(O(d \log N))$ (ignoring subleading factors), which indicates an exponential speedup for the quantum walk as a function of $N$ for fixed $d$.  An open question is whether one can prove that this exponential separation persists for $d = \Omega(N)$.

We expect that the techniques in this section could be applied to graph ensembles well beyond the examples we studied.  Instead of starting with a cubic lattice, {\em any} bipartite graph can be turned into a Lieb lattice by bisecting each edge with a new vertex.  We similarly obtain a special zero-energy eigenstate with many of the same features we needed above, but some questions remain about how to construct the random-graph ensemble.  An intriguing open question is whether there are natural forms of random fluctuations in these graph models that would yield $\exp(O(\sqrt{N}))$-time quantum algorithms, as in the 1D case.

\paragraph{Classical lower bounds.}
We also prove classical lower bounds that establish a strict
superpolynomial quantum-classical separation.  Our proof strategy is
similar to that of~\cite{CCDFGS03} but with some additional
difficulties.  In the original proof, it was argued that the classical
algorithm had negligible probability of ever encountering a cycle and
thus could not distinguish the explored region from a tree.  In our
case, the classical algorithm can certainly find small cycles in the
region of the graph with small supernodes. However, we can argue that
these cycles do not provide any useful information about the large
supernodes that are far from these small regions.

\begin{theorem}[Informal, see \cor{1D-LB} and \cor{GFF-LB}]
Any classical algorithm with only query access to the graph requires
exponentially many queries to traverse the graph models we consider
here, namely the 1D hierarchical graphs generalizing welded trees in
\sec{random-graph} and the $d$-dimensional graphs based on the
biased Gaussian free field in \sec{highdim}.
\end{theorem}

\paragraph{Concrete realizations of hierarchical graphs.}
Our quantum algorithms and classical lower bounds apply to graphs with a certain hierarchical structure, but it is still necessary to provide explicit constructions of such graphs.  Beginning with generalizations of the welded
tree, we describe a broad family of 1D hierarchical graphs in
\sec{construction-1D}. These ideas are extended to higher-dimensional
lattices in \sec{construction-highD}, with a construction that makes
use of statistical mechanics models.

\paragraph{Hierarchical graphs from sparsification.}
Lattices are convenient because there are established techniques in
the literature for analyzing their localization properties.  However,
our results also apply to completely general hierarchical graphs.  In \sec{construction-sparse}, we describe an even more general
framework for constructing hierarchical graphs whose supergraphs can
be chosen arbitrarily.  This approach is based on graph
sparsification.

This relaxes yet another rigid constraint present in the original welded trees speedup.  In the random graph ensemble for the welded trees, each
instance exactly obeyed a certain ``balanced'' property which
guarantees that the Krylov subspace will be small.  Our sparsified
ensemble in \sec{construction-sparse} obeys the balanced property
on average but individual instances will be far from balanced.
Despite this, individual instances still have a small subspace that remains
approximately invariant.  This shows that the symmetry behind the
quantum walk's success need only apply to the ensemble of graphs rather than individual instances.

Using these techniques, we can construct hierarchical graphs with whatever input supergraph we please. The classical lower bounds still hold under fairly general conditions.  Although the analysis of the quantum runtime depends on the particular choice of supergraph, we expect superpolynomial or exponential speedups whenever these supergraphs have delocalized eigenstates.

\paragraph{Variable edge density within supervertices.}

In the hierarchical graphs that we have thus far described, edges only connect vertices which are in neighboring supervertices.  This property allowed the effective Hamiltonian dynamics to map onto a random hopping/bond model, where there is no onsite (or diagonal) disorder.  In \sec{Anderson}, we consider hierarchical graphs where edges can also connect between vertices in the same supervertex.  The effective Hamiltonian dynamics of this system has onsite disorder as well as bond disorder.

Assuming no bond disorder (i.e.~the hopping amplitudes are all the same), such an effective Hamiltonian resembles an Anderson model.  If the diagonal elements are constant, we can find generalizations of the results we derived (\thm{onsite-superpolynomial}) and the superpolynomial speedup remains.  However, for a random onsite potential, the eigenstates Anderson localize with exponentially decaying tails, thus prevent a superpolynomial speedup.  An example of such a graph ensemble is illustrated in \thm{Andersonspeedup}.  This indicates that it is still possible to construct particular classes of hierarchical graphs where localization can completely inhibit the performance of the quantum walk.

However, this result is only specific to one and two dimensions.  A popular conjecture argued in the seminal paper by Abrahams et al.~\cite{AALR} states that for a Hamiltonian with nearest neighbor hopping and random onsite potential in a $d$-dimensional lattice, a nonzero fraction of eigenstates are delocalized when $d \geq 3$. Assuming this, as long as $G$ is a $d$-dimensional hypercubic lattice and the effective Hamiltonian is constrained to have random onsite potential, the quantum walk can easily traverse the graph and we obtain a rich class of random graphs that exhibit an exponential quantum-classical separation (\cor{Anderson-transition}).  An interesting open question is the case when there is competing bond and onsite disorder; general results such as Furstenberg's theorem~\cite{furstenberg1963noncommuting} may be useful in understanding whether localization still persists in low dimensions, and we leave this question to future works.

\paragraph{Summary.}
Our results are summarized in the table below.  Results labelled with a $(\ast)$ are based on well-established conjectures, and $c_{\text{cl}}$, $c_{\text{q}}$ are constants whose values do not concern us.  The first row is related to the welded tree speedup of \cite{CCDFGS03}, with a minor difference being that their construction required perfect binary trees while ours only requires random connectivity between supervertices.

\begin{center}
\begin{tabular}{c c c c}
 \toprule
 Supergraph dimension & Disorder & Classical Runtime & Quantum Runtime \\ [0.5ex]
\midrule
 1 & None & $\exp(c_{\text{cl}} N)$ & $\poly(N)$ \\
 1 & Off-diagonal & $\exp(c_{\text{cl}} N)$ & $\exp(c_{\text{q}} \sqrt{N})$ \\
 1 & Diagonal & $\exp(c_{\text{cl}} N)$ & $\exp(c_{\text{q}} N)$ \\
 2 & Off-diagonal (BGFF) & $\exp(c_{\text{cl}} N)$ & $\poly(N)$ \\
 2 & Diagonal & $\exp(c_{\text{cl}} N)$ & $\exp(c_{\text{q}} N)$ $(\ast)$ \\
 $3 \leq d \leq \log N$ & Off-diagonal (BGFF) &
 $\exp(c_{\text{cl}} N d)$ & $\poly(N^d) \exp(c_{\text{q}} \sqrt{d\, 2^d \log N})$ \\
 $\log N < d \leq N$ &
 Off-diagonal (BGFF) &
 $\exp(c_{\text{cl}} N d)$ & $\poly(N^d) \exp(c_{\text{q}} d)$ \\
 $d \geq 3$ & Diagonal &
 $\exp(c_{\text{cl}} N d)$ & $\poly(N^d)$ $(\ast)$ \\
\bottomrule
\end{tabular}
\end{center}

\paragraph{Related work.}
The welded tree has various applications to other quantum algorithms. For instance, Rosmanis~\cite{rosmanis2011snake} (and more recently by Childs, Coudron, and Gilani~\cite{childs2022forgetting}) studied the difference between hitting the exit from the entrance and finding an explicit path between the entrance and exit on the welded tree. Furthemore, Krovi and Brun~\cite{krovi2007quantum} studied quantum walks with an associated group of symmetries, and proved that such quantum walks confined to the subspace corresponding to this symmetry group can be seen as an effective quantum walk on a smaller quotient graph. Ben-David et al.~\cite{BCG+20} used the welded tree as a gadget to construct a property testing problem in the adjacency list model for bounded-degree graphs and showed an exponential quantum-classical separation in query complexity. Hastings, Gily{\'e}n, and Vazirani~\cite{GSV21} recently proved subexponential oracle separation between classical algorithms and adiabatic quantum computation with no sign problem, i.e., restricted to adiabatic paths of stoquastic Hamiltonians. Experimentally, Ref.~\cite{shi2020photonic} implemented the welded tree on quantum photonic chips with tree depth at most 16 and branching rate from 2 to 5.

\paragraph{Conclusions.}
Our paper concludes in \sec{concl}  with a discussion of open problems including the potential for our work applying to optimization and sampling problems.

\section{Random hierarchical graphs}\label{sec:graph-model}

In this section, we introduce the definitions of the random graph ensembles and explain how they can be generated.  The ensemble hierarchical graphs that we construct sets up the quantum walk problem.  We then define the oracle assumption that we operate under, and discuss the results relating to the performance of the algorithm.  We then provide a lower bound analysis on the query complexity of classical algorithms.

We start with the construction of a hierarchical graph whose supergraph is a 1D chain.
\begin{definition}[Hierarchical graph on supergraph $G$]
  A hierarchical graph on a supergraph $G = (V,E)$ is defined by a set of nodes $S_v$ for each $v \in V$ and a set of edges $\mathcal{E}_{v,w}$ for each $(v,w) \in E$ such that $s_v \triangleq |S_v|$
and $e_{vw} \triangleq |\cE_{v,w}|$.
There are two special start and exit nodes $v_{\text{init}}$ and $v_{\text{exit}}$  with
$s_{\text{init}} = s_{\text{exit}} = 1$,
Define $\cV \triangleq \cup_{v\in V} S_v$, $\cE \triangleq \cup_{(v,w)\in E}\cE_{v,w}$, and $\cG = (\cV,\cE)$.
For each $(v,w)  \in E$, the edge set $\mathcal{E}_{v,w}$ denotes the set of edges between nodes in $S_v$ and $S_w$.
\end{definition}

Next, we restrict the above definition so that it satisfies certain symmetry constraints that allows the quantum algorithm to efficiently explore the graph:
\begin{definition}[Balanced hierarchical graph]\label{def:balancedhier}
  A hierarchical graph on supergraph $G = (V,E)$ is said to be balanced if for every $(u,v) \equiv e \in E$, the number of edges connecting a fixed node $\alpha \in S_u$ to nodes in $S_v$ is the same for all $\alpha$.
\end{definition}

In order to confuse the classical algorithm, we require that the above graph be $D$-regular.  This avoids the degree of a vertex leaking information about its position in the graph to the classical algorithm.  (It is possible to have different degree for
the entrance and exit vertices, or even a small neighborhood around them, and indeed this occurs in the original welded tree paper.)
The trivial observation follows:
\begin{proposition}
For a $D$-regular balanced hierarchical graph on supergraph $G = (V,E)$, the condition $\sum_{v \sim u}e_{uv} = D s_u$ holds, where $v \sim u$ denotes the neighbors of $u$ in supergraph $G$.
\end{proposition}
\begin{proof}
Follows from the structure and the $D$-regular nature of a hierarchical graph.
\end{proof}

\begin{remark}
We note that it is possible to extend our random graph ensembles so that they are not perfectly regular.  In particular, in the above, we assume that the degree of each node is $D$ in our random graph ensemble.  However, as long as the left and right degrees of all nodes in a given supervertex are the same, we are free to choose different degrees $D_i$ for all of the nodes in a fixed supervertex $S_i$, as long as the $D_i$ are sampled independently and randomly from degree distribution $\Delta$.  The sampling assumption is important, because arbitrarily changing degrees of certain vertices could leak information to the classical algorithm about where it is in the graph.  The sampling assumption however turns out not to give the classical algorithm any advantage over the case when all of the degrees are fixed to $D$.
\end{remark}

\subsection{Oracle model and supervertex subspace}

Now we will discuss the input model.  We assume that the graph is accessible implicitly, meaning that given a vertex $v$ we can efficiently list its neighbors.  Since vertices are not naturally specified by bit strings, we assume there is an encoding function $\Enc$ that maps vertices to bit strings.  For simplicity we assume that the graph is $D$-regular.
\begin{definition}[Oracle model]\label{def:oracle}
  Let $f_\cG\colon \cV\times [D] \to \cV$ map $x\in \cV$ and $s\in [D]$ to the $s^{\text{th}}$ neighbor of $x$ in $\cG$.  Our oracle takes as input $(\Enc(x), s)$ and outputs $\Enc(f_\cG(x,s))$.  Upon input $(a,s)$ with $a$ not equal to any $\Enc(x)$, the oracle simply returns INVALID.
  \end{definition}
  As is usual for oracle algorithms, we could simply assume access to such an oracle, or come up with a subroutine that supplies the oracle.  Our paper will focus on the question of what is possible given access to the oracle.

  \begin{remark}\label{rem:encoding}
    As in \cite{CCDFGS03} we usually take $\Enc$ to be a random injective function from $V$ to a much larger set (say of size $>|V|^2$) to force any algorithm using the oracle to explore the graph by following paths from the start vertex and not jumping to random vertices.    This is mostly for convenience.  If we instead took $\Enc$ to be a random map to a set of size $|V|$ then the classical algorithms would gain only the ability to jump to a random vertex in $V$.  For most graphs we consider this would not significantly improve their traversability by classical algorithms.
  \end{remark}

  Next, we show that the hierarchical structure that the graphs described above have implies that the quantum walk operates within a much smaller subspace of the full Hilbert space which we call the \emph{supervertex subspace}.  The proposition below holds for any balanced hierarchical graph whose supergraph is not necessarily a line graph.
\begin{proposition}
  Let $\cG$ be a balanced hierarchical graph with supergraph $G = (V,E)$ and adjacency matrix $A$.
  Define the ``supervertex states''
\begin{equation}
    \ket{S_u} = \frac{1}{\sqrt{s_u}} \sum_{\alpha \in S_u} \ket{\alpha},
  \end{equation}
  and define the supervertex subspace to be the span of $\{\ket{S_u} : u \in V\}$.   The supervertex subspace is an invariant subspace of $A$.  As a corollary, if $\ket\psi$ is in the supervertex subspace and we evolve according to a Hamiltonian $H = - A$, then $e^{-iHt}\ket\psi$ is also in the supervertex subspace for all $t$.
\end{proposition}

In other words, the Krylov subspace mentioned in the introduction is contained within the supervertex subspace.

\begin{proof}
Because the hierarchical graphs are balanced, a given state $\ket{S_u}$ is transformed by $A$ to
\begin{equation}
    A \ket{S_u} = \sum_{v \sim u} \frac{e_{uv}}{\sqrt{s_u s_v}} \ket{S_v}.
  \end{equation}
Writing $e^{-iHt}$ as a power series, we see that it also maps any supervertex state to another element of the supervertex subspace.
\end{proof}

  One can think of the balanced assumption as a symmetry that treats all nodes in $S_u$ on an equal footing.  This symmetry results in the time evolution acting within a much smaller subspace.  Since we may work only in the supervertex space, it makes sense to construct an effective Hamiltonian in the supervertex basis.
\begin{lemma}[Hamiltonian of hierarchical graph]
    The ``graph Hamiltonian'' $H$ is defined to be negative of the adjacency matrix of a balanced hierarchical graph.  When $H$ is projected onto the supervertex subspace $\text{span}\left\{\ket{S_u}, u \in V\right\}$, the matrix elements of the Hamiltonian are:
\begin{equation}
    -\mel{S_u}{H}{S_v} = \frac{e_{uv}}{\sqrt{s_u s_v}}, \hspace{0.5cm} u \sim v
\end{equation}
and $\mel{S_u}{H}{S_v} = 0$ otherwise.
\end{lemma}
\begin{proof}
This follows from the action of $A \ket{S_u}$ computed in the previous proposition.
\end{proof}

\subsection{One-dimensional random graph model}

Next, we will construct a random graph ensemble for the case in which $G$ is a one-dimensional line graph with nodes $v_0, v_1, v_2, \ldots, v_{2n}$.  First, we define the following quantities:
\begin{definition}[Edge-edge ratio]
A hierarchical graph on the line supergraph $G = (V,E)$ which has nodes $0, 1, 2, \ldots, 2n$ possesses edge-edge ratios $r_1, r_2, r_3, \ldots, r_{2n-1}$ with:
\begin{equation}
    r_k = \frac{e_{k, k+1}}{e_{k-1, k}} \triangleq \frac{e_{k+1}}{e_k}.
\end{equation}
\end{definition}
\begin{definition}[Edge-vertex ratio]
A hierarchical graph on the line supergraph $G = (V,E)$ which has nodes $0, 1, 2, \ldots, 2n$ possesses edge-vertex ratios $\kappa_1, \kappa_2, \kappa_3, \ldots, \kappa_{2n}$ with:
\begin{equation}
    \kappa_k = \frac{e_{k-1, k}}{s_k} \triangleq \frac{e_k}{s_k}.
\end{equation}
\end{definition}
Both of these definitions quantify the expansion properties of the graph.  In turn, this allows us to define a random graph ensemble based on the edge-edge or edge-vertex ratios:
\begin{definition}[Random balanced hierarchical graph] \label{def:ratios}
A random balanced hierarchical graph on a line supergraph $G = (V,E)$ satisfies one of the following conditions:
\begin{itemize}
    \item Edge-edge ratios $r_i \sim D_i$ where $D_i$ are independent and bounded, and the distributions for $\log r_i$ are independent and bounded; furthermore, $\mathbb{E}_{D_i}[\log^c r_i] \leq \infty$ and $\mathbb{E}_{D_i}[r_i^c] \leq \infty$ for all positive integers $c$.
    \item Edge-vertex ratios $\kappa_i \sim F_i$ where $F_i$ are independent and bounded, and the distributions for $\log \kappa_i$ are independent and bounded; furthermore, $\mathbb{E}_{F_i}[\log^c r_i] \leq \infty$ and $\mathbb{E}_{F_i}[r_i^c] \leq \infty$ for all positive integers $c$.
\end{itemize}
\end{definition}
The assumptions on boundedness of certain expectation values will be useful for various details in the proofs.  In addition to the definition above, we need to assume that the random graph model has an intrinsic bias: on average the number of nodes in each supernode grows exponentially as a function of depth until the halfway point, after which it decays exponentially with depth.  The reason for this assumption is to show that no classical algorithm can penetrate this graph due to the existence of ``large supervertices'' with an exponential number of vertices.

\begin{assumption}
A random balanced hierarchical graph on a line supergraph $G = (V,E)$ is said to have a bias when the edge-edge ratios $r_i > 1$ over a depth of $O(|V|)$ and $r_i < 1$ in the remaining part of the supergraph.
\end{assumption}

The edge-edge and edge-vertex ratios can be related to one another by the following proposition:
\begin{proposition}
A $D$-regular random balanced hierarchical graph on line supergraph $G = (V,E)$ with independently random edge-edge ratios has independently random edge-vertex ratios related by
\begin{equation}
    \kappa_i = \frac{D}{1 + r_i}.
\end{equation}
\end{proposition}
\begin{proof}
We know that $e_i = \kappa_i s_i$ and $e_i + e_{i+1} = D s_i$.  Furthermore, from the definition of the edge-edge ratio, $e_{i+1} = r_i e_{i}$.  Combining these two expressions, we find the desired proposition.  Because $\kappa_i$ is only a function of $r_i$, independence of one distribution implies independence of the other distribution and vice versa.  Furthermore, as $r_i$ are bounded and positive, $\kappa_i$ are bounded and positive.
\end{proof}
Thus, it suffices to restrict our attention to one of these quantities; we will be focusing on edge-edge ratios for the rest of the paper.  Finally, we explicitly write down the Hamiltonian of the 1D hierarchical supergraph when restricted to the supervertex subspace.
\begin{proposition}
If the supergraph $G = (V,E)$ is a line graph of size $|V| = 2n+1$, the matrix elements of the Hamiltonian restricted to the supervertex space is equal to
\begin{equation}\label{eqn:hamiltonian}
    \mel{S_i}{H}{S_j} = \begin{cases}
t_i \triangleq \frac{e_i}{\sqrt{s_i s_{i+1}}}, & \text{for } j=i+1\\
t_{i-1} = \frac{e_{i-1}}{\sqrt{s_{i-1} s_{i}}}, & \text{for } j=i-1\\
0, & \text{otherwise }\\
\end{cases}\quad\ \ i,j\in\{0,\ldots,2n\}.
\end{equation}
\end{proposition}

This Hamiltonian is a tridiagonal matrix with zeroes along the principal diagonal.  This constraint of a zero principal diagonal will be essential to argue that quantum walks can efficiently traverse these graphs.

Finally, we remark that although we have described the general structure of hierarchical graphs, we have not provided an explicit algorithm which can generate these graphs.  A few ways of explicitly constructing these graphs are discussed in \sec{metropolis}.

\subsection{Higher-dimensional random graph model}\label{sec:high-dim}

We also discuss the case where the supergraph $G$ is not one-dimensional, but rather a higher-dimensional regular lattice.  The analysis when $G$ is a two dimensional grid is significantly more difficult for reasons we will elaborate on later, so we instead consider the case where $G$ is a Lieb graph, defined below:
\begin{definition}[Lieb graph of graph $G$]
A Lieb graph of graph $G$ is defined to be a new graph $G' = (E', V')$ where the vertices $V' = V \cup \Delta$ where $|\Delta_1| = |E|$ is a new set of vertices such that each vertex in $\Delta_1$ is placed on one edge on graph $G$.  The size of the new edge set for $G'$ satisfies $|E'| = 2|E|$.
\end{definition}
We also provide a definition of the Lieb lattice.
\begin{definition}[$d$-dimensional Lieb lattice]
  Consider a $d$-dimensional cubic lattice $L$ with side length $N$ and vertices labeled  by $[N]^d$.    A $d$-dimensional Lieb lattice is a Lieb graph of $L$, meaning that it has vertices labeled by $[N]^d$ as well as a vertex corresponding to each edge of $L$.   The vertices corresponding to edges of $L$ are labeled by half-integer coordinates.  Specifically the vertex on the edge $((x_1,\ldots,x_d), (x_1,\ldots,x_i+1,\ldots, x_d))$ in the original graph is labeled by $(x_1,\ldots,x_i+\frac 12,\ldots,x_d)$ in the Lieb lattice.
\end{definition}
Following the analysis in the 1D case, we also must define the equivalent of edge-edge ratios in the hierarchical Lieb lattice (we will not work with edge-vertex ratios for the remainder of this paper):
\begin{definition}[Edge-edge ratio on Lieb lattice]\label{def:EELieb}
Consider a hierarchical graph $\mathcal{G}$ whose supergraph is a $d$-dimensional Lieb lattice.  Using the notation where vertices $v$ correspond to a tuple of Cartesian coordinates, we define the edge-edge ratios to be
\begin{equation}
    r^{(i)}_{(x_1, x_2, \ldots, x_d)} = \frac{e_{(x_1, x_2, \ldots, x_i + 1/2, \ldots, x_d), (x_1, x_2, \ldots, x_i + 1, \ldots, x_d)}}{e_{(x_1, x_2, \ldots, x_i, \ldots, x_d), (x_1, x_2, \ldots, x_i + 1/2, \ldots, x_d)}}
\end{equation}
where $x_k \in [N]$ and $i \in [d]$ denotes the direction along the lattice the edge ratio is measured.
\end{definition}
For reasons which we explain later, it is also convenient to introduce what we call ``height fields''.  Though their definition is unmotivated currently, it is required in our construction of the random graph ensemble:
\begin{definition}[Height fields]\label{def:heightfields}
Define a height field $\Phi_{v} \in \mathbb{R}$ for each site of a $d$-dimensional Lieb lattice.  The difference of neighboring height fields are related to the edge-edge ratios by
\begin{equation}\label{eqn:heighttorat}
    \Phi_{(x_1, x_2, \ldots, x_i + 1, \ldots, x_d)} - \Phi_{(x_1, x_2, \ldots, x_i, \ldots, x_d)} = \log r^{(i)}_{(x_1, x_2, \ldots, x_i, \ldots, x_d)}.
\end{equation}

Alternatively, we may repackage the set of height fields $\varphi$ into $d+1$ height fields $(\varphi, \chi^{(1)}, \ldots, \chi^{(d)})$, where $\varphi$ corresponds to $\Phi_v$ where $v$ has integer valued coordinates, and $\chi^{(i)}$ for $i \geq 1$ corresponds to the set $\varphi_v$ where $v$ has coordinates
\begin{equation}
    v = \left(x_1, x_2, \ldots, x_{i} + \frac{1}{2}, \ldots, x_d\right)
\end{equation}
with $x_i \in \mathbb{Z}$.  Each of the $d+1$ height fields live on the vertices of a $d$-dimensional cubic sublattice.
\end{definition}

A simple example with the edge-edge ratios and height fields explicitly labeled for $d = 2$ is shown below:

\begin{equation}
\centering
\includegraphics[scale=0.5]{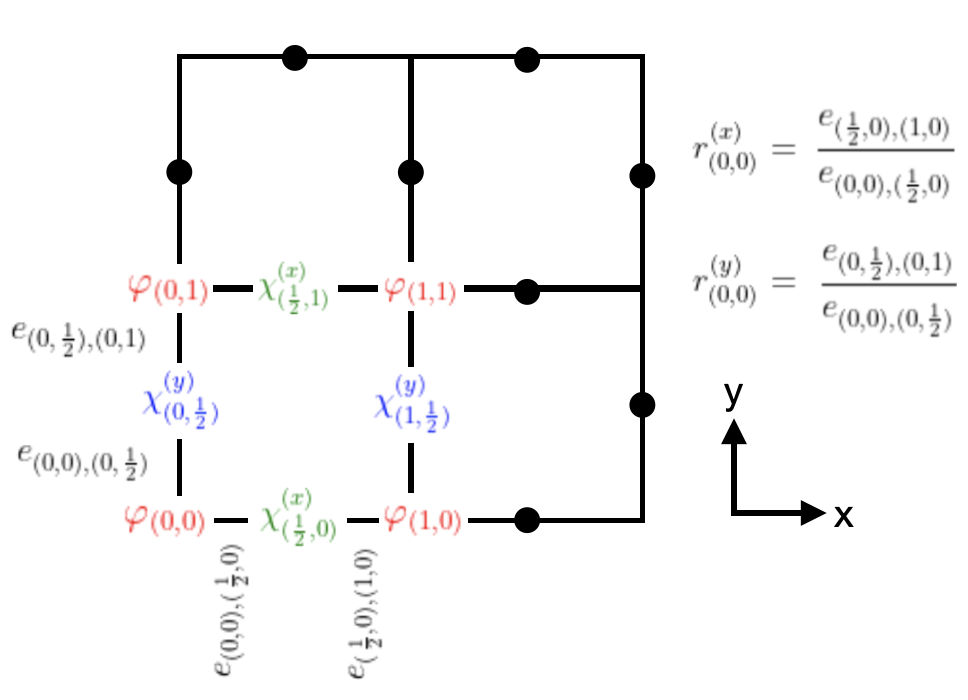}
\end{equation}

In this example, the height field $\varphi_{(1,1)}$ is given by
\begin{equation}
\varphi_{(1,1)} = \varphi_{(0,0)} + \log r_{(0,0)}^{(y)} + \log r_{(0,1)}^{(x)}
\end{equation}
and is also equivalent to
\begin{equation}
\varphi_{(1,1)} = \varphi_{(0,0)} + \log r_{(0,0)}^{(x)} + \log r_{(1,0)}^{(y)}.
\end{equation}
This equivalence follows from the condition that $\sum \Delta \varphi = 0$ around any loop.  For the $\chi$ fields, we have
\begin{equation}
\chi_{(\frac{1}{2},1)}^{(x)} = \chi_{(0,\frac{1}{2})}^{(y)} + \log \frac{e_{(0,1), (\frac{1}{2},1)}}{e_{(0,\frac{1}{2}), (0,1)}}
\end{equation}
and so on.  With the height fields defined, we now can define the random graph ensemble that we will require for the higher-dimensional graphs:
\begin{definition}[Random hierarchical graph on Lieb lattice supergraph]\label{def:liebensemble}

A random balanced hierarchical graph on a $d$-dimensional Lieb lattice supergraph is constructed in terms of a probability distribution $p(\varphi, \chi^{(1)}, \ldots, \chi^{(d)})$ defined over the $d+1$ height fields, which we may alternatively denote by $\Phi$.  This provides a joint distribution for the edge-edge ratios $r^{(i)}_v$.
\end{definition}

The choice of the probability distribution $p(\cdot)$ dictates the behavior of the landscape.  As in the 1D case that led to the welded trees, it is natural to have vertex size vary slowly between adjacent vertices, but to vary significantly once we look across the entire graph.  To model this behavior, we analyze a natural generalization of the 1D random walk called the Biased Gaussian free field (BGFF).  In the BGFF the differences in the height field $\varphi$ will have independent Gaussian fluctuations around some specified mean, which is called the ``bias''.

\begin{definition}[Biased Gaussian free field (BGFF)]\label{def:BGFF}
Let $G=(V,E)$ be a directed graph and let $J\colon E\rightarrow\R$. For a function $\Phi\colon V\rightarrow \R$,
define the ``action functional''
\begin{equation}
    S[\Phi] = \frac{1}{2}\sum_{(v,w)\in E} (\Phi_{v} - \Phi_{w} - J_{v,w})^2.
\end{equation}
We call  $J$ the source term, or bias.  The Gaussian free field is generated by the probability distribution $p[\Phi] = \frac{1}{Z} \exp\left(-S[\Phi]\right)$, under the measure
\begin{equation}
    \int \mathcal{D}\Phi \triangleq \int \prod_{v}d\Phi_v.
\end{equation}
\end{definition}

 Our results generally hold for a larger range of probability distributions than a GFF over the height fields.  In particular, as the edge-edge ratios are constrained to be rational integers, we must allow for distributions which are close to the biased Gaussian free field.  Therefore, we define a biased sub-Gaussian free field, which we denote as BsGFF:

\begin{definition}[Biased sub-Gaussian free field (BsGFF)]\label{def:BsGFF}
Let $G=(V,E)$ be a directed graph and let $J\colon E\rightarrow\R$. For a function $\Phi\colon V\rightarrow \R$,
define the action functional
\begin{equation}
    S[\Phi] = \frac{1}{2}\sum_{(v,w)\in E} (\Phi_{v} - \Phi_{w} - J_{v,w})^2.
\end{equation}
and the BGFF $p_J(\Phi)$ from $S[\Phi]$.  A distribution $q_J(\Phi)$ is said to be a BsGFF if there exists a constant $\lambda > 0$ such that for all $t$
\begin{equation}
    \mathbb{E}_{q_J}\left[\exp\left((\Phi_v - \Phi_u) t\right)\right] \leq \mathbb{E}_{p_J}\left[\exp\left((\Phi_v - \Phi_u) \lambda t\right)\right].
\end{equation}
We refer to the BGFF $p_J(\Phi)$ as the parent BGFF of $q_J(\Phi)$.
\end{definition}

In the following sections we will consider the case where the fields $\varphi$ and $\chi^{(i)}$ obey independent BsGFF distributions.  Each of these $d+1$ fields is defined over a cubic lattice, and together these $d+1$ cubic lattices comprise the Lieb lattice.

As in the 1D case the fluctuations will determine the quantum runtime while the shape of the lattice, which is dominated by the source/bias term, will determine the classical runtime.  This again will enable superpolyomial or exponential speedups.

\subsection{Quantum walk traversal time}

Fundamental to the results of the paper is a new bound on the traversal time of a quantum walk from a start node $S_\text{init}$ to an exit node $S_\text{exit}$.  \emph{We first make the assumption that the Hamiltonian we time evolve has a zero energy eigenstate}.  In particular, there exists a $\ket{\psi}$ such that $H \ket{\psi} = 0$, which we will henceforth label as $\ket{0}$.  The theorem below shows that the traversal time of a quantum walk is related both to the overlap of the initial and final state with the zero energy eigenstate as well as the gap to the smallest nonzero eigenvalue in magnitude.

\begin{lemma} \label{lem:evolution}
  Denote the gap of the Hamiltonian $H$ in the supervertex subspace by $\Delta$, which is defined to be the smallest \emph{non-zero} eigenvalue of $H$ in absolute value.
  Consider the following algorithm
  \bit
\item Start with the state $\ket{S_{\text{init}}}$.
\item Evolve for a random time $t$ drawn uniformly between $[0,\tau]$.  Here $\tau = \frac{4}{\Delta \left|\braket{S_{\text{exit}}}{0} \braket{0}{S_{\text{init}}}\right|}$.
  \item Measure the state in the vertex basis.
    \eit
    The probability of reaching $S_{\text{exit}}$, denoted $\mathbb{P}(\text{EXIT})$, satisfies
\begin{equation}
    \mathbb{P}(\text{EXIT}) \geq \frac{1}{4}\left|\braket{S_{\text{exit}}}{0} \braket{0}{S_{\text{init}}}\right|^2.
\end{equation}
\end{lemma}
\begin{proof}
The quantity of interest we want to calculate is a probability of measuring $\ket{S_{\text{exit}}}$ starting from $\ket{S_{\text{init}}}$ and evolving for time $t$, averaged between $[0,\tau]$:
\begin{align}
  \mathbb{P}(\text{EXIT}) &= \frac{1}{\tau} \int_{0}^{\tau} \dd t\,\left|\mel{S_{\text{exit}}}{e^{-iHt}}{S_{\text{init}}}\right|^2 \nonumber \\
                          & \geq \left(\frac{1}{\tau} \int_0^\tau \dd t\,\left|\mel{S_{\text{exit}}}{e^{-iHt}}{S_{\text{init}}}\right|\right)^2
                          & \text{(Cauchy-Schwarz)}
                            \nonumber \\
                          & \geq \frac{1}{\tau^2} \left|\int_0^\tau \dd t\,\mel{S_{\text{exit}}}{e^{-iHt}}{S_{\text{init}}}\right|^2.
                           & \text{(triangle inequality)}
\end{align}
Next, we decompose the time evolution operator in terms of its energy eigenbasis in order to write
\begin{align}
  \frac{1}{\tau}\int_0^\tau \dd t\,\mel{S_{\text{exit}}}{e^{-iHt}}{S_{\text{init}}}
  &= \frac{1}{\tau}\int_0^\tau \dd t\, \sum_E e^{-iEt} \braket{S_{\text{exit}}}{E}\braket{E}{S_{\text{init}}}\nonumber \\
  &= \braket{S_{\text{exit}}}{0} \braket{0}{S_{\text{init}}} + \sum_{E \neq 0} \frac{1 - e^{-iE\tau}}{i E \tau} \braket{S_{\text{exit}}}{E}\braket{E}{S_{\text{init}}},
    \label{eq:zero-mode-overlap}
\end{align}
where in the second line, we have specifically assumed that the Hamiltonian has a unique zero energy eigenstate.  For the one-dimensional chain, we have already shown uniqueness of the zero energy eigenstate for odd-length chains.  More generally we will later consider the case of a degenerate zero-energy subspace.  In this case, we would replace the $\braket{S_{\text{exit}}}{0} \braket{0}{S_{\text{init}}}$ term with $\bra{S_{\text{exit}}} \Pi \ket{S_{\text{init}}}$ where $\Pi$ projects only the zero-energy subspace.  To make the analysis tractable in the higher-dimensional case, we will impose constraints on the graph that guarantee that there is a single state (which we call $\ket{0}$) that is within the zero-energy subspace and has overlap with $\ket{S_{\text{init}}}$ and $\ket{S_{\text{exit}}}$.  This is discussed in more detail in \sec{2D}.  The result is that even in the case of a degenerate zero-energy subspace, we can consider only the single zero-energy state having overlap with $\ket{S_{\text{init}}}$ and $\ket{S_{\text{exit}}}$, and we can use Eqn.~\eq{zero-mode-overlap}.

Continuing, we calculate
\begin{align}
  &\left|\frac{1}{\tau}\int_0^\tau \dd t\,\mel{S_{\text{exit}}}{e^{-iHt}}{S_{\text{init}}}\right|\nonumber \\
  &\geq \left|\braket{S_{\text{exit}}}{0} \braket{0}{S_{\text{init}}}\right| - \left|\sum_{E \neq 0} \frac{1 - e^{-iE\tau}}{i E \tau} \braket{S_{\text{exit}}}{E}\braket{E}{S_{\text{init}}}\right|
        & \text{(triangle inequality)}\nonumber \\
  & \geq \left|\braket{S_{\text{exit}}}{0} \braket{0}{S_{\text{init}}}\right| - \frac{2}{\Delta \tau}\sum_{E \neq 0} \left|\braket{S_{\text{exit}}}{E}\braket{E}{S_{\text{init}}}\right|
        &\big(\Delta \triangleq \min_{E_i \neq 0} |E_i|\big)
          \nonumber \\
  & \geq \left|\braket{S_{\text{exit}}}{0} \braket{0}{S_{\text{init}}}\right| - \frac{2}{\Delta \tau}
    \sqrt{\sum_{E \neq 0} \abs{\braket{S_{\text{exit}}}{E}}^2 \sum_{E\neq 0}\abs{\braket{E}{S_{\text{init}}}}^2}
        &\text{(Cauchy-Schwarz)}   \nonumber \\
  & \geq \left|\braket{S_{\text{exit}}}{0} \braket{0}{S_{\text{init}}}\right| - \frac{2}{\Delta \tau}.
\end{align}
Selecting $\tau = \frac{4}{\Delta \left|\braket{S_{\text{exit}}}{0} \braket{0}{S_{\text{init}}}\right|}$, we find that
\begin{equation}
    \mathbb{P}(\text{EXIT}) \geq \frac{1}{4}\left|\braket{S_{\text{exit}}}{0} \braket{0}{S_{\text{init}}}\right|^2.
\end{equation}
\end{proof}
The result above will fail in general if there are multiple zero modes, which follows from the fact that the zero modes may destructively interfere with one another.  In the multidimensional case, we will utilize particular properties of the Lieb lattice that help us avoid this situation.

\section{One-dimensional random graph ensembles}
\label{sec:random-graph}

In this section we study the zero energy eigenstate (which we also refer to as the zero mode) of hopping Hamiltonians on a one-dimensional line.  Following Eqn.~(\ref{eqn:hamiltonian}), we work with the effective Hamiltonian
\begin{equation}
    H = \begin{pmatrix}
    0 & t_{0}   & \cdots     & 0  & 0 \\
    t_{0}  & 0       & t_{1}   & \cdots     & 0   \\
     \vdots &       t_{1}  & \ddots         & t_{2n-2}    & \vdots    \\
     0 &        \vdots &          t_{2n-2} &     0      & t_{2n-1}   \\
     0 &      0   &        \cdots   &       t_{2n-1}    & 0
  \end{pmatrix},
\end{equation}
This is a tridiagonal matrix with zeroes on the principal diagonal, but random off-diagonal elements $t_0,\ldots,t_{2n-1}$.  In the condensed matter literature, this is known as a random hopping model.  This model has the special property that it possesses an exact zero energy eigenstate.

\begin{lemma}\label{lem:zero-vector}
  There exists an eigenvector of $H$ with zero eigenvalue within the
  subspace spanned by $\ket{S_0},\ldots,\ket{S_{2n}}$.
\end{lemma}

\begin{proof}
 The number of supervertices, $2n+1$, is odd. The recurrence relation for the determinants of principal submatrices of $H$ (denoted by $h_i$) satisfies $h_i = -t_{i-1}^2 h_{i-2}$, where $t_i = \frac{e_i}{\sqrt{s_i s_{i+1}}}$.  The boundary conditions satisfy $h_0 = 1$ and $h_{2n} = 0$, and the determinant is zero when $n$ is odd.
\end{proof}
Next, we may explicitly write down the zero-energy state and analyze its decay properties.  This state ensures a subexponential time quantum algorithm, as the runtime is linked to the decay properties of the zero mode:

\begin{lemma}\label{lem:prob-entrance-exit}
The vector in \lem{zero-vector} can be written as $\sum_{i=0}^{2n}
\psi_i \ket{S_i}$, and the coefficients $\psi_{i}$ satisfy
\begin{equation}
  |\psi_{2n}||\psi_0| \geq \frac{C}{n} \cdot
  \exp\left[\frac{1}{2}\sum_{k=0}^{n-1}  \log
    \frac{r_{2k}}{r_{2k+1}}-\max_j\left(\sum_{k=0}^{j-1}  \log
      \frac{r_{2k}}{r_{2k+1}}\right)\right],
   \label{eqn:zero-vec-bounds}
\end{equation}
where $r_{i}:=e_{i}/e_{i-1}$.  Assuming that $r_i$ is bounded,  $C$ is a constant.
\end{lemma}

\begin{proof}
The eigenvalue equation for the zero energy eigenstate can be written down explicitly.  For even indexed sites, the equation is
\begin{equation}
    t_{i-1} \psi_{i-1} + t_{i} \psi_{i+1} = 0,
\end{equation}
while for the odd indexed sites the eigenstate has zero amplitude.  From the eigenvalue equation, we may write
\begin{equation}
    \psi_{2j} = \left(\prod_{k=0}^{j-1} -\frac{t_{2k}}{t_{2k+1}}\right) \psi_{0},
\end{equation}
and we may use the identity
\begin{equation}
  \prod_{k=0}^{j-1} \frac{t_{2k}}{t_{2k+1}}
  = \frac{\frac{e_0}{\sqrt{s_0 s_{1}}} \frac{e_2}{\sqrt{s_2 s_{3}}} \cdots \frac{e_{2j-2}}{\sqrt{s_{2j-2} s_{2j-1}}}}{\frac{e_1}{\sqrt{s_1 s_{2}}} \frac{e_3}{\sqrt{s_3 s_{4}}} \cdots \frac{e_{2j-1}}{\sqrt{s_{2j-1} s_{2j}}}}
  = \sqrt{\frac{s_{2j}}{s_0}} \prod_{k=0}^{j-1}\frac{e_{2k}}{e_{2k+1}}.
\end{equation}
Next, we use the fact that there is a constraint between $s_i$ and $e_i$, given that all nodes have degree $D$: in particular, $D s_i = e_{i-1} + e_{i}$ for all $i$.  Also recall that $e_i = r_i e_{i-1}$.  The expression above is simplified to (assuming $e_0 = D$):
\begin{align}\label{eqn:ratioequality}
  \prod_{k=0}^{j-1} \frac{t_{2k}}{t_{2k+1}}
  &= \sqrt{\frac{e_{2j-1}}{D s_0}(1+r_{2j})} \prod_{k=0}^{j-1}\frac{1}{r_{2k+1}} \nonumber \\
  &= \sqrt{\frac{D(1+r_{2j})}{s_0 r_0}} \prod_{k=0}^{j-1}\frac{\sqrt{r_{2k} r_{2k+1}}}{r_{2k+1}} \nonumber\\
  &= \sqrt{\frac{D(1+r_{2j})}{s_0 r_0}} \exp\left(\sum_{k=0}^{j-1} \frac{1}{2} \log \frac{r_{2k}}{r_{2k+1}}\right),
\end{align}
where in the second line, we used the fact that $e_j = e_0 \prod_{i=1}^j r_i$.  The coefficient multiplying the exponential is bounded from above and below by a constant, and we also know that $\max_i |\psi_i| \geq \Omega(1/\sqrt{n})$.  We then find the inequality
\begin{equation}
   \frac{C'}{\sqrt{n}} \leq \max_i |\psi_i| =  |\psi_0|\cdot C \exp\left[\frac{1}{2}\max_j\left(\sum_{k=0}^{j-1}  \log \frac{r_{2k}}{r_{2k+1}}\right)\right],
\end{equation}
or
\begin{equation}
   |\psi_0| \geq \frac{C}{\sqrt{n}} \cdot \exp\left[-\frac{1}{2}\max_j\left(\sum_{k=0}^{j-1}  \log \frac{r_{2k}}{r_{2k+1}}\right)\right].
\end{equation}

Finally, we can evaluate $|\psi_{2n}|$ in terms of $|\psi_0|$ by
setting $j = n$ in the expression for Eqn.~\eqn{ratioequality}.  In
conjunction with the lower bound above, we can compute
$|\psi_{2n}||\psi_0|$ to prove \eqn{zero-vec-bounds}.
\end{proof}

As a result, we establish the following lemma, which proves that the product of the wavefunction at the first and last sites does not become too small:
\begin{lemma}\label{lem:prob-Doob}
  There exists a universal constant $C>0$ such that for
for any $K>0$, the zero energy eigenstate of Hamiltonian $H$ satisfies
\begin{equation}
    |\psi_{2n}||\psi_0| \geq \frac{C}{n} \cdot \exp\left(-\frac 32 \sqrt{K n}\right),
  \end{equation}
  with probability $\geq 1-\exp(-K/\delta^2)$ over the choice of graph, where $|\log(r_k/r_{k+1})|\leq \delta$.
\end{lemma}

\begin{proof}
  Define the random variable $\eta_k := \log \frac{r_{2k}}{r_{2k+1}}$.  Since the $r_k$ are i.i.d.~we have $\bbE[\eta_k]=0$.   We also define the partial sums $X_j = \sum_{i=0}^j \eta_i$.
  In terms of the $X_j$, we can rewrite the RHS of Eqn.~\eqn{zero-vec-bounds} as
  \be
  \frac C n  \exp[\frac{X_n}{2} - \max_j X_j].
  \ee
The lemma will follow from proving that
  \be
  \max_j |X_j| \leq \sqrt{Kn}  \label{eq:max-Xj}
  \ee
  with high probability.  To prove this, observe that the $X_j$ form a martingale due to the independence of the $\eta_k$.  Moreover, the $\exp(X_j)$ form a submartingale, which can be seen via
  \be
  \mathbb{E}[\exp(X_{k+1})|X_k] = \mathbb{E}[\exp(\eta_k)]\exp(X_{k}) \geq \exp(\mathbb{E}[\eta_k]) \exp(X_{k}) = \exp(X_{k}).
  \ee
  Applying Doob's martingale inequality~\cite[Thm 7.3.1]{shiryaev_2019} gives
\begin{align}
  \mathbb{P}\left(\left|\max_j X_j\right| \geq \sqrt{K n} \right)
  &= \mathbb{P}\left(\left|\max_j\exp\left(\lambda X_j\right)\right| \geq \exp(\lambda \sqrt{K n})\right) \nonumber \\
  &\leq \frac{\mathbb{E}[\exp(\lambda X_n)]}{\exp(\lambda \sqrt{K n})}
    & \text{(Doob's inequality)} \nonumber\\
  &= \frac{1}{\exp(\lambda \sqrt{K n})} \prod_{i=1}^n \mathbb{E}[\exp(\lambda \eta_i)]
    & \text{(independence of the $\eta_i$)}\nonumber \\
  &\leq \frac{\exp(n \lambda^2 \delta^2/2)}{\exp(\lambda \sqrt{K n})}
  & (|\eta_i|\leq \delta)  \nonumber\\
  & = \exp(-\frac{K}{\delta^2}). & \left(\text{setting $\lambda=\frac{\sqrt{K}}{\sqrt{n}\delta^2}$}\right).
\end{align}

This completes the proof.  We remark that using a Hoeffding bound  and union bound instead of Doob's inequality would have replaced the $\sqrt{n}$ in the exponent with a $\sqrt{n\log(n)}$.
\end{proof}

Next, we would like to perform a time evolution of the Hamiltonian of the system in the supervertex subspace and construct a protocol for which we can traverse the graph with sufficiently high probability.  The existence of a zero energy state and its sub-exponential decay is crucial to facilitating a superpolynomial speedup.  This can be made more clear via the following lemma:
\begin{lemma}
Denote the gap of the Hamiltonian $\mathcal{H}$ in the supervertex subspace by $\Delta$, which is defined to be the smallest \emph{non-zero} eigenvalue of $\mathcal{H}$ in absolute value.  Starting from the state $\ket{S_{\text{init}}}$ and evolving for time $t$ averaged between $[0,\tau]$ for $\tau = \frac{4n}{C \Delta} \exp\left(1.5 \sqrt{K n}\right)$ where $C,K>0$ are constants, the EXIT probability satisfies
\begin{equation}
  \mathbb{P}(\text{EXIT}) \geq \frac{C^2}{4n^2}\exp\left(-3 \sqrt{K n}\right)
  \left(1-\exp(-\frac{K}{\delta^2})\right).
\end{equation}
\end{lemma}
To remind the reader, we select a particular graph from the random graph ensemble, and time evolve this graph for a time uniformly chosen in $t \in [0,\tau]$.  As long as $\tau$ is not too short (scaling as $\Delta^{-1} \exp(O(\sqrt{N}))$), the probability of traversal is large.
\begin{proof}
Recognizing from the previous lemma that $ |\psi_{2n}||\psi_0| \geq \frac{C}{n} \cdot \exp\left(-1.5 \sqrt{K n}\right)$ and utilizing the result from \lem{evolution}, we obtain the desired result.
\end{proof}

This result links the runtime of the algorithm to the spectral gap
from the zero energy eigenstate.  As long as the spectral gap is not
too small, the runtime of the algorithm will be sub-exponential.  The
final step is to prove that the spectral gap is larger than
$\exp(-\Theta(\sqrt{n}))$ with high probability.  We offer two proofs
for this. The first one is based on rigorous results relating to the
spectral theory of random hopping Hamiltonians, which gives a tight
bound of the spectral gap. It has the advantage of connecting to the
condensed-matter literature, as well as giving asymptotically matching
upper and lower bounds.  The other one follows directly from some
linear algebraic arguments.  While it only provides a lower bound, it
is much simpler and will be important in the next section when we
generalize to higher dimensions.   A reader in a hurry can safely skip the first proof in \sec{Dyson}.

\subsection{Proof by Dyson singularity and Kotowski-Vir{\'a}g theorem}
\label{sec:Dyson}
In this section, we provide an estimate on the spectral gap by leveraging known results from the literature.  The main observation in this subsection is that the spectral gap can be characterized by computing the density of states of the random hopping Hamiltonian, which is a measure of the spectral density of the eigenvalue distribution.  It is known that for the random hopping model, the density of states has a divergence near $E = 0$; this is known as the Dyson singularity~\cite{dyson1953dynamics}, and is thought to be a universal feature of disordered hopping systems in one dimension (although Dyson's original exact solution assumed Poisson random variables).  Due to this singularity, the spectral gap can become very small and one needs to carefully determine the behavior of the spectral density near the singularity.  A heuristic proof of the behavior near this singularity was first provided by Eggarter and Riedinger \cite{EggarterRiedinger}; more recently, Kotowski and Vir{\'a}g~\cite{KV17} have revisited this problem and rigorously proved properties of the spectral measure.  In particular, they proved the following theorem:
\begin{theorem}[{\cite[Theorem 3.10]{KV17}}]\label{thm:KV17}
Define the random variable $U_k = 2 \log \left|\frac{t_{2k}}{t_{2k+1}}\right|$.  Suppose $U_k$ satisfy the functional central limit theorem, and for some $\gamma > 2$, we have that $\mathbb{E} \left|\log |t_i|\right|^\gamma \leq \infty$ and for any $n > 1$ and $C$ a constant
\begin{equation}
    \mathbb{E}\left|U_1 + U_2 + \cdots + U_n\right|^\gamma \leq C n^{\gamma/2}.
\end{equation}
Choose $n = K^2 \log^2 (1/\epsilon)$, and assume that $n$ is even. If for all $i$, we have that $\mathrm{Var } \log |t_i| = \sigma^2$ and $t_i$ are identically distributed, then the spectral measure satisfies
\begin{equation}
    \left|\mu_n(-\epsilon,\epsilon) - \frac{\sigma^2}{|\log^2 (1/\epsilon)|}\right| \leq \frac{c_K}{|\log^2 (1/\epsilon)|},
\end{equation}
where $c_K$ is a constant that approaches zero as $K \to \infty$.
\end{theorem}
Next, we prove that the $\{U_i\}$ meet the conditions of \thm{KV17}.
We assume that the $r_i$ are independent random variables, which encompasses both random graph scenarios that we discussed above.
\begin{lemma}
For the random variable $t_i =\frac{ e_i}{\sqrt{s_i s_{i+1}}}$, the conditions of the Kotowski-Vir{\'a}g theorem are satisfied with $\gamma = 4$.
\end{lemma}
\begin{proof}
First, notice that $t_i = \frac{e_i}{\sqrt{s_i s_{i+1}}}$ which is identically distributed, but may possess correlations.  Furthermore,
\begin{equation}
    U_k = 2 \log \left|\frac{t_{2k}}{t_{2k+1}}\right| = 2 \log \left|\frac{e_{2k}}{e_{2k+1}}\sqrt{\frac{s_{2k+2}}{s_{2k}}}\right|.
\end{equation}
Therefore, from the analysis of \lem{prob-entrance-exit}
\begin{equation}
    \sum_{k=1}^{n} U_k = \log \left|\frac{1+r_{2n+1}}{D s_2}\right| + \sum_{k=1}^{n} \log \left|\frac{r_{2k}}{r_{2k+1}}\right|.
\end{equation}
If $\mathrm{Var } \log \left|\frac{r_{2k}}{r_{2k+1}}\right| = \sigma^2$, and since $\log \left|\frac{r_{2k}}{r_{2k+1}}\right|$ are independent, then by the central limit theorem, we know that
\begin{equation}
    \frac{1}{n \sigma^2}\sum_{k=1}^{n} U_k \overset{d}{\longrightarrow} \mathcal{N}(0,1).
\end{equation}
Next, we can verify the second condition by choosing $\gamma = 4$.  Define $V_k = \log \left|\frac{r_{2k}}{r_{2k+1}}\right|$ and observe that $\bbE[V_k]=0$ and that $\mathbb{E} \left|\log \left|\frac{r_{2k}}{r_{2k+1}}\right|\right|^\gamma = k_\gamma < \infty$.  Then we can calculate
\begin{align}
    \mathbb{E} (U_1 + U_2 + \cdots + U_n)^4 &\leq  4 \log \left|\frac{1+r_{2n+1}}{D s_2}\right| \mathbb{E} \left(\sum_{k=1}^{n} \log \left|\frac{r_{2k}}{r_{2k+1}}\right|\right)^3 + \mathbb{E} \left(\sum_{k=1}^{n} \log \left|\frac{r_{2k}}{r_{2k+1}}\right|\right)^4 \nonumber \\& \hspace{2cm} +6 \left(\log \left|\frac{1+r_{2n+1}}{D s_2}\right|\right)^2 \mathbb{E} \left(\sum_{k=1}^{n} \log \left|\frac{r_{2k}}{r_{2k+1}}\right|\right)^2 \nonumber \\
    &\leq 6 C' \sigma^2 n + 4 C \cdot \left[\mathbb{E} \sum_{i \neq j \neq k}V_i V_j V_k + 3 \mathbb{E} \sum_{i \neq j}V_i^2 V_j + \mathbb{E} \sum_{i}V_i^3\right] + \mathbb{E} \sum_{i \neq j \neq k \neq \ell}V_i V_j V_k V_\ell \nonumber \\& \hspace{1.5cm}+ 4 \mathbb{E} \sum_{i \neq j \neq k}V_i^2 V_j V_k + 6 \mathbb{E} \sum_{i\neq j}V_i^2 V_j^2 + 4 \mathbb{E} \sum_{i \neq j}V_i^3 V_j + \mathbb{E} \sum_{i}V_i^4 \nonumber \\
    &\leq 6 C' \sigma^2 n + 4 C \cdot \left[0 + 0 + 2 k_3 n\right] + 0 + 0 + 2 k_2^2 n^2 + 0 + 2 k_4 n \leq C'' n^2,
\end{align}
where a $0$ indicates that the corresponding expectation value in the previous expression vanishes (using $\mathbb{E}[V_i] = 0$).  Thus, the conditions of the Kotowski-Vir\'ag theorem are satisfied for our random graph ensemble.
\end{proof}

There is another criterion in the Kotowski-Vir{\'a}g theorem, which is the assumption that $n$ is even.  Recall that $n$ being odd is a requirement for the existence of a zero energy eigenstate.  Thus, when $n$ is odd the density of states should have an additional point mass term at $E = 0$.  Despite this additional point mass term, the spectral gap for even length chains and the spectral gap of odd length chains can be exactly related using the following lemma:
\begin{lemma}\label{lem:eig-odd-even}
Define $H_{2n}$ to be a random hopping Hamiltonian on $2n$ sites, and define $H_{2n-k}$ to be the same Hamiltonian with the last $k$ sites removed.  For $n \in \mathbb{Z}$ the spectral gap of Hamiltonian $H_{2n-1}$ is larger than or equal to the smallest eigenvalue of $H_{2n}$ in absolute value.
\end{lemma}
\begin{proof}
For Hamiltonian $H_{2n}$, a Hamiltonian of odd dimension can be constructed by restricting $H_{2n}$ to the submatrix formed with the first $2n-1$ columns and rows.  Calling $E_i(H)$ the $i^{\text{th}}$ largest eigenvalue of $H$.  By the eigenvalue interlacing theorem, we know that
\begin{equation}
    E_1(H_{2n}) \geq E_1(H_{2n-1}) \geq E_2(H_{2n}) \geq E_2(H_{2n-1}) \geq \cdots \geq E_{2n-1}(H_{2n-1}) \geq E_{2n}(H_{2n}).
\end{equation}
Because $H_{2n-1}$ has an exact zero eigenvalue, the smallest eigenvalue in absolute value of $H_{2n}$ satisfies $\Delta_{2n-1} \geq \min_{i} |E_{2n, i}|$, implying the thesis.
\end{proof}
Next, we bound the minimum eigenvalue $\min_{i} |E_{2n, i}|$ with high probability, using the following lemma:
\begin{lemma}\label{lem:spectral-gap-lower}
With constant probability, the spectral gap satisfies $\Delta \geq \exp(-\sqrt{2n}L)$ for a large enough $L$.
\end{lemma}

\begin{proof}
We choose $\epsilon = \exp(-\sqrt{2n}L)$.  Then, the result of the Kotowski-Vir{\'a}g theorem implies that
\begin{equation}
    \mu_n\left(-\exp(-\sqrt{2n}L), \exp(-\sqrt{2n}L)\right) = \frac{\sigma^2 }{2nL^{2}}\left(1 + o(1)\right),
\end{equation}
where $o(1) \to 0$ as $K \to \infty$ (recall $n = K^2 \log^2(1/\epsilon)$).  As a result, the expected number of eigenvalues residing in the interval $(-\exp(-\sqrt{2n}L), \exp(-\sqrt{2n}L))$ is
\begin{equation}
    \mathbb{E}[N] \triangleq \mathbb{E} \left(\sum_{i=1}^{2n} \mathbbm{1}\left[|E_i| \leq \exp(-\sqrt{2n}L)\right]\right) = \frac{\sigma^2}{L^{2}}\left(1 + o(1)\right).
\end{equation}
By Markov's inequality, we can find a lower bound on the spectral gap
\begin{align}
   \mathbb{P}\left(N \leq 1\right) &= 1 - \mathbb{P}\left(N > 1\right) \nonumber \\
  &\geq 1 - \mathbb{E}[N] \nonumber \\
   &= 1 - \sigma^2/L^{2}\cdot \left(1 + o(1)\right) ,
\end{align}
Therefore, for a large enough $L$, the spectral gap $\min_{i} |E_{2n, i}| \leq \exp(-\sqrt{2n}L)$ with constant probability.  By \lem{eig-odd-even}, this implies that $\Delta \geq \exp(-\sqrt{2n}L)$ with constant probability.
\end{proof}

Having proven that the spectral gap is not too small, this results in our main theorem:
\begin{theorem}\label{thm:main-1D}
With probability at least $3/4$, a graph selected from our random graph ensemble with random edge-edge ratios (see \defn{ratios}) will find the exit node with probability at least $3/4$ after taking $\exp(O(\sqrt{n}))$ queries to the quantum oracle in \defn{oracle}.
\end{theorem}

\begin{proof}
By \lem{spectral-gap-lower}, with constant probability $\Delta \geq \exp\left(-\Theta(\sqrt{n})\right)$. In \lem{evolution}, we have proved that the quantum walk reaches the EXIT with probability at least $1/n^2\cdot\exp(-\Theta(\sqrt{n}))$ if the maximal simulation time $\tau$ satisfies $\tau \leq \frac{4n}{C \Delta} \exp\left(\Theta(\sqrt{n})\right)$.  The exit probability can be boosted by standard amplification procedures to a constant close to 1 after $\exp(O(\sqrt{n}))$ trials.  This results in a total query complexity and runtime of $\exp(O(\sqrt{n}))$ because the $n^{O(1)}$ term can be absorbed by the exponential term.  The probability of successfully selecting such a graph from the random graph ensemble can be chosen to be a large constant, say $3/4$, by choosing a large enough $K$.
\end{proof}

\subsection{Proof by exact inversion} \label{sec:Haminv}

There is an alternate way of deriving the spectral gap for this model.  This method is simpler, but provides only a lower bound, while the results of Kotowski and Vir{\'a}g~\cite{KV17} imply tightness.  However, the simplicity of the proof makes this more useful for analysis in higher dimensions, which is explored in \sec{highdim}.  The key observation is that for an even length one-dimensional chain, the effective Hamiltonian in the supervertex space (Eqn.~\eqn{hamiltonian}) admits an exact form for its inverse.  Supposing the dimension of the Hamiltonian is even, for a Hamiltonian taking the form
\begin{equation}
    H = \begin{pmatrix}
    0 & t_{0}   & \cdots     & 0  & 0 \\
    t_{0}  & 0       & t_{1}   & \cdots     & 0   \\
     \vdots &       t_{1}  & \ddots         & t_{2n-2}    & \vdots    \\
     0 &        \vdots &          t_{2n-2} &     0      & t_{2n-1}   \\
     0 &      0   &        \cdots   &       t_{2n-1}    & 0
  \end{pmatrix},
\end{equation}
the elements of the inverse of $H$ are given by (for $i < j$)
\begin{equation}
    \mel{S_i}{H^{-1}}{S_j} = \begin{cases}
\frac{1}{t_i} \prod_{k = (i+1)/2}^{(j-2)/2}\left(-\frac{t_{2k}}{t_{2k+1}}\right), & \text{for } j \text{ even, } i \text{ odd}\\
0, & \text{otherwise }\\
\end{cases}
\end{equation}
and the inverse matrix is symmetric, which determines entries when $i > j$.  Therefore, we have the following result for the spectral gap:
\begin{lemma}\label{lem:easyboundgap}
With high probability, the spectral gap satisfies $\Delta = \left|\lambda_{\text{min}}\right| \geq (2n)^{3/2} \exp\left(-\sqrt{K n}\right)$ for some constant $K>0$.
\end{lemma}
\begin{proof}
We know that because $H$ is symmetric, the minimum eigenvalue satisfies
\begin{equation}
  \left|\lambda_{\text{min}}\right| = \frac{1}{\norm{H^{-1}}}
  \geq \frac{1}{\sqrt{2n} \norm{H^{-1}}_{1,\infty}}.
\end{equation}
Here $\norm{A}$ means the operator norm or largest singular value of
$A$, and $\norm{A}_{1,\infty}$ means we take the $\ell_1$ norm of each
column and then take the $\ell_\infty$ norm of all of those; in other
words, it is the maximum sum of the absolute values of the entries of
a column.

Knowing the exact structure of the inverse, it is simple to bound the max column norm
\begin{equation}
    \norm{H^{-1}}_{1,\infty} = \max\left[\max_{i \text{ odd}}  \sum_{j}\frac{1}{t_i} \prod_{k = (i+1)/2}^{(j-2)/2}\left(\frac{t_{2k}}{t_{2k+1}}\right) \genfrac{}{}{0pt}{}{}{,} \max_{j \text{ even}}  \sum_{i}\frac{1}{t_i} \prod_{k = (i+1)/2}^{(j-2)/2}\left(\frac{t_{2k}}{t_{2k+1}}\right)\right].
\end{equation}
We use Eqn.~\eqn{ratioequality} to relate the ratio of the hoppings to the ratio of the variables $r_k$.
We showed in \lem{prob-Doob} that the random variable $\eta_k = \log\frac{r_{2k}}{r_{2k+1}}$ satisfies a concentration inequality since it has a bounded distribution.
In particular, we found that $\mathbb{P}\left(\max_j \left|\sum_{k=1}^j \eta_k\right| \geq 1.5\sqrt{K n} \right) \leq e^{-K/\delta^2}$.  The expression for the infinity norm above involves partial sums $S_{ij} = \sum_{k=i}^j \eta_k = X_j - X_{i-1}$, where $X_j = \sum_{k\leq j} \eta_k$.  This implies that
\begin{equation}
\mathbb{P}\left(\max_{i,j}\left|S_{ij}\right| \geq 3 \sqrt{K n} \right) \leq e^{-K/\delta^2}.
\end{equation}
The right hand side goes to zero for $K$ large enough.  We also have the factor of $1/t_i$ to bound.  Simply rewriting of this quantity
\begin{equation}
\frac{1}{t_i} = \frac{\sqrt{s_i s_{i+1}}}{e_i} = \frac{1}{e_i}\sqrt{\frac{e_i}{d_{R,i}} \frac{e_i}{d_{L,i+1}}} = \frac{1}{\sqrt{d_{R,i} d_{L,i+1}}} \leq \frac{1}{2}
\end{equation}
and is therefore bounded by a constant for all $i$.  Therefore,
\begin{equation}
    \mathbb{P}\left(\norm{H^{-1}}_{1,\infty} \geq 2n \exp\left(3\sqrt{Kn}\right)\right) \leq e^{-K/\delta^2}.
  \end{equation}
In other words, with high probability $\left|\lambda_{\min}\right| \geq (2n)^{3/2} \exp\left(-3\sqrt{K n}\right)$. The spectral gap from the zero energy state is larger than $|\lambda_{\text{min}}|$, which implies the thesis.
\end{proof}

In general, for an arbitrary one-dimensional landscape, the spectral gap can be lower bounded in terms of the distance metric $d(2i,2j+1) = \sum_{k=i}^j \log \frac{r_{2k}}{r_{2k+1}}$.  The spectral gap has the scaling behavior
\begin{equation}
    \Delta^{-1} \leq \text{poly}(n) \cdot
    \exp(\sup_{i,j} \left|d(i,j)\right|).
\end{equation}
This suggests that if there is a ``bottleneck'' region in which this distance metric can grow to be linear in system size, then the gap is exponentially small and a quantum speedup is not possible.  This is an unusual distance metric in that it depends on the difference of logarithms of $r$; thus, if $r$ has positive expectation value, the metric is still small.  The fact that a bias ($\mathbb{E}[r] > 0$) does not affect the quantum bottleneck but considerably changes the classical bottleneck is an intuitive way of understanding the reason for the superpolynomial speedup.

We note that the spectral gap of classical walks can be also characterized, to a quadratic factor, in terms of a quantity called {\em conductance} that can be thought of as measuring the tightest bottleneck in the graph; see \cite{levin09}. However, the algorithmically relevant classical quantity is the {\em hitting time}.  In the model we consider this is polynomially related to $\sum_v s_v / \min_v s_v$, or equivalently, $\max_{i<j} \exp(\sum_{i\leq k<j} \log r_k)$.


\section{Higher-dimensional random graph ensembles} \label{sec:highdim}
In the previous section, our analysis was restricted to a one-dimensional random graph model.  As promised, we will consider a more general supernode graph $G$ and argue that it is possible to find superpolynomial \emph{and} exponential speedups.  To remind the reader, we associate a state with each supernode $\ket{S_v} = \frac{1}{\sqrt{s_v}} \sum_{\alpha \in S_v} \ket{\alpha}$.  The effective Hamiltonian in the supervertex space then has the matrix elements
\begin{equation}
    \mel{S_u}{H}{S_v} = \frac{e_{uv}}{\sqrt{s_u s_v}}, \hspace{0.5cm} u \sim v
\end{equation}
and zero otherwise, where $e_{uv}$ is the total number of edges connecting nodes in supernode $u$ and nodes in supernode $v$.  Analogous to the one-dimensional chain, we assume that the number of edges connecting a fixed node $\alpha \in S_u$ to nodes in $S_v$ is the same for all nodes $\alpha$.  We also require that the graph be $D$-regular, giving the additional constraint $\sum_{v}e_{uv} = D s_u$.  For now, we will assume that $\mel{S_u}{H}{S_v} \triangleq t_{uv}$ and we will later substitute in graph variables.

Throughout this section, we will modify our notation, calling the linear dimension of the lattice $N$ rather than $n$.

\subsection{2D Lieb lattice}\label{sec:2D}
We first consider a square lattice $G$ of size $N \times N$.  Denote the vertices of this lattice by $\Delta_0$.  For each edge on the lattice, decorate it with an additional vertex.  Call the set of ``decorated'' vertices $\Delta_1$. The resulting graph, with vertices $\Delta_0 \cup \Delta_1$, will be called the \emph{Lieb lattice} for the square lattice $G$.  The Lieb lattice is bipartite, and a random hopping model defined on it will generically have zero modes.  The following is true from a dimension counting argument:

\begin{lemma}\label{lem:zeromodes}
The number of zero modes on $G$ with support exclusively on sublattice $\Delta_1$ is at least $N^2-2N$.
\end{lemma}
\begin{proof}
We solve for the eigenvalue problem $H \psi = 0$, where the Hamiltonian has matrix elements $\mel{S_u}{H}{S_v} = t_{uv}$.  Solving this eigenvalue condition on a node $i$ in sublattice $\Delta_0$ gives $\sum_{i\sim j} t_{ij} \psi_i = 0$, which gives an equation relating amplitudes of the wavefunction on sites in sublattice $\Delta_1$.  Since the lattice is bipartite, solving for the eigenvalue condition for a node on sublattice $\Delta_1$ gives an equation only in terms of the amplitudes of wavefunctions on sublattice $\Delta_0$.  Therefore, there are $N^2$ equations relating amplitudes in sublattice $\Delta_1$ and $2N(N-1)$ unknown amplitudes.  The dimension of the kernel of this system of equations is at least $N^2-2N$, which is the number of linearly independent zero modes.
\end{proof}

Our strategy for guaranteeing a zero mode with the right properties will be to impose a restriction on the hopping amplitudes which we call a \emph{gauge condition}.  This is inspired by ideas from cohomology that we will discuss in more generality in \sec{cohomology}.  For now, we will specialize the argument to the 2D Lieb lattice.

\begin{definition}[Gauge condition]\label{def:gauge}
A set of hoppings $t_{ij}$ given to the tight binding Hamiltonian $\mathcal{H}$ is said to satisfy a gauge condition if for each loop $\ell$ defined by the sequence of nodes $\ell(0), \ell(1), \ldots, \ell(2n-1), \ell(2n) = \ell(0)$,
\begin{equation}
    \prod_{k = 0}^{n-1}-\frac{t_{\ell(2k), \ell(2k+1)}}{t_{\ell(2k+1), \ell(2k+2)}} = 1.
\end{equation}
In particular, it can be seen that if the gauge condition is satisfied for a particular cycle basis of the supernode graph, then it will be satisfied for all cycles in the graph.
\end{definition}

The reason for the gauge condition is that it leads to the existence of another zero mode which is immediately orthogonal to all other zero modes and is supported in the $\Delta_0$ sublattice.  This provides the basis for our speedup.
\begin{lemma}\label{lem:gauge-uniqueness}
If the gauge condition is satisfied on $G$, then there exists a unique zero mode with support exclusively on sublattice $\Delta_0$.
\end{lemma}
\begin{proof}
We solve the zero mode equation for sites in sublattice $\Delta_1$.  For a given site $i$ with two neighbors $j$ and $k$, the zero mode equation is $\psi_j t_{ij} + \psi_k t_{jk} = 0$, which gives $\psi_j = -\frac{t_{jk}}{t_{ij}} \psi_k$.  Call the amplitude of node at the bottom left corner of the 2D grid $\psi_{0,0}$.  Then, calling $\ell_{m,n}$ a directed path from $(0,0)$ to $(m,n)$ of length $|\ell_{m,n}|$, the amplitude $\psi_{m,n}$ can be written as
\begin{equation}\label{eqn:2Dzeromode}
    \psi_{m,n} = \psi_{0,0} \prod_{k = 0}^{|\ell_{m,n}|/2-1}-\frac{t_{\ell_{m,n}(2k), \ell_{m,n}(2k+1)}}{t_{\ell_{m,n}(2k+1), \ell_{m,n}(2k+2)}}.
\end{equation}
Normally, one would choose two different paths $\ell$ and $\ell'$ which start at $(0,0)$ and end at $(m,n)$ which yields for the right hand side $\psi_{0,0} \cdot C(\ell)$ and $\psi_{0,0} \cdot C(\ell')$.  This immediately implies $\psi_{0,0} = 0$ and no nontrivial zero mode exists.  However, since the gauge condition is satisfied, the product is \emph{path independent} by definition, and the zero mode is unique and given by the above equation.

We note that since the square lattice (removing the decorated nodes but keeping the edges) is bipartite, there are no odd length cycles and hence the negative sign in the gauge condition can be removed: therefore, the hoppings can be chosen to be positive, which is necessary by design.  On a non-bipartite lattice, this is no longer the case.
\end{proof}
Next, we show that it is possible to exactly count the number of zero modes with support on sublattice $\Delta_1$:
\begin{lemma}
If the gauge condition is satisfied on $G$, then there exist exactly $(N-1)^2$ zero modes with support exclusively on sublattice $\Delta_1$.
\end{lemma}
\begin{proof}
From \lem{zeromodes}, there are $N^2$ linear equations with $2N(N-1)$ unknown amplitudes.  We must determine how many equations are linearly independent; call the set of linear equations $L$.  We note that a particular amplitude $\psi_i$ appears in two linear equations, and there is no strict subset of linear equations $S \subset L$ where \emph{each} amplitude present in this subset appears in two equations.  Thus, any subset $S$ of equations are linearly independent, and the only possibility for linear dependence is if there exists constants $c_i \neq 0$ \emph{for all $i$} such that $\sum_{i} c_i \vec{v}_i = 0$ where $\vec{v}_i$ are the vectors of coefficients of each of the linear equations.

Suppose we take two equations: $\sum_{i\sim j} t_{ij} \psi_i = 0$ and $\sum_{i\sim k} t_{ik} \psi_i = 0$, which have site $\psi_a$ in common.  Since the only contribution from $\psi_a$ comes from the vectors corresponding to these two equations, the coefficients $c_1$ and $c_2$ must satisfy $c_1 = -\frac{t_{ja}}{t_{ak}} c_2$.  Repeating this for all pairs of vectors sharing a common site, the relationship between the coefficients $c_i$ is precisely equal to the zero eigenvalue equations solving for the amplitude of the zero mode on sublattice $\Delta_0$; therefore, if the gauge condition is satisfied, there is a unique solution for the $c_i$.

As there is a single unique set of $c_i$ solving $\sum_{i} c_i \vec{v}_i = 0$, there are therefore $N^2-1$ linearly independent equations.  The kernel of this system then has dimension exactly equal to $2N(N-1) - (N^2-1) = (N-1)^2$.
\end{proof}

We may use this result to argue that the zero mode above is the only one that we need to consider when computing the amplitude of measuring the final state $\psi_{N,N}$.  This is summarized in the Lemma below:
\begin{lemma} \label{lem:evolution-high}
  Denote the gap of the Hamiltonian $H$ in the supervertex subspace by $\Delta$, which is defined to be the smallest \emph{non-zero} eigenvalue of $H$ in absolute value.
  Consider the following algorithm
  \bit
\item Start with the state $\ket{S_{\text{init}}}$, where $\text{init}$ is the vertex located at $(0,0)$
\item Evolve for a random time $t$ drawn uniformly between $[0,\tau]$.  Here $\tau = \frac{4}{\Delta \left|\braket{S_{\text{exit}}}{0} \braket{0}{S_{\text{init}}}\right|}$, where $\text{exit}$ is the vertex located at $(N,N)$.
  \item Measure the state in the vertex basis.
    \eit
    The probability of reaching $S_{\text{exit}}$, denoted $\mathbb{P}(\text{EXIT})$, satisfies
\begin{equation}
    \mathbb{P}(\text{EXIT}) \geq \frac{1}{4}\left|\braket{S_{\text{exit}}}{\psi} \braket{\psi}{S_{\text{init}}}\right|^2.
\end{equation}
where $\ket{\psi}$ is the zero mode with support on sublattice $\Delta_0$ (which exists assuming the gauge condition).
\end{lemma}
\begin{proof}
In particular, we follow the same procedure as outlined in \lem{evolution}, up to the point where we decompose the time averaged amplitude:
\begin{align}
    \frac{1}{\tau}\int_0^\tau \dd t\,\mel{S_{N,N}}{e^{-iHt}}{S_{0,0}} &= \frac{1}{\tau}\int_0^\tau \dd t\, \sum_E e^{-iEt} \braket{S_{N,N}}{E}\braket{E}{S_{0,0}} \nonumber \\
    &= \sum_{m}\braket{S_{N,N}}{0_m} \braket{0_m}{S_{0,0}} + \sum_{E \neq 0} \frac{1 - e^{-iE\tau}}{i E \tau} \braket{S_{N,N}}{E}\braket{E}{S_{0,0}},
\end{align}
where now the subscript $m$ in $\ket{0_m}$ denotes the $m$th zero mode, and $|S_{0,0}\>$, $|S_{N,N}\>$ denote the supervertex state at the bottom-left and top-right corners of the 2D lattice.  We separate the zero modes into the $N^2$ zero modes with support on sublattice $\Delta_1$ and one zero mode with support on sublattice $\Delta_0$.  However, since $\ket{S_{0,0}}$ and $\ket{S_{N,N}}$ have support only on sublattice $\Delta_0$, the above expression reduces to
\begin{align}
    \frac{1}{\tau}\int_0^\tau \dd t\,\mel{S_{N,N}}{e^{-iHt}}{S_{0,0}} &= \frac{1}{\tau}\int_0^\tau \dd t\, \sum_E e^{-iEt} \braket{S_{N,N}}{E}\braket{E}{S_{0,0}}  \nonumber
     \\ &= \braket{S_{N,N}}{\psi} \braket{\psi}{S_{0,0}} + \sum_{E \neq 0} \frac{1 - e^{-iE\tau}}{i E \tau} \braket{S_{N,N}}{E}\braket{E}{S_{0,0}}, \label{eqn:time2D}
\end{align}
and the rest of the argument follows as in \lem{evolution}.
\end{proof}

Thus, we only need to bound the decay property of the zero mode $\ket{\psi}$ as well as the spectral gap to the smallest nonzero eigenvalue.

First, we would like to prove a subexponential decay rate of the zero mode.  This can be done assuming the hoppings are random in conjunction with a central limit theorem.  We also must select a random graph model with respect to the parameters $s_i$ and $e_{ij}$ in the original graph; in terms of the $t_{ij}$ variables, this results in correlated randomness.  We start by writing the zero mode in Eqn.~\eqn{2Dzeromode} in terms of the $s_i$ and $e_{ij}$ variables.  This gives
    \begin{equation}
        \psi_{m,n} = \psi_{0,0} \sqrt{\frac{s_{m,n}}{s_{0,0}}} \prod_{k \in \ell_{m,n}}^{|\ell|} \frac{e_{\ell(2k), \ell(2k+1)}}{e_{\ell(2k+1), \ell(2k+2)}}.
    \end{equation}
    Since the result is path independent, we choose $\ell$ to be the path from $(0,0)$ to $(m,0)$ and subsequently from $(m,0)$ to $(m,n)$.  Utilizing the edge-edge ratios defined in \defn{EELieb} for $d = 2$, i.e. $r^{(1)}_{i,j} = \frac{e_{(i+1/2,j) \to (i+1,j)}}{e_{(i,j) \to (i+1/2,j)}}$ and $r^{(2)}_{i,j} = \frac{e_{(i,j+1/2) \to (i,j+1)}}{e_{(i,j) \to (i,j+1/2)}}$, we may write the zero mode amplitudes as
    \begin{align}
        \psi_{m,n} &= \psi_{0,0} \sqrt{\frac{s_{m,n}}{s_{0,0}}} \prod_{i = 0}^{m-1} \frac{1}{r^{(1)}_{i,0}} \prod_{j = 0}^{n-1} \frac{1}{r^{(2)}_{m,j}}= \psi_{0,0} \sqrt{\frac{s_{m,n}}{s_{0,0}}} \prod_{i = 0}^{m-1} \frac{1}{r^{(1)}_{i,n}} \prod_{j = 0}^{n-1} \frac{1}{r^{(2)}_{0,j}},
    \end{align}
    where the second line is an equivalent expression due to gauge invariance.  We note that $s_{m,n} = \frac{1}{D}\left(e_{(m,n) \to (m+1/2,n)} + e_{(m-1/2,n) \to (m,n)} + e_{(m,n) \to (m,n+1/2)} + e_{(m,n-1/2) \to (m,n)}\right)$ for nodes in sublattice $\Delta_0$.  Using the edge-edge ratios, we may re-express this as
    \begin{equation}
        s_{m,n} = \frac{1}{D}\left(e^{(2)}_{m,0} (1+r^{(2)}_{m,n-1/2})\prod_{i=0}^{2n-2} r^{(2)}_{m,i/2} + e^{(1)}_{0,n} (1+r^{(1)}_{m-1/2,n})\prod_{i=0}^{2m-2} r^{(1)}_{i/2,n}\right).
    \end{equation}

    In general, the variables $e_{0,n}^{(1)}$ and $e_{m,0}^{(2)}$ are free and can be considered to be boundary conditions, which along with the edge-edge ratios uniquely determine the shape of the graph.  As such, we choose them (somewhat arbitrarily) to be: \footnote{This choice is made in particular so that the graph roughly has a similar number of edges connecting to each neighbor.  We did not explore the effect of changing these boundary conditions on the hitting time of the quantum walk.}
    \begin{equation}\label{eqn:free-parameter-choice}
        e_{0,n}^{(1)} = 
        e_{0,0}^{(2)}\prod_{i = 0}^{2n-1} r_{0,i/2}^{(2)} = e_{0,n}^{(2)} \hspace{1cm} e_{m,0}^{(2)} = 
        e_{0,0}^{(1)}\prod_{i = 0}^{2m-1} r_{i/2,0}^{(1)} = e_{m,0}^{(1)}.
            \end{equation}
    Then, the amplitude of the wavefunction can be written as
    \begin{equation}
        \hspace{-0.5cm} \psi_{m,n} = \frac{\psi_{0,0}}{\sqrt{D s_{0,0}}} \sqrt{ e_{0,0}^{(2)} (1+r^{(1)}_{m,n}) \frac{\prod_{i=0}^{m-1}  r_{i+1/2, n}^{(1)}}{\prod_{i=0}^{m}  r_{i, n}^{(1)}} \frac{\prod_{j=0}^{n}  r_{0, j+1/2}^{(2)}}{\prod_{j=0}^{n}  r_{0, j}^{(2)}} +    e_{0,0}^{(1)} (1+r^{(2)}_{m,n}) \frac{\prod_{j=0}^{n-1}  r_{m, j+1/2}^{(2)}}{\prod_{j=0}^{n}  r_{m, j}^{(2)}} \frac{\prod_{i=0}^{m}  r_{i+1/2, 0}^{(1)}}{\prod_{i=0}^{m}  r_{i, 0}^{(1)}}}.
    \end{equation}
    Let us analyze this bulky expression in more detail.  The quantities
    \begin{equation}
        e_{0,0}^{(2)} (1+r^{(1)}_{m,n}) \hspace{0.5cm} \text{and} \hspace{0.5cm} e_{0,0}^{(1)} (1+r^{(2)}_{m,n})
    \end{equation}
    are bounded above and below by constants. Thus, it suffices to prove concentration for the quantity
    \begin{equation}\label{eqn:quantities}
        \sum_{i = 0}^{m-1} \left(\log r_{i+1/2, n}^{(1)} - \log r_{i, n}^{(1)}\right) + \sum_{j = 0}^{n-1} \left(\log r_{0, j+1/2}^{(1)} - \log r_{0, j}^{(1)}\right)
    \end{equation}
and the analogous quantity in the second term in the square root.  These have expectation value zero, but in general the terms in the sum are \emph{not} independent because the gauge condition introduces correlations; thus, computing the variance and higher moments requires a more careful analysis.  We may solve this problem by working with height variables.  Define the height field
\begin{equation}\label{eqn:height}
    \varphi_{i,j} \triangleq \sum_{k=0}^{\left|\ell_{i,j}\right|/2-1} \log r_{\ell_{i,j}(2k), \ell_{i,j}(2k+1)},
\end{equation}
where $\left|\ell_{i,j}\right|$ is a path from $(0,0)$ to $(i,j)$.  One can show that this is consistent with the height fields introduced in \defn{heightfields} when $d = 2$.  We note that in terms of the $r$ variables the gauge constraint is that the sum of $\log r$ around a loop is zero; this is consistent with $\log r$ written as a difference of height functions.  The height function is constraint free, so in principle we may choose any distribution for it and this will realize an ensemble of random graphs, which is consistent with the random graph ensemble introduced in \defn{liebensemble}.

We must also assign height functions to sites located at half integer coordinates due to a similar gauge invariance condition; these height functions share the same definition as in Eqn.~\eqn{height}, except $\ell$ is a path which starts and ends at a node in sublattice $\Delta_0$, the edge sublattice:
\begin{align}
    \chi^{(2)}_{i+1/2,j} &\triangleq \sum_{k=0}^{\left|\ell_{i+1/2,j}\right|/2-1} \log r_{\ell_{i+1/2,j}(2k), \ell_{i+1/2,j}(2k+1)} \\
   \chi^{(1)}_{i,j+1/2} &\triangleq \sum_{k=0}^{\left|\ell_{i,j+1/2}\right|/2-1} \log r_{\ell_{i,j+1/2}(2k), \ell_{i,j+1/2}(2k+1)}.
\end{align}
In accordance with the notation used in \defn{heightfields}, the 3 height fields can be compactly written as $(\varphi, \chi^{(1)}, \chi^{(2)})$.

We saw that for 1D supergraphs, the classical runtime depended on the expectation of $\log(r)$ while the quantum walk performance depended on the standard deviation of $\log(r)$.  Similarly in 2D, the quantum walk's runtime will depend on fluctuations in the height fields, while the classical runtime will depend on their mean values.  To ensure a quantum-classical separation we want the entrance and exit to be separated by a barrier of large height.  One way to achieve this is to choose the height fields from a distribution that is small in the corners near the entrance and exit, large in the middle, and with bounded variation between neighboring points.  This can be viewed as a generalization of the welded tree to 2D (as well as the more general connectivity between supervertices that we already saw in 1D).

The simplest height field distribution for $\varphi$ that captures all of these requirements is the biased sub-Gaussian free field, which is introduced in \defn{BsGFF}.  As a result, the delocalization properties of the zero mode are related to expectation values under the biased Gaussian free field, whose lattice partition function is given by
\begin{equation}
 Z = \frac{1}{Z}\int \mathcal{D}\varphi \exp\left(-\frac{1}{2 g^2}\sum_{\langle p, q \rangle} (\varphi_{p} - \varphi_{q} - J_{pq})^2\right).
 \end{equation}
where $J_{pq}$ is a local gradient for the typical configuration of a single mountain.  In particular, the biased Gaussian free field in the continuum limit is equivalent to a massless free field theory (in the Euclidean signature) with a source:
\begin{equation}
    Z = \int \mathcal{D}{\varphi(x,y)} \exp\left(-\int d^2x\, \left[(\nabla \varphi)^2 + \mathcal{J}(x,y) \varphi\right]\right).
\end{equation}
where $\mathcal{J}(x,y) = \nabla J(x,y)$ denotes a source for the field $\varphi$.  We also to define distributions for the two sets of height fields defined on the midpoints of the edges, $\chi^{(1)}$ and $\chi^{(2)}$, which we choose to be independently and identically distributed BsGFF distributions as well.  Despite this assumption, our results can be readily generalized to the case where the distributions of the different height fields are not independent.

Choosing the biased Gaussian free field is out of calculational convenience, as it allows us to compute the variance and higher moments of the quantities in Eqn.~(\ref{eqn:quantities}).  We first convert the following expression to one in terms of height fields:
\begin{equation}\label{eq:quantityheightfields}
\sum_{i = 0}^{m-1} \left(\log r_{i+1/2, n}^{(1)} - \log r_{i, n}^{(1)}\right) + \sum_{j = 0}^{n-1} \left(\log r_{0, j+1/2}^{(1)} - \log r_{0, j}^{(1)}\right) = (\varphi_{m,n} - \varphi_{0,0}) + (\chi_{m+1/2,0}^{(1)} - \chi_{0,0}^{(1)}) + (\chi_{m,n+1/2}^{(2)} - \chi_{m,0}^{(2)}).
\end{equation}
Using the biased Gaussian free field to evaluate the expectation value of the exponential of the above quantity, we can prove the following Lemma:

\begin{figure}
    \centering
    \includegraphics[scale=0.4]{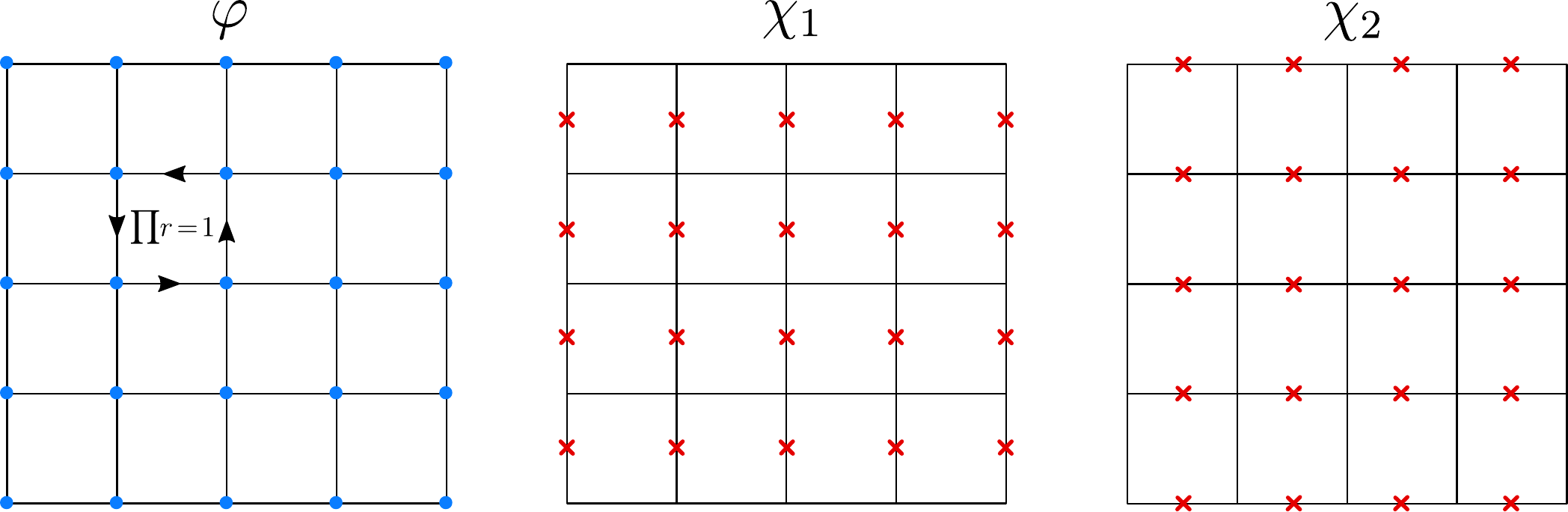}
    \caption{The three sets of fields which independently obey a biased GFF.}
    \label{fig:heightfields}
\end{figure}

\begin{lemma}\label{lem:2D-Laplacian}
With high probability, $\exp\left(\sum_{i = 0}^{m-1} \log r_{i+1/2, n}^{(1)} - \log r_{i, n}^{(1)} + \sum_{j = 0}^{n-1} \log r_{0, j+1/2}^{(2)} - \log r_{0, j}^{(2)}\right)$ is polynomial in $N$ for any $m,n\in[N]$.
\end{lemma}
\begin{proof}
We first note that the second term in each sum of the exponential can be treated independently of the first, because $r$ defined at half integer coordinates are independent of the $r$ defined at integer coordinates.  Next, we use the sub-Gaussian property of the distributions for the height fields, namely that $\exists \lambda$ such that for all $u$, $v$, $t$:
\begin{equation}
    \mathbb{E}_{\text{BsGFF}}\left[\exp\left((\varphi_v - \varphi_u) t\right)\right] \leq \mathbb{E}_{\text{BGFF}}\left[\exp\left((\varphi_v - \varphi_u) \lambda t\right)\right].
\end{equation}
We now need to compute the right hand sides of the above expression under a BGFF with an unspecified source term. These will be used to bound the moment generating function for the height field distributions we work with, which will yield  a reliable concentration inequality.  We now unpack the expression in Eqn.~\eq{quantityheightfields}, computing the expectation value of the exponential of it.  We first use the identity
\begin{equation}
\mathbb{E}_{\text{BGFF}}\left[\exp\left(\lambda t \sum_{j = 0}^{m-1} \log r_{j, n}^{(1)} + \lambda t\sum_{j = 0}^{n-1} \log r_{0, j}^{(2)}\right)\right] = \mathbb{E}_{\text{BGFF}}\left[\exp\left((\varphi_{m,n} - \varphi_{0,0}) \lambda t\right)\right]
\end{equation}
and for the BGFF,
\begin{equation}
 \mathbb{E}_{\text{BGFF}}\left[\exp\left((\varphi_{m,n} - \varphi_{0,0}) \lambda t\right)\right] = \frac{1}{Z}\int \mathcal{D}\varphi \exp\left(-\frac{1}{2 g^2}\sum_{\langle p, q \rangle} (\varphi_{p} - \varphi_{q} - J_{pq})^2 + \lambda t(\varphi_{m,n} - \varphi_{0,0})\right).
 \end{equation}
Using the standard formula for multivariate Gaussian integration
\begin{equation}
    \int d^n x \, \exp\left(-\frac{1}{2} \bm{x}^T \bm{A} \bm{x} + \bm{b}^T \bm{x}\right) = \sqrt{\frac{(2\pi)^n}{\det A}} e^{\frac{1}{2} \bm{b}^T \bm{A}^{-1} \bm{b}},
\end{equation}
an explicit computation gives
\begin{align}
    \mathbb{E}_{\text{BGFF}}\left[\exp\left((\varphi_{m,n} - \varphi_{0,0}) \lambda t\right)\right] &= \frac{\int \mathcal{D}\varphi \exp\left(-\frac{1}{2g^2}\sum_{\langle p, q \rangle} (\varphi_{p} - \varphi_{q} - J_{pq})^2 + t(\varphi_{m,n} - \varphi_{0,0})\right)}{\int \mathcal{D}\varphi \exp\left(-\frac{1}{2 g^2}\sum_{\langle p, q \rangle} (\varphi_{p} - \varphi_{q} - J_{pq})^2\right)} \nonumber\\
    &= \frac{\exp\left(\frac{1}{2} \left(\frac{\bm{\mathcal{J}}^T}{g^2} + \lambda t \bm{\delta}^T\right) \bm{K}^{-1} \left(\frac{\bm{\mathcal{J}}^T}{g^2} + \lambda t \bm{\delta}\right)\right)}{\exp\left(\frac{1}{2} \bm{\mathcal{J}}^T \bm{K}^{-1}  \bm{\mathcal{J}}\right)} \nonumber\\
    &= \exp\left(\frac{g^2 \lambda t}{2} \bm{\mathcal{J}}^T \bm{K}^{-1} \bm{\delta} + \frac{g^2 \lambda t}{2} \bm{\delta}^T \bm{K}^{-1} \bm{\mathcal{J}} + \frac{g^2 \lambda^2 t^2}{2} \bm{\delta}^T \bm{K}^{-1} \bm{\delta}\right),
\end{align}
where $\bm{\delta}$ is a vector which has an entry $1$ at $(m,n)$ and an entry $-1$ at $(0,0)$ and $\bm{\mathcal{J}}$ is a vector of source terms with $\mathcal{J}_{p} = 2 J_{p,p+(1,0)} - 2 J_{p,p-(1,0)} + 2 J_{p,p+(0,1)} - 2 J_{p,p-(0,1)}$ \footnote{Recall that the Laplacian of a square lattice has a zero eigenvalue corresponding to an eigenvector with uniform amplitudes at each site.  In the matrix $\bm{K}^{-1}$ we therefore ignore the contribution of this zero eigenvalue state.}.  Next, we consider the quantities $\mathbb{E}\left[\exp\left(-t \sum_{i = 0}^{m-1} \log r_{i+1/2, n}^{(1)}\right)\right]$ and $\mathbb{E}\left[\exp\left(-t \sum_{i = 0}^{n-1} \log r_{0, i+1/2}^{(2)}\right)\right]$ both of which are expressible in terms of the height fields $\chi_1$ and $\chi_2$:
\begin{align}
    \hspace{-2mm}\mathbb{E}\left[\exp\left(-\lambda t \sum_{i = 0}^{m-1} \log r_{i+1/2, n}^{(1)}\right)\right] &= \mathbb{E}_{\text{BGFF}}\left[\exp\left((\chi^{(1)}_{m+1/2,0} - \chi^{(1)}_{1/2,0}) \lambda t\right)\right] \nonumber \\ &= \exp\left(-\frac{g^2 \lambda t}{2} \bm{\mathcal{J}}^T \bm{K}^{-1} \bm{\varepsilon}_1 - \frac{g^2 \lambda t}{2} \bm{\varepsilon}_1^T \bm{K}^{-1} \bm{\mathcal{J}} + \frac{g^2 \lambda^2 t^2}{2} \bm{\varepsilon}_1^T \bm{K}^{-1} \bm{\varepsilon}_1\right)
\end{align}
where $\bm{\varepsilon}_1$ is a vector with an entry $1$ at $(m+1/2,0)$ and $-1$ at $(1/2,0)$.  Similarly, we find that
\begin{align}
    \mathbb{E}\left[\exp\left(-t \sum_{i = 0}^{n-1} \log r_{0, i+1/2}^{(2)}\right)\right] &= \mathbb{E}_{\text{BGFF}}\left[\exp\left((\chi^{(2)}_{m,n+1/2} - \chi^{(2)}_{m,1/2}) \lambda t\right)\right] \nonumber \\ &= \exp\left(-\frac{g^2 \lambda t}{2} \bm{\mathcal{J}}^T \bm{K}^{-1} \bm{\varepsilon}_2 - \frac{g^2 \lambda t}{2} \bm{\varepsilon}_2^T \bm{K}^{-1} \bm{\mathcal{J}} + \frac{g^2 \lambda^2 t^2}{2} \bm{\varepsilon}_2^T \bm{K}^{-1} \bm{\varepsilon}_2\right)
\end{align}
where $\bm{\varepsilon}_2$ is a vector with an entry $1$ at $(m,n+1/2)$ and $-1$ at $(m,1/2)$.  In particular, we see that $\bm{\varepsilon}_1 + \bm{\varepsilon}_2 = \bm{\delta}$.  Thus, multiplying the contributions from the 3 terms computed above, we find that $\mathbb{E}\left[\exp\left(\sum_{i = 0}^{m-1} \log r_{i+1/2, n}^{(1)} - \log r_{i, n}^{(1)} + \sum_{j = 0}^{n-1} \log r_{0, j+1/2}^{(2)} - \log r_{0, j}^{(2)}\right)\right]$ is:
\begin{equation}
    \mathbb{E}\left[\exp \left(\lambda t\sum_{i = 0}^{m-1} \log \frac{r_{i+1/2, n}^{(1)}}{r_{i, n}^{(1)}} + \lambda t \sum_{j = 0}^{n-1} \log \frac{r_{0, j+1/2}^{(2)}}{r_{0, j}^{(2)}}\right)\right] \leq \exp \left(\frac{g^2 \lambda^2 t^2}{2} \left(\bm{\delta}^T \bm{K}^{-1} \bm{\delta} + \bm{\varepsilon}_1^T \bm{K}^{-1} \bm{\varepsilon}_1 + \bm{\varepsilon}_2^T \bm{K}^{-1} \bm{\varepsilon}_2\right)\right).
\end{equation}

Next, we would like to upper bound the norm of $\bm{K}^{-1}$.  Note that $\bm{K}$ has matrix elements:
\begin{equation}
    K_{pq} = \frac{1}{g^2} \begin{cases}
D_p & \text{ if } p=q \\
-1 & \text{ if } |p-q|=1 \\
0 & \text{ o.w.}
\end{cases}
\end{equation}
which is the discrete Laplacian of a finite square grid ($D_p$ denotes the degree of node $p$).  For a square lattice of size $N \times N$, this Laplacian has eigenvectors
\begin{equation}
    \varphi_{m,n}(k, \ell) = \frac{2}{N} \cos \left(\frac{\pi k}{N}\left(m - \frac{1}{2}\right)\right) \cos \left(\frac{\pi \ell}{N}\left(n - \frac{1}{2}\right)\right),
\end{equation}
where $k$ and $\ell$ label the quantum numbers of this eigenstate and $m$ and $n$ indicate the position of the wavefunction in the lattice.  The corresponding energy eigenvalues on a square lattice are
\begin{equation}
    E(k, \ell) = 4 \sin^2\left(\frac{\pi k}{2 N}\right) + 4\sin^2\left(\frac{\pi \ell}{2 N}\right).
\end{equation}
Then, we may write the quadratic form as
\begin{align}
    \bm{\varepsilon}_2^T \bm{K}^{-1} \bm{\varepsilon}_2 &= \sum_{(k, \ell) \neq (0,0)} \frac{\braket{\delta}{k,\ell}\braket{k,\ell}{\delta}}{E(k,\ell)}, \nonumber\\
    &= \frac{4}{N^2}\sum_{(k, \ell) \neq (0,0)} \frac{\cos^2 \left(\frac{\pi \ell}{2N}(m-1)\right) \left(\cos \left(\frac{\pi \ell}{N}\left(n - \frac{1}{2}\right)\right) - \cos \left(\frac{\pi \ell}{2N}\right)\right)^2}{4 \sin^2\left(\frac{\pi k}{2 N}\right) + 4\sin^2\left(\frac{\pi \ell}{2 N}\right)}, \nonumber \\
    &\leq \frac{4}{N^2}\sum_{(k, \ell) \neq (0,0)} \frac{1}{\sin^2\left(\frac{\pi k}{2 N}\right) + \sin^2\left(\frac{\pi \ell}{2 N}\right)} \approx \frac{4}{\pi^2}\int_{\pi/(2N)}^{\pi/2} \int_{\pi/(2N)}^{\pi/2} \dd{x}\,dy\, \frac{4}{\sin^2 x + \sin^2 y},\nonumber\\
    &\leq \frac{4}{\pi^2} \int_{\pi/(2N)}^{\pi/2} \int_{\pi/(2N)}^{\pi/2} \dd{x}\,\dd{y}\, \frac{\pi^2}{x^2 + y^2} \leq \frac{4}{\pi^2} \int_{\pi/(2N)}^{\pi/2} r \,dr \int_{0}^{\pi/2} d\theta \, \frac{\pi^2}{r^2},\nonumber\\
    &= 2 \pi \log N
\end{align}
where in the third to last line, $x = \frac{\pi k}{2 N}$ and $y = \frac{\pi \ell}{2 N}$.  A similar bound holds when $\bm{\varepsilon}_2$ is replaced with $\bm{\varepsilon}_1$ or $\bm{\delta}$.  Then, a Chernoff bound gives
\begin{equation}
    \mathbb{P}\left(\sum_{i = 0}^{m-1} \log \frac{r_{i+1/2, n}^{(1)}}{r_{i, n}^{(1)}} + \sum_{j = 0}^{n-1} \log \frac{r_{0, j+1/2}^{(2)}}{r_{0, j}^{(2)}} > a\right) \leq e^{-at} e^{-(K'/2) t^2 \log N}
\end{equation}
for some constant $K'$.  Choosing $t = a/(K' \log N)$, we find that
\begin{equation}
    \mathbb{P}\left(\sum_{i = 0}^{m-1} \log \frac{r_{i+1/2, n}^{(1)}}{r_{i, n}^{(1)}} + \sum_{j = 0}^{n-1} \log \frac{r_{0, j+1/2}^{(2)}}{r_{0, j}^{(2)}} > a\right) \leq e^{-\frac{a^2}{2 K' \log N}},
\end{equation}
and with $a = c \log N$, the probability is less than $1/N^{p}$ for some constant $p = c^2/(2K')$.  This means that with high probability, $\exp\left(\sum_{i = 0}^{m-1} \log r_{i+1/2, n}^{(1)} - \sum_{j = 0}^{n}  \log r_{i, n}^{(1)}\right)$ is at most polynomial in $N$.
\end{proof}

To develop an intuitive understanding of the result above, we perform a rough computation of what would occur on a 1D lattice.  There, the quadratic form would have a similar structure, except that
\begin{equation}
    \bm{\varepsilon}_2^T \bm{K}^{-1} \bm{\varepsilon}_2 \leq \frac{2}{N}\sum_{k, \ell = 1}^{N} \frac{1}{\sin^2\left(\frac{\pi k}{2 N}\right)} \approx 2 \int_{1/N}^1 \dd{x}\, \frac{1}{\sin^2\left(\frac{\pi}{2} x\right)}.
\end{equation}
The integral above is linear in $N$, so this suggests in 1D $\sum_{i = 0}^{m-1} \log r_{i+1/2, n}^{(1)} - \sum_{j = 0}^{n}  \log r_{i, n}^{(1)} > \sqrt{N \log N}$ with high probability, consistent with what we found in the 1D case.  In fact, $d=2$ is a critical dimension where the speedup changes from superpolynomial to exponential.  When $d \geq 3$, the integrals above no longer depend on $N$.  The fact that $d=2$ is a special dimension has to do with the fact that in a free scalar field theory, correlation functions grow logarithmically with distance in $d=2$, and decay with distance in $d \geq 3$.  In particular, for a free scalar field theory:
\begin{equation}
    \langle \varphi(x) \varphi(y) \rangle = \frac{1}{Z}\int \mathcal{D}\varphi \, \varphi(x) \varphi(y) \exp\left(- \int \dd[d]{x}\, (\nabla \varphi)^2\right) \sim \frac{1}{|x-y|^{d-2}}.
\end{equation}
When $d = 2$ in the above expression, the correlation function receives a logarithmic correction, and this correction does not appear for any other $d$.  Thus, the lack of correlation decay in low dimensions is directly tied to the nature of the superpolynomial speedup.

In the original random graph model that we consider, the logarithms of the edge-edge ratios are biased; as a result, the graph has an expansion property where the number of nodes in the middle becomes exponentially large.  However, the quantity that enters into the amplitude of the zero mode, as in the 1D case, is a sum over differences of adjacent edge-edge ratio logarithms.  These quantities have zero expectation value and are unbiased, and therefore the zero mode amplitudes are dominated by fluctuations about the bias, which is governed by a Gaussian free field.

Next, we will use these concentration properties to explicitly show that the zero mode has polynomial decay across the square lattice:
\begin{lemma}\label{lem:polydecay}
With high probability, the product of the amplitudes of being at the lower left corner and the upper right corner of $G$ satisfy $|\psi_{N,N}||\psi_{0,0}| \geq \frac{1}{K N^{3p+1}}$ for constants $K$ and $p$.
\end{lemma}
\begin{proof}
\lem{2D-Laplacian} shows that with probability $> 1-1/N^{p'}$:
\begin{equation}
    K N^p |\psi_{0,0}| \geq |\psi_{m,n}| \geq \frac{K |\psi_{0,0}|}{N^p}
\end{equation}
for constants $K$, $p$, and $p'$.  By a union bound, the probability that all $\psi_{m,n}$ satisfy this is $1/N^{p'-2}$.  The normalization condition $\sum_{m,n} |\psi_{m,n}|^2 = 1$ gives a bound on $|\psi_{0,0}|$:
\begin{equation}
    |\psi_{0,0}| \geq \frac{1}{K N^{p+1}},
\end{equation}
and thus with high probability $|\psi_{N,N}||\psi_{0,0}| \geq \frac{1}{K N^{3p+1}}$, which decays inverse polynomially in $N$.
\end{proof}

Having shown that the zero mode is not exponentially localized, we have to argue that it is sufficiently gapped from the rest of the spectrum.  This requires that we prove that the spectral gap to the zero mode decays inverse polynomially.  Recall from \lem{zeromodes}, there are $N^2-2N$ zero modes with support on sublattice $\Delta_1$ and one zero mode with support on sublattice $\Delta_0$; thus, $(N-1)^2$ zero modes in total.  The following generalization of the eigenvalue interlacing theorem will be useful to us:
\begin{theorem}[Eigenvalue interlacing theorem]
Let $A$ be an $N\times N$ symmetric matrix, and let $B$ be an $M \times M$ symmetric matrix formed by deleting some number of columns and corresponding rows of $A$.  If $A$ has eigenvalues $\lambda_1 \leq \lambda_2 \leq \cdots \leq \lambda_N$ and $B$ has eigenvalues $\mu_1 \leq \mu_2 \leq \cdots \leq \mu_M$, then
\begin{equation}
    \lambda_k \leq \mu_k \leq \lambda_{k+N-M}
\end{equation}
where $k$ ranges from $1$ to $M$.
\end{theorem}

Next, we shall apply the eigenvalue interlacing theorem to the decorated square lattice.  When we delete a row and column $i$ from the Hamiltonian, this corresponds to deleting a node in the square lattice.  Thus, if we delete $(N-1)^2+1$ nodes and maintain that the size of the $\Delta_1$ and $\Delta_0$ sublattices are the same, then the subsequent Hamiltonian will not have any zero modes.  The nodes we choose to delete reside on the blue edges shown in \fig{nodeelimination}; there are $(N-1)^2$ such blue edges, and the corresponding graph is equivalent to a 1D ``snake graph'' of length $2N^2-1$.  We delete a final node at the end of the line so the graph now has $2N^2-2 \in 2\mathbb{Z}$ nodes.  Calling the snake graph $G'$, the eigenvalue interlacing theorem gives the following result:
\begin{lemma}
The minimum eigenvalue in absolute value $\nu$ of the Hamiltonian associated with $G'$ satisfies $|\nu| \leq \Delta$, where $\Delta$ is the gap in $G$.
\end{lemma}
\begin{proof}
We have removed $(N-1)^2+1$ nodes from $G$.  Call $\alpha$ the negative eigenvalue closest to zero and $\beta$ the positive eigenvalue closest to zero.  Since between $\alpha$ and $\beta$ there are $(N-1)^2$ zero eigenvalues, the eigenvalue interlacing theorem tells us that there exists an eigenvalue $\mu$ and $\mu'$ of $G'$ satisfying $0 < \mu < \beta$ and $\alpha < \mu' < 0$.  Finally, as the Hamiltonian for $G'$ has no zero eigenvalues, $\nu = \min(|\mu|, |\mu'|)$ is the smallest eigenvalue in absolute value.
\end{proof}

\begin{figure}
    \centering
    \includegraphics[scale=0.4]{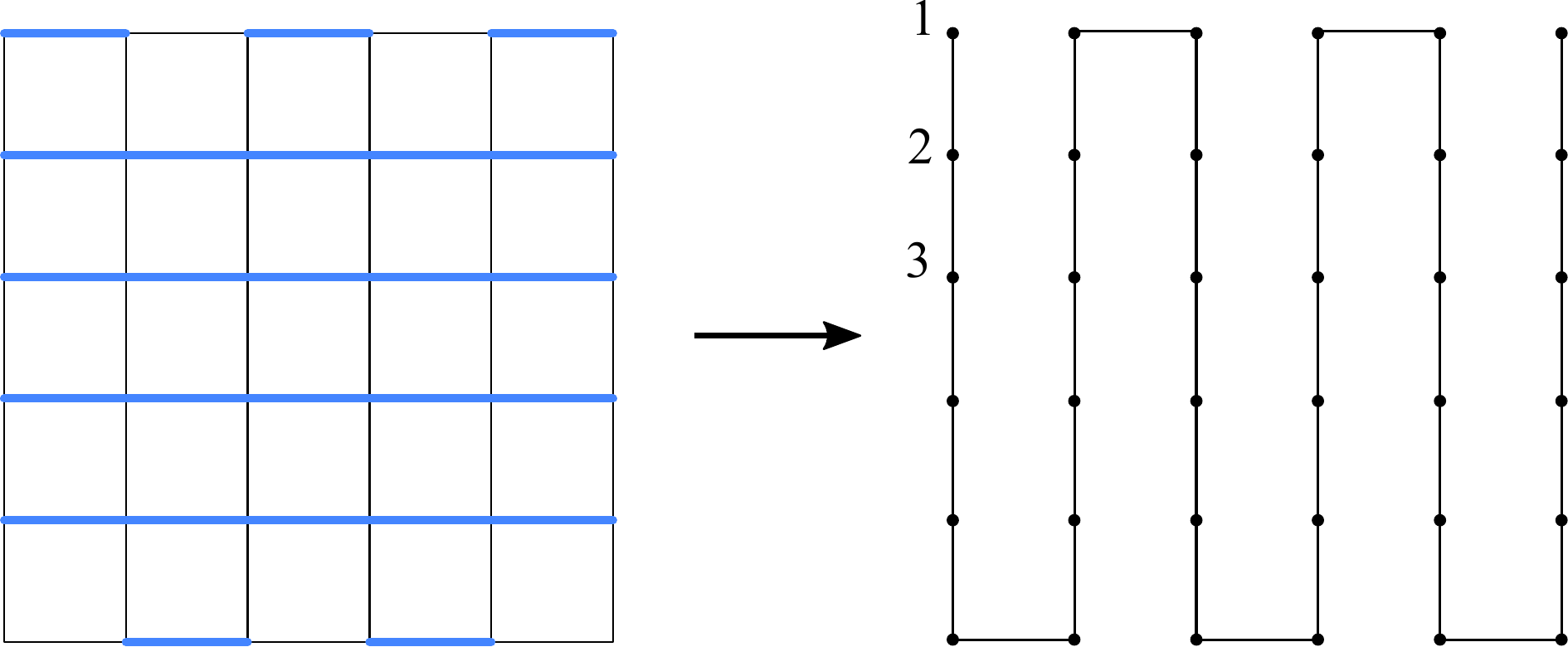}
    \caption{The nodes removed are indicated by the blue edges, and the corresponding snake graph is shown on the right.}
    \label{fig:nodeelimination}
\end{figure}

The Kotowski-Vir{\'a}g theorem cannot be applied to a snake graph because there are large correlations between points far apart if the snake graph is unwound into a one-dimensional line.  However, the resulting Hamiltonian for the snake graph is exactly invertible, so we may modify the analysis from \sec{Haminv}.  Thus, we are required to bound the magnitude of $\prod_{k \in \ell} \frac{t_{\ell(2k), \ell(2k+1)}}{t_{\ell(2k+1), \ell(2k+2)}}$, where $\ell$ is a directed path along the line.  If $\ell$ is a directed path from $(a,b)$ to $(c,d)$, then
\begin{equation}
    \prod_{k \in \ell[(a,b) \to (c,d)]} \frac{t_{\ell(2k), \ell(2k+1)}}{t_{\ell(2k+1), \ell(2k+2)}} = \frac{\psi_{c,d}}{\psi_{a,b}}.
\end{equation}
Prior lemmas from this section establish that the right hand side of this expression has magnitude $\text{poly}(N)$ with probability $> 1-N^{-p'}$ for some $p' > 0$.

We also need to bound the magnitude of an individual $t_{p,q}^{-1} = \sqrt{s_p s_q}/e_{p,q}$ where $p$ and $q$ are coordinates of sites on the lattice.  Suppose $p + (1,0) = q$ (a similar argument will hold for the remaining three cases where $p + (-1,0) = q$, $p + (0,1) = q$, and $p + (0,-1) = q$).  We may expand this and write it as
\begin{align}
    t_{p,q}^{-1} = \frac{1}{D e_{p,p+(1,0)}}&\sqrt{e_{p,p+(1,0)}^2\left[1 + r^{-1}_{p-(1,0) \to p} + \frac{e_{p,p+(0,1)}}{e_{p,p+(1,0)}}\cdot (1 + r^{-1}_{p \to p-(0,1)})\right]}\nonumber\\
    &\cdot\sqrt{\left[1 + r_{p \to p+(1,0)} + \frac{e_{p+(1,0),p+(1,1)}}{e_{p,p+(1,0)}}\cdot (1 + r^{-1}_{p+(1,-1) \to p+(1,0)})\right]}.
\end{align}
Each term in the square root is an edge-edge ratio which is upper and lower bounded by a constant in the hierarchical graph.  Thus, $t_{p,q}^{-1}$ is bounded by a constant.


This allows us to arrive at the following lemma:
\begin{lemma}\label{lem:2D-gap}
The gap of $H$, defined as the smallest nonzero energy eigenvalue in absolute value, is greater than $C/N^p$ for suitable constants $C$ and $p$.
\end{lemma}
\begin{proof}
First, label nodes of the one-dimensional line with integers (which we will denote by $i$, $j$, or $k$) and define $t_k \triangleq t_{k \to k+1}$. Recall, from before, that the norm of the inverse of a 1D Hamiltonian satisfies the equality
\begin{equation}
    \norm{{H'}^{-1}}_{1,\infty} = \max\left[\max_{i \text{ odd}}  \sum_{j}\frac{1}{t_i} \prod_{k = (i+1)/2}^{(j-2)/2}\left(\frac{t_{2k}}{t_{2k+1}}\right) \genfrac{}{}{0pt}{}{}{,} \max_{j \text{ even}}  \sum_{i}\frac{1}{t_i} \prod_{k = (i+1)/2}^{(j-2)/2}\left(\frac{t_{2k}}{t_{2k+1}}\right)\right].
\end{equation}
We upper bound the maxima with sums.  Each term in the sum, by the arguments above is upper bounded by a polynomial in $N$ with probability $1 - N^{-p}$.  If $p$ is large enough, the sum is polynomial in $N$ with high probability.  Therefore
\begin{equation}
    \Delta \geq \mu \geq \frac{1}{\sqrt{2N} \norm{{H'}^{-1}}_{1,\infty}} \geq \frac{1}{\text{poly}(N)}
\end{equation}
with high probability.
\end{proof}

Finally, this gives us the following theorem, in conjunction with \lem{polydecay} and Eqn.~\eqn{time2D}:
\begin{theorem}\label{thm:2D-main}
For the random graph ensemble with gauge condition in \defn{gauge} and free parameters chosen in Eqn.~\eqn{free-parameter-choice} following the BsGFF in \defn{BsGFF}, with probability at least $3/4$, the graph selected from the ensemble admits a quantum algorithm using $\poly(N)$ queries to the quantum oracle in \defn{oracle} to find the exit node with probability at least $3/4$.
\end{theorem}

\subsection{Lieb lattice in higher dimensions}
An advantage of the models described in the previous section is that they may easily be extended to higher dimensions.  In particular, we introduce a $d$-dimensional Lieb lattice, which is a $d$-dimensional cubic lattice with an additional site on each edge.  As before, the undecorated lattice needs to be bipartite in order to construct a zero mode, and in this section we will restrict ourselves to the case where the undecorated lattice is a $d$-dimensional cubic lattice (but note that our results can be generalized to an arbitrary undecorated bipartite lattice).

As before, the gauge condition implies that the zero mode $\psi_i$ is exact, written as the lattice gradient of a scalar potential.  The zero mode can be written as in Eqn.~(\ref{eqn:2Dzeromode}), except the path $\ell$ occurs in a $d$-dimensional manifold.  As before, we can precisely formulate the number of zero modes in this model:
\begin{lemma}
On a $d$-dimensional Lieb lattice, the number of zero modes is $(d-1) N^d - d N^{d-1} + 2$: there are $(d-1) N^d - d N^{d-1} + 1$ zero modes with support in sublattice $\Delta_1$ and one zero mode with support in sublattice $\Delta_0$.
\end{lemma}
\begin{proof}
The proof follows from the analysis of the last section, noting that there are $N^d$ nodes in sublattice $\Delta_0$ and $d N^{d-1}(N-1)$ nodes in sublattice $\Delta_1$.
\end{proof}
Next, we make an assumption that the height fields obey a biased GFF: to extend to $d$ dimensions, we need to introduce $d$ additional height fields $\chi_i$ ($i = 1,\cdots,d$) associated with the edge sites .  As in the 2D case, the zero mode can be exactly written as
\begin{equation}
\psi_{m_1,m_2,\ldots, m_d} = \psi_{0,0,\ldots, 0} \sqrt{\frac{s_{m_1,m_2,\ldots, m_d}}{s_{0,0,\ldots, 0}}} \prod_{k \in \ell_{m,n}}^{|\ell|} \frac{e_{\ell(2k), \ell(2k+1)}}{e_{\ell(2k+1), \ell(2k+2)}}.
\end{equation}
Due to the gauge condition, we are free to choose any path from $(0,\ldots,0)$ to $(m_1, \ldots, m_d)$.  Generalizing the path chosen in the 2D case, we choose $(0,0,\ldots,0) \to (m_1,0,\ldots,0) \to (m_1,m_2,\ldots,0) \to \ldots \to (m_1,m_2,\ldots,m_d)$.  In terms of the edge ratio variables, the zero mode can be expressed as
\begin{equation}
\frac{\psi_{m_1,m_2,\ldots, m_d}}{\psi_{0,0,\ldots, 0}} = \sqrt{\frac{s_{m_1,m_2,\ldots, m_d}}{s_{0,0,\ldots, 0}}} \prod_{i = 0}^{m_1-1} \frac{1}{r_{i,0,0,\ldots,0}^{(1)}} \prod_{j = 0}^{m_2-1} \frac{1}{r_{m_1,j,0,\ldots,0}^{(2)}}\prod_{k = 0}^{m_3-1} \frac{1}{r_{m_1,m_2,k,\ldots,0}^{(3)}} \ldots \prod_{q = 0}^{m_d-1} \frac{1}{r_{m_1,m_2,m_3,\ldots, q}^{(d)}}.
\end{equation}
By gauge invariance (i.e. path independence due to the gauge constraint), this is not the unique way to write the product above.  We may write the number of nodes in each supernode as
\begin{align}
    s_{m_1, m_2, \ldots, m_d} &= \frac{1}{D}\sum_{i=1}^d e_{m_1, m_2, \ldots, m_d}^{(i)} (1 + r^{(i)}_{m_1, m_2, \ldots, m_d}) \\
    &= \frac{1}{D}\sum_{i=1}^d (1 + r^{(i)}_{m_1, m_2, \ldots, m_d}) \prod_{i_1 = 1}^{2m_1} r^{(1)}_{i_1/2,0,\ldots,0} \frac{e^{(2)}_{m_1,0,\ldots,0}}{e^{(1)}_{m_1,0,\ldots,0}} \prod_{i_2 = 1}^{2m_2} r^{(2)}_{m_1,i_2/2,\ldots,0} \frac{e^{(3)}_{m_1,m_2,\ldots,0}}{e^{(2)}_{m_1,m_2,\ldots,0}}\ldots \nonumber\\&\hspace{7cm}\cdot \prod_{i_d = 1}^{2m_d} r^{(d)}_{m_1,m_2,\ldots,i_d/2} \frac{e^{(i)}_{m_1,m_2,\ldots,m_d}}{e^{(d)}_{m_1,m_2,\ldots,m_d}} .
\end{align}
The full amplitude of the zero mode is then given by:
\begin{align}
    \psi^2_{m_1, m_2, \ldots, m_d} &= \frac{\psi^2_{0,0,\ldots,0}}{D s_{0,0}} \sum_{k=1}^d(1+r^{(k)}_{m_1, m_2, \ldots, m_d}) \frac{e^{(2)}_{m_1,0,\ldots,0}}{e^{(1)}_{m_1,0,\ldots,0}} \frac{e^{(3)}_{m_1,m_2,\ldots,0}}{e^{(2)}_{m_1,m_2,\ldots,0}} \ldots \frac{e^{(k)}_{m_1,m_2,\ldots,m_d}}{e^{(d)}_{m_1,m_2,\ldots,m_d}}
    \nonumber \\ &\cdot  \frac{\prod_{i=0}^{m_1-1}  r^{(1)}_{i+1/2,0,\ldots,0}}{\prod_{i=0}^{m_1-1}  r^{(1)}_{i,0,\ldots,0}}\frac{\prod_{j=0}^{m_2-1}  r^{(2)}_{m_1, j+1/2,\ldots,0}}{\prod_{j=0}^{m_2-1}  r^{(2)}_{m_1,j,\ldots,0}}\ldots \frac{\prod_{q=0}^{m_d-1}  r^{(d)}_{m_1, m_2,\ldots,q+1/2}}{\prod_{j=0}^{m_d-1}  r^{(d)}_{m_1,m_2,\ldots,q}}.
\end{align}
We now pick appropriate boundary conditions on the hypercube, analogously to the 2D case.  The boundary of the hypercube is defined to be the set of points where at least one of $m_1, m_2, \ldots, m_d $ is $0$.  We choose the values of $e^{(i)}$ for boundary sites (not including boundary edge sites) to be the same for all $i$.
\begin{equation}
   e^{(1)}_{m_1,m_2,\ldots,m_d} = e^{(2)}_{m_1,m_2,\ldots,m_d} = \ldots = e^{(d)}_{m_1,\ldots,m_d} \,\, \forall (m_1,m_2,\ldots,m_d), \,\, \prod_i m_i = 0.
\end{equation}
Again this choice was somewhat arbitrary although our requirement that the underlying graph be regular implies an upper bound the possible differences between the $e^{(i)}$'s.

In fact, throughout the graph, the bounded degree implies that $e^{(i)}_{\vec{r}}/e^{(j)}_{\vec{r}}$ is bounded.  These assumptions simplify the above equation to
\begin{align}
    \psi^2_{m_1, m_2, \ldots, m_d} & \geq
    \frac{C \psi^2_{0,0,\ldots,0}}{D^2} \frac{\prod_{i=0}^{m_1-1}  r^{(1)}_{i+1/2,0,\ldots,0}}{\prod_{i=0}^{m_1-1}  r^{(1)}_{i,0,\ldots,0}}\frac{\prod_{j=0}^{m_2-1}  r^{(2)}_{m_1, j+1/2,\ldots,0}}{\prod_{j=0}^{m_2-1}  r^{(2)}_{m_1,j,\ldots,0}}\ldots \frac{\prod_{q=0}^{m_d-1}  r^{(d)}_{m_1, m_2,\ldots,q+1/2}}{\prod_{j=0}^{m_d-1}  r^{(d)}_{m_1,m_2,\ldots,q}}.
    \\
    & \triangleq \frac{C \psi^2_{0,0,\ldots,0}}{D^2} \exp(\Phi_{m_1,\cdots,m_d})
\end{align}
for some constant $C$.  The localization properties of $\psi$, and thus the quantum runtime, will depend exponentially on $\Phi$.

As before, we can introduce height fields $\varphi$ and $\chi^{(i)}$ where $i = 1,\cdots,d$, analogous the 2D case.  Similarly, we can write $\Phi$ in terms of these height fields and quantify the concentration of $\log \Phi$, resulting in the following Lemma:
\begin{lemma}\label{lem:highD}
With high probability, $\sum_{a = 0}^{n} \log r_{a, m_2, \ldots, m_d}^{(1)} - \sum_{a = 0}^{n} \log r_{a+1/2, m_2, \ldots, m_d}^{(1)}=O(\sqrt{\log N})$ assuming that the dimension $d$ and the size of the lattice $N$ satisfy $d = O(N)$.
\end{lemma}
\begin{proof}
First we use the sub-Gaussian condition to bound
\begin{equation}
    \mathbb{E}_{\text{BsGFF}}\left[\exp\left( t \sum_{a=0}^{n}  \log \frac{r_{a, m_2, \ldots, m_d}^{(1)}}{r_{a+1/2, m_2, \ldots, m_d}^{(1)}}\right)\right] \leq \mathbb{E}_{\text{BGFF}}\left[\exp\left( \lambda t  \sum_{a=0}^{n}  \log \frac{r_{a, m_2, \ldots, m_d}^{(1)}}{r_{a+1/2, m_2, \ldots, m_d}^{(1)}}\right)\right].
\end{equation}
for some $\lambda$.  Under a biased Gaussian free field, we may compute the moment generating function
\begin{equation}
    \mathbb{E}\left[\exp\left(t \sum_{a=0}^{n}  \log r_{a, m_2, \ldots, m_d}^{(1)}\right)\right] = \frac{1}{Z}\int \mathcal{D}\varphi \exp\left(-\frac{1}{2 g^2}\sum_{\langle p, q \rangle} (\varphi_{p} - \varphi_{q} - J_{pq})^2 + t(\varphi_{n,m_2,m_3,\ldots} - \varphi_{0,m_2,m_3,\ldots})\right).
\end{equation}
We use the standard formula for multivariate Gaussian integration as before to evaluate the above quantity.  We also compute the contribution from the half integer height field $\chi^{(1)}$ (which we abbreviate with $\chi$ for brevity):
\begin{equation}
    \mathbb{E}\left[\exp\left(t \sum_{a=0}^{n}  \log r_{a+1/2, m_2, \ldots, m_d}^{(1)}\right)\right] = \frac{1}{Z}\int \mathcal{D}\chi \exp\left(-\frac{1}{2 g^2}\sum_{\langle p, q \rangle} (\chi_{p} - \chi_{q} - J_{pq})^2 + t(\chi_{n,m_2,m_3,\ldots} - \chi_{0,m_2,m_3,\ldots})\right).
\end{equation}
Dividing both terms to compute the expectation value, all terms linear in $t$ in the argument of the exponential cancel, and we get
\begin{equation}
    \mathbb{E}\left[\exp\left(\lambda t \sum_{a = 0}^{n} \log r_{a, m_2, \ldots, m_d}^{(1)} - \lambda t\sum_{a = 0}^{n} \log r_{a+1/2, m_2, \ldots, m_d}^{(1)}\right)\right] = \exp \left(g^2 \lambda^2 t^2 \bm{\varepsilon}_1^T \bm{K}^{-1} \bm{\varepsilon}_1\right),
\end{equation}
where $K$ is the Laplacian of a $d$-dimensional square lattice, and as in the 2D case, $\bm{\varepsilon}_1$ is a vector with components $1$ at $(n, m_2, \ldots, m_d)$ and $-1$ at $(0, m_2, \ldots, m_d)$.  Note that we require $d-1$ more such terms (corresponding to the other height fields $\chi^{(i)}$ for $i = 2, 3, \cdots, d$) in order the construct the full path, but we defer this consideration to the next Lemma.  Next, we would like to upper bound the norm of $\bm{K}^{-1}$.  We note that this Laplacian has eigenvectors
\begin{equation}
    \varphi_{m_1,m_2,\ldots,m_d}(k_1, k_2, \ldots, k_d) =  \prod_{i = 1}^d \left(\frac{A_i}{N}\right)^{1/2} \cos \left(\frac{\pi k_i}{N}\left(m_i - \frac{1}{2}\right)\right),
\end{equation}
with corresponding energy eigenvalues
\begin{equation}
    E(k_1, k_2, \ldots, k_d) = 4 \sum_{i=1}^d \sin^2\left(\frac{\pi k_i}{2 N}\right),
\end{equation}
and $A$ is a normalization constant given by
\begin{equation}
    A_i^{-1} = \frac{1}{N}\sum_{m_i = 1}^N \cos^2 \left(\frac{\pi k}{N}\left(m_i - \frac{1}{2}\right)\right) = \frac{1}{2}.
\end{equation}
The quadratic form can be bounded by\footnote{Similarly, we only sum over momenta satisfying $\vec{k} \neq \vec{0}$.}
\allowdisplaybreaks
\begin{align}
    \bm{\varepsilon}_1^T \bm{K}^{-1} \bm{\varepsilon}_1 &= \sum_{k_1, k_2, \ldots, k_d = 0}^{N} \frac{\braket{\varepsilon_1}{k_1, k_2, \ldots, k_d}\braket{k_1, k_2, \ldots, k_d}{\varepsilon_1}}{E(k_1, k_2, \ldots, k_d)} \nonumber\\
    &= \frac{1}{N^d}\sum_{k_1, k_2, \ldots, k_d = 0}^{N} \frac{\prod_{i=2}^d 2 \cos^2 \left(\frac{\pi k_i}{N} \left(m_i - \frac{1}{2}\right)\right)\left(\cos \left(\frac{\pi k_1}{N}\left(n - \frac{1}{2}\right)\right) - \cos \left(\frac{\pi k_1}{2N}\right)\right)^2}{4 \sum_{i=1}^d \sin^2\left(\frac{\pi k_i}{2 N}\right)}.
\end{align}
At this point, we will proceed by using the AM-GM inequality in the denominator.  One has to be careful however, because we have to split between the cases when $k_i = 0$ and when $k_i \neq 0$.  Let us consider the case when the number of $k_i$ which are nonzero is $q-1$ (WLOG we take these to be $i = 2,3,\ldots,q$).  Therefore, application of AM-GM gives
\begin{align}
    \frac{1}{4 N^d} \frac{\prod_{i=2}^d 2 \cos^2 \left(\frac{\pi k_i}{N} \left(m_i - \frac{1}{2}\right)\right)\left(\cos \left(\frac{\pi k_1}{N}\left(n - \frac{1}{2}\right)\right) - \cos \left(\frac{\pi k_1}{2N}\right)\right)^2}{\sum_{i=1}^d \sin^2\left(\frac{\pi k_i}{2 N}\right)}\nonumber \leq \\ \frac{1}{4 q N^d} \frac{\prod_{i=2}^d 2 \cos^2 \left(\frac{\pi k_i}{N} \left(m_i - \frac{1}{2}\right)\right)\left(\cos \left(\frac{\pi k_1}{N}\left(n - \frac{1}{2}\right)\right) - \cos \left(\frac{\pi k_1}{2N}\right)\right)^2}{\prod_{i=1}^q \sin^{2/q}\left(\frac{\pi k_i}{2 N}\right)}.
\end{align}
We first deal with the sum over $k_1$.  Notice that the $k_1 = 0$ case will never occur because the summand will always evaluate to $0$.  Therefore, in the above expression (once summed over $\vec{k}$) we may always factor out the term
\begin{equation}
    \frac{1}{N} \sum_{k_1=1}^N \frac{\left(\cos \left(\frac{\pi k_1}{N}\left(n - \frac{1}{2}\right)\right) - \cos \left(\frac{\pi k_1}{2N}\right)\right)^2}{\sin^{2/q}\left(\frac{\pi k_1}{2 N}\right)} \lesssim 4 \int_{1/N}^1 \frac{\dd{x}}{\sin^{2/q}\left(\frac{\pi x}{2}\right)} \triangleq 4 I_q.
\end{equation}
Having decoupled this term from the full expression at large, the remaining part of the expression can be regrouped in the following suggestive form:
\begin{align}
     \frac{4I_q}{4 q N^{d-1}} \frac{\prod_{i=2}^d 2 \cos^2 \left(\frac{\pi k_i}{N} \left(m_i - \frac{1}{2}\right)\right)}{ \prod_{i=2}^q \sin^{2/q}\left(\frac{\pi k_i}{2 N}\right)}  = \frac{ I_q}{q} \prod_{i = 2}^{q} \frac{2 \cos^2 \left(\frac{\pi k_i}{N} \left(m_i - \frac{1}{2}\right)\right)}{N \sin^{2/q}\left(\frac{\pi k_i}{2 N}\right)} \prod_{j = q+1}^d \frac{2}{N} \cos^2 \left(\frac{\pi k_i}{N} \left(m_i - \frac{1}{2}\right)\right).
\end{align}
The contribution to the original sum corresponding to fixing the $d - q$ values of $k_i$ to $0$ while summing over all possible nonzero values for the other $k_i$ gives:
\begin{align}
     \frac{I_q}{q} \left(\frac{2}{N}\right)^{d-q}\prod_{i = 2}^{q} \sum_{k_i = 1}^{N} \frac{2 \cos^2 \left(\frac{\pi k_i}{N} \left(m_i - \frac{1}{2}\right)\right)}{N \sin^{2/q}\left(\frac{\pi k_i}{2 N}\right)}.
\end{align}
Depending on the value of $\frac{2}{q}\log \sin \left(\frac{\pi k}{2N}\right)$, the remaining sum will be bounded differently.  Suppose we work in a large dimension limit, where $d \geq \log N$, and suppose that $q \geq \log N$ as well.  Then, to leading order, we find
\begin{equation}\label{eq:sinappr}
   \sin^{2/q} \left(\frac{\pi k}{2N}\right) \approx 1 + \frac{2}{q}\log \sin \left(\frac{\pi k}{2N}\right) + \cdots
\end{equation}
where $\cdots$ include higher order terms.  Ignoring the prefactor $\frac{I_q}{q} \left(\frac{2}{N}\right)^{d-q}$ for now, we find that
\begin{align}
\prod_{i = 2}^{q} \sum_{k_i = 1}^{N} \frac{2 \cos^2 \left(\frac{\pi k_i}{N} \left(m_i - \frac{1}{2}\right)\right)}{N \sin^{2/q}\left(\frac{\pi k_i}{2 N}\right)} \,
&\lesssim \, \prod_{i = 2}^{q} \frac{2}{N} \sum_{k_i = 1}^{N} \cos^2 \left(\frac{\pi k_i}{N} \left(m_i - \frac{1}{2}\right)\right)\left(1 - \frac{2}{q} \log \sin \left(\frac{\pi k_i}{2N}\right)\right) \nonumber \\
&\lesssim \, \left(1 - \frac{1}{N} + \frac{4 \log 2}{q} + \cdots\right)^{q-1} \nonumber
\\
&\lesssim \,  1. 
\end{align}
where from the first to the second inequality we used the fact that
\begin{equation}
    -\frac{1}{N} \sum_{k = 1}^N \cos^2\left(\frac{\pi k}{N} \left(m_i - \frac{1}{2}\right)\right) \log \sin \left(\frac{\pi k}{2N}\right)
    \, \lesssim -\, \int_{1/N}^1 \dd{x}\, \log \sin \frac{\pi x}{2} = \log 2
\end{equation}
with the last integral is evaluated using standard tricks.  Next, let us consider the case when $d \geq \log N$ and $q \leq \log N$.  In this case, we cannot rely on the expansion used in Eqn.~\eq{sinappr}.  Instead, we may cheaply bound the following sum:
\begin{equation}
   \sum_{k_i = 1}^{N} \frac{2 \cos^2 \left(\frac{\pi k_i}{N} \left(m_i - \frac{1}{2}\right)\right)}{N \sin^{2/q}\left(\frac{\pi k_i}{2 N}\right)} \leq \frac{2}{N} \sum_{k = 1}^N \frac{1}{\sin^{2/q}\left(\frac{\pi k}{2 N}\right)} \lesssim \int_{1/N}^1 \dd{x}\, \frac{2}{\sin^{2/q} \frac{\pi x}{2}}.
\end{equation}
This integral is simply $2 I_q$.  When $q = 1$ this integral is proportional to $N$ and when $q = 2$ this integral is proportional to $\log N$.  For $q > 2$, the integral is bounded by a positive constant.  Still considering the case where $d \geq \log N$ and $q \leq \log N$, we sum over all possible configurations of $\vec{k}$ where the number of nonzero elements of the vector is $q \leq\log N$.  Call this quantity $W_{\leq \log N}$, which we may bound by (reintroducing the prefactor of $I_q/q$):
\begin{align}
   W_{\leq \log N} &\lesssim \sum_{q = 1}^{O(\log N)} \binom{d-1}{q-1} \left(\frac{2}{N} \right)^{d-q} \frac{(I_q)^{q}}{q} \nonumber\\
   &\lesssim \frac{N^2\, 2^d}{N^{d}} + \frac{N^2 d\, 2^d \log^2 N}{N^{d}} + \sum_{q = 3}^{O(\log N)} \binom{d-1}{q-1} \left(\frac{2}{N} \right)^{d-q} c^q.
\end{align}
Making the further assumption that $d = O(N)$, we see that all three terms in the sum above rapidly decay to zero as $d,N \to \infty$.  Next, the contribution to the sum when $q > \log N$ is
\begin{align}
   W_{> \log N} &\lesssim \sum_{q = 1}^{d} \binom{d-1}{q-1} \left(\frac{2}{N} \right)^{d-q} 
   \frac{I_q}{q} \nonumber\\
   &\lesssim \frac{N^2}{N^d} + \frac{N^2\,\log N}{N^d}  + \sum_{q = 1}^{d} \binom{d-1}{q-1} \left(\frac{2}{N} \right)^{d-q} \nonumber \\
   &\lesssim \left(1 + \frac{2}{N}\right)^{d-1} \leq \exp(2d/N).
\end{align}
The original sum is therefore given by
\begin{align}
   \bm{\varepsilon}_1^T \bm{K}^{-1} \bm{\varepsilon}_1 &= \sum_{k_1, k_2, \ldots, k_d = 0}^{N} \frac{\braket{\varepsilon_1}{k_1, k_2, \ldots, k_d}\braket{k_1, k_2, \ldots, k_d}{\varepsilon_1}}{E(k_1, k_2, \ldots, k_d)} \nonumber\\
    &\leq W_{\leq \log N} + W_{> \log N} \nonumber \\
    &\lesssim \exp(2d/N)
\end{align}
which is upper bounded by a constant when $d = O(N)$.  Then, a Chernoff bound gives
\begin{equation}
    \mathbb{P}\left( \sum_{a = 0}^{n} \log r_{a, m_2, \ldots, m_d}^{(1)} - \sum_{a = 0}^{n} \log r_{a+1/2, m_2, \ldots, m_d}^{(1)} > a\right) \leq e^{-at} e^{(g\lambda t)^2 (\bm{\epsilon}_1^T \bm{K}^{-1} \bm{\epsilon}_1)^2/2} \leq e^{-at} e^{-K' t^2/2}
\end{equation}
for some constant $K'$.  Choosing $t = a/K'$, we find that
\begin{equation}
    \mathbb{P}\left( \sum_{a = 0}^{n} \log r_{a, m_2, \ldots, m_d}^{(1)} - \sum_{a = 0}^{n} \log r_{a+1/2, m_2, \ldots, m_d}^{(1)} > a\right) \leq e^{-\frac{a^2}{2 K'}},
\end{equation}
and with $a = c \sqrt{\log N}$, the probability is less than $1/N^{\gamma}$ for some constant $\gamma = c^2/(2K')$.

Finally, for the case when $d = O(\log N)$, the previous analysis can be modified to show that $\bm{\varepsilon}_1^T \bm{K}^{-1} \bm{\varepsilon}_1 \leq c^d$ for some $c$.  As long as $c^d = o(N^2)$, we will not get exponential decay.  The case where $d = k \log N$ for any constant $k$ is trickier and will not be addressed in this paper.
\end{proof}
Next, we have analyzed the ratio of products along a single dimension of the $d$-dimensional hypercube.  To form a path from the corner of the hypercube to any point inside the hypercube, we will need to compute $d$ of these quantities.  In the following lemma, we quantify the concentration properties of these $d$ quantities:
\begin{lemma}
Define the following variables:
\begin{align}
    \Lambda_1 &= \sum_{a = 0}^{m_1} \log r_{a, 0, \ldots, 0}^{(1)} - \sum_{a = 0}^{m_1} \log r_{a+1/2, 0, \ldots, 0}^{(1)} \nonumber \\
    \Lambda_2 &= \sum_{a = 0}^{m_2} \log r_{m_1, a, \ldots, 0}^{(2)} - \sum_{a = 0}^{m_d} \log r_{m_1, a+1/2, \ldots, 0}^{(2)} \nonumber \\
    &\vdots \nonumber \\
    \Lambda_d &= \sum_{a = 0}^{m_d} \log r_{m_1, m_2, \ldots, a}^{(d)} - \sum_{a = 0}^{m_2} \log r_{m_1, m_2, \ldots, a+1/2}^{(d)}.
\end{align}

With high probability, $\sum_{i=1}^d \Lambda_i = O(\sqrt{d \log N})$ when $\log N < d \leq N$ and $\sum_{i=1}^d \Lambda_i = O(\sqrt{d 2^d \log N})$ when $d \leq \log N$.
\end{lemma}
\begin{proof}
Using arguments similar to those of the previous Lemma we can see that each $\Lambda_i$ is sub-Gaussian.  However, they are not independent, since $\Lambda_i$ depends on $\varphi$ and $\chi_i$.  Thus, we need to directly calculate the moment-generating function of their sum.  First we relate the BsGFF distribution to a BGFF distribution in the usual way.

\begin{equation}
\mathbb{E}_{\text{BsGFF}}\left[\exp\left(-t \sum_{i=1}^d \Lambda_i\right)\right] \leq
\mathbb{E}_{\text{BGFF}}\left[\exp\left(-\lambda t \sum_{i=1}^d \Lambda_i\right)\right]
\end{equation}
We can convert the $\Lambda_i$'s in terms of height fields and evaluate the expectation value under a biased Gaussian free field, yielding the expression
\begin{equation}
    \mathbb{E}\left[\exp\left(-\lambda t \sum_{i=1}^d \Lambda_i\right)\right] = \exp\left(\frac{g^2 \lambda^2 t^2}{2} \bm{\delta}^T \bm{K}^{-1} \bm{\delta} + \frac{g^2 \lambda^2 t^2}{2} \sum_{i=1}^d \bm{\varepsilon}_i^T \bm{K}^{-1} \bm{\varepsilon}_i\right)
\end{equation}
where $\bm{\delta} = \sum_i \bm{\varepsilon}_i$ and $\bm{\varepsilon}_i$ are generalizations of $\bm{\varepsilon}_1$ and $\bm{\varepsilon}_2$  defined as in the 2D section.  Along the same lines as the previous lemma, we can bound the RHS of this expression by $e^{d K t^2}$, for some $K=O(1)$.  Applying a Chernoff bound choosing $t = \frac{a}{2 K d}$, we have that
\begin{equation}
    \mathbb{P}\left( \sum_{i=1}^d \Lambda_i \geq a \right) \leq e^{-\frac{a^2}{4 K d}},
\end{equation}
and in particular, choosing $a = \sqrt{c \, d \log N}$ for constant $c > 0$ yields the RHS decaying like $N^{-p}$ for some $p > 0$.

When $d \leq \log N$, a similar result holds, where since each of the quadratic forms $\bm{\varepsilon}_i^T \bm{K}^{-1} \bm{\varepsilon}_i$ are upper bounded by $2^d$,
\begin{align}
    \mathbb{P}\left( \sum_{i=1}^d \Lambda_i \geq a \right) \leq e^{-at} e^{d \, 2^d K t^2},
\end{align}
and therefore we select $a = \sqrt{c\, d \, 2^d \log N}$.
\end{proof}

Next, we may use this result to formulate the following lemma:
\begin{lemma}
    With probability $1-N^{-p}$ and when $\log N < d \leq N$, the product of amplitudes of being at $(0,0,\ldots, 0)$ and $(N, N, \ldots, N)$ on a $d$-dimensional Lieb lattice satisfies $|\psi_{N,N,\ldots,N}||\psi_{0,0,\ldots,0}| \geq K\cdot N^{-d}\cdot e^{-3\sqrt{d^2 + c\,d\log N}}$ for constants $c$, $p$, and $K$.

    With probability $1-N^{-p}$ and when $d \leq \log N$, the product of amplitudes of being at $(0,0,\ldots, 0)$ and $(N, N, \ldots, N)$ on a $d$-dimensional Lieb lattice satisfies $|\psi_{N,N,\ldots,N}||\psi_{0,0,\ldots,0}| \geq K\cdot N^{-d}\cdot e^{-3\sqrt{d^2 + c\,d\, 2^d\log N}}$ for constants $c$, $p$, and $K$.
\end{lemma}
\begin{proof}
Identical to \lem{polydecay} in the 2D Lieb lattice case.  When $\\log N < d \leq N$, we note that, based on the previous lemma, there exists constants $p,K,c$ such that with probability $1 - \frac{1}{N^{d + p}}$,
\begin{equation}
    K e^{\sqrt{d^2 + c\, d\log N}} |\psi_{(0,0,\ldots,0)}| \geq |\psi_{(m_1,m_2,\ldots,m_d)}| \geq \frac{K}{e^{\sqrt{d^2 + c\, d\log N}}} |\psi_{(0,0,\ldots,0)}|.
\end{equation}
Then, using a union bound, the probability that this is true for all $N^d$ choices $(m_1,m_2,\ldots,m_d)$ is $1 - \frac{N^d}{N^{d + p}} = 1-\frac{1}{N^{p}}$.  The rest follows from using the normalization condition $\sum_{m_1, m_2, \ldots} |\psi_{(m_1,m_2,\ldots,m_d)}|^2 = 1$ to conclude
\begin{equation}
    |\psi_{(0,0,\ldots,0)}|^2 \geq \frac{1}{K^2\,N^d e^{2 \sqrt{d^2 + c\, d\log N}}}.
\end{equation}
Therefore, we arrive at the desired thesis.  The case when $d \leq \log N$ is analogous.
\end{proof}

Next, we must prove an estimate on the spectral gap.  There are $(d-1)N^d - d N^{d-1} + 2$ zero modes in total, and removing them results in $2N^d - 2$ remaining nodes.  As before, we may eliminate nodes until the resulting graph is a 1D snake graph wrapped up in a $d$-dimensional cube.  Call this new graph $G'$ and the $d$-dimensional Lieb lattice $G$.  $G'$ has an even number of nodes, and therefore hosts no zero modes, which allows us to make the following statement:
\begin{lemma}
The minimum eigenvalue in absolute value of the Hamiltonian constructed on $G'$ is less than the spectral gap of $G$.
\end{lemma}
As the Hamiltonian for $G'$ is exactly invertible, we may solve for the minimum eigenvalue by lower bounding it by $1/\norm{H^{-1}}_{1,\infty}$, and utilizing the results of the previous Lemma.  Following the logic in the previous subsection, in particular the proof of \lem{2D-gap}, the gap scales in the same way as $|\psi_{N,N,\ldots,N}||\psi_{0,0,\ldots,0}|$ (after performing a tedious but straightforward computation) and therefore, we arrive at the following theorem:
\begin{theorem}\label{thm:highD-main}
On a $d$-dimensional Lieb lattice, the time it takes starting at the entrance node $(0,0,\ldots,0)$ to arrive at the exit node $(N,N,\ldots,N)$ under a continuous time quantum walk is (with high probability):
\begin{itemize}
    \item $O(N^{3d} \exp(3d))$ for $\log N < d \leq N$
    \item $O\left(N^{3d} \exp(3\sqrt{d\,2^d \log N})\right)$ for $2 < d \leq \log N$
    \item $O(\textnormal{poly}(N))$ for $d = 2$
    \item $O\left(\exp(\sqrt{N})\right)$ for $d = 1$
\end{itemize}
\end{theorem}
For completeness, we also include the $d=1$ and $d=2$ results, both of which need to be treated separately (in $d=1$, our proof in \sec{random-graph} does not restrict to Gaussian free fields; in $d=2$, the proof method is different because correlation functions in \sec{2D} grow logarithmically).

To compare the performance of this with a classical algorithm, we will need to formally compute the lower bound for this hitting time problem, which is discussed in \cor{GFF-LB} in the next section.  However, we can heuristically identify when we expect the quantum algorithm to dominate.  We expect the classical algorithm in the best case to solve this problem in a number of queries that scales polynomially with the total number of nodes in the graph.  To estimate the number of nodes in the graph, we note that the shortest (Manhattan) distance between two randomly chosen points on the $d$-dimensional Lieb lattice is $O(N d)$.  Therefore, we expect that there exist a large number of supernodes that host $\exp(c N d)$ nodes, and that the total number of nodes is roughly at least $(N/k)^{d}\cdot \exp(c N d)$.  We now compare this to the quantum algorithm, whose performance is upper bounded by $N^{3d}$ for large dimension.  Therefore, for $\log N < d \leq N$ the quantum algorithm outperforms the best possible classical algorithm.  For small dimension $d \leq \log N$ the speedup is less pronounced but it still occurs.  \emph{Both of these speedups are exponential with respect to $N$}.  We believe our bounds can be significantly strengthened to deal with the case when $d = \Omega(N)$.

\subsection{General distributions and metric}

The results derived above depend particularly on the scaling of correlation functions in the Gaussian free field.  One may wonder whether there is a general prescription which is independent of the distribution we place on the height fields, apart from an independence assumption between the $\varphi$ and $\chi^{(i)}$ fields.  Recall that the crucial quantity whose fluctuations we had to bound was
\begin{align}
    \exp(\Phi_{m_1,\cdots,m_d}) =
     \frac{\prod_{i=0}^{m_1-1}  r^{(1)}_{i+1/2,0,\ldots,0}}{\prod_{i=0}^{m_1-1}  r^{(1)}_{i,0,\ldots,0}}\frac{\prod_{j=0}^{m_2-1}  r^{(2)}_{m_1, j+1/2,\ldots,0}}{\prod_{j=0}^{m_2-1}  r^{(2)}_{m_1,j,\ldots,0}}\ldots \frac{\prod_{q=0}^{m_d-1}  r^{(d)}_{m_1, m_2,\ldots,q+1/2}}{\prod_{j=0}^{m_d-1}  r^{(d)}_{m_1,m_2,\ldots,q}}.
\end{align}
Note that this quantity is the equivalence of the bottleneck metric discussed in the 1D case.  Written in terms of the height fields, this is given by
\begin{equation}
\exp(\Phi_{m_1,\cdots,m_d}) = \exp\left(\varphi_{m_1, \cdots, m_d} - \varphi_{0, \cdots, 0} + \sum_{i=1}^d (\chi^{(i)}_{m_1, \cdots, m_i, 0,\cdots} - \chi^{(i)}_{m_1, \cdots, 0, 0,\cdots})\right).
\end{equation}
Suppose the value of this quantity is upper and lower bounded by some function $f(d, N)$ with high probability.  Then, the quantum runtime is given by $O(N^{O(d)} \text{poly}(f(d,N)))$.  Therefore, if $f(d,N)$ grows subexponentially in $Nd$, then the quantum runtime is expected to be superpolynomially faster than the classical runtime.

Identifying this metric as quantifying the quantum runtime is important when the height field distribution does not obey a Gaussian free field.  It would be interesting to study the quantum runtime in these cases.

\subsection{A cochain complex for the zero mode}\label{sec:cohomology}

The reader may notice that the use of the gauge condition and the structure of the zero mode may be more succinctly stated in terms of a cochain complex.  In this subsection, we very briefly flesh out this intuition.

Call $X$ a manifold in $d$ dimensions and pick some cellulation for it.  In particular, define $\Delta_k$ to be the set of $k$-cells of $X$ (for example, a 0-cell is the boundary of a line, a 1-cell is the boundary of an area, a 2-cell is the boundary of a volume, etc.). Define the complex vector space of $k$-chains by defining a basis state for each $k$-cell:
\begin{equation}
\Omega_k(\Delta) = \text{span}_{\mathbb{C}}\{\ket{u} \, | \, u \in \Delta_k\}.
\end{equation}
From this, one may define the cochain complex
\begin{equation}
   \Omega_0 \xrightarrow{d_0} \Omega_1 \xrightarrow{d_1} \Omega_2 \xrightarrow{d_2} \cdots
\end{equation}
where $d_0, d_1, d_2, \ldots$ are maps between the indicated vector spaces and the above constitutes a cochain complex if $d_{i+1} d_{i} = 0$ for all $i$.
For example, for any edge $(u,v)$, $\ket{uv}\in \Omega_1$, $\ket{u},\ket{v}\in\Omega_0$, and
\be
d_0^T \ket{uv} = \ket{u} - \ket{v}.
\ee

The cohomology of this cochain complex is
\begin{equation}
    H^i(\Delta, \mathbb{C}) = \frac{\text{ker}(d_i: \Omega_i \to \Omega_{i+1})}{\text{im}(d_{i-1}: \Omega_{i-1} \to \Omega_{i})}.
\end{equation}
Let us now discuss how the construction of the zero modes relates to a cochain complex.  We consider a Lieb lattice in two spatial dimensions for simplicity.  This means we can define vector spaces on $0$-cells (sites), $1$-cells (links), and $2$-cells (plaquettes), where the cells are defined with respect to a square lattice that lacks additional edge decorations.  First define a basis for $\Omega_0$ to be vectors $\ket{u}$ labelling sites. Also define a basis for $\Omega_1$ to be vectors $\ket{uv}$ labelling links and a basis for $\Omega_2$ to be vectors $\ket p$ labelling $2$-cells.
From any $0$-chain $\ket{C_0} \in \Omega_0$, we can define an (unnormalized) wavefunction
\be
\ket\psi = \sum_{u \in \Delta_0} \exp(\braket{u}{C_0}) \ket{u}.
\ee
Similarly, starting from a  $1$-chain $\ket{C_1} \in \Omega_1$, we can choose a subset of hopping-hopping ratios in the Lieb lattice to be
\be \widetilde{r}_{uv} := \frac{t_{uv,v}}{t_{u,uv}} = \exp(\braket{uv}{C_1}).\ee

Note that we need to specify a direction for the links to define the hopping-hopping ratio, and we choose the horizontal links to point to the right and the vertical links to point upward.

It appears rather unconventional that our choice of vector spaces is associated to the amplitude of the zero modes and the values of the hopping, but we will see that this naturally provides a cochain complex.  For $\ket\psi\in \Omega_0$, the zero-mode condition $H\ket{\psi} = 0$ can be equivalently expressed as
\begin{equation}
    \log \psi_u - \log \psi_v = \log \widetilde{r}_{uv} \label{eq:psi-r-tilde}
\end{equation}
for  adjacent $u,v$.
Defining $\ket{C_0} \triangleq \sum_u \log(\psi_u)\ket{u}$ and $\ket{C_1}\triangleq \sum_v \log \widetilde{r}_{uv} \ket{uv}$ lets us
write the zero-mode condition Eqn.~\eq{psi-r-tilde} as
\be
d_0 \ket{C_0} = \ket{C_1}.\ee

Second, consider a 2-cell $p=uvwx$ with vertices $u,v,w,x$ as follows.
\be
  \begin{tikzpicture}[baseline={([yshift=1.5ex]current bounding box.center)},vertex/.style={anchor=base,circle,fill=black!25,minimum size=18pt,inner sep=2pt}]
   \draw (0,0) -- (1.6,0) node[below right] {$v$} -- (1.6,1.6) node[above right] {$w$}-- (0,1.6) node[above left] {$x$}-- (0,0) node[below left] {$u$};
    \draw[->] (0,0)--(0.8,0);
    \draw[->] (1.6,0)--(1.6,0.8);
    \draw[->] (0,1.6)--(0.8,1.6);
    \draw[->] (0,0)--(0,0.8);
    \draw (0.8,0.8) node {$p$};
  \end{tikzpicture}
\ee
The hopping ratios satisfy  the gauge condition
\begin{equation}
    \log \widetilde{r}_{uv} + \log \widetilde{r}_{vw} - \log \widetilde{r}_{xw} - \log \widetilde{r}_{ux} = 0.
\label{eq:gauge-cond}  \end{equation}

To express this gauge condition in terms of the co-boundary maps, we will define the vector space $\Omega_2$ spanned by $\ket{p}$ or alternatively $\ket{uvwx}$.  Similarly, we then may define the map $d_1$ by
\be
d_1^T \ket{uvwx} = \ket{uv} + \ket{vw} - \ket{xw} - \ket{ux}
\ee
We could equivalently write $d_1^T \ket{uvwx} = \ket{uv} + \ket{vw} + \ket{wx} + \ket{xu}$.
Using this, the gauge condition Eqn.~\eq{gauge-cond} becomes
  \ba
  d_1 \ket{C_1} = 0, \hspace{0.5cm}
  \ket{C_1} = \sum_{uv\in \Delta_1}  \log(\widetilde{r}_{uv}) \ket{uv} .
  \ea
Moreover, one can check that $d_1 d_0 = 0$ and therefore $\Omega_0 \xrightarrow{d_0} \Omega_1 \xrightarrow{d_1} \Omega_2$ forms a cochain complex.

We conclude with a short discussion of the relationship between this cochain complex and the zero mode.  In general, any vector in $\ker(d_1)$ satisfies the gauge condition.  Any zero mode corresponds a vector of (log) hopping ratios in $\text{im}(d_0)$.  Therefore, the cohomology $H^1(\Delta, \mathbb{R}) = \ker(d_1)/\text{im}(d_0)$ being trivial implies the existence of a valid zero mode.  If there are nontrivial cohomology classes, then there are non-local 2-cells where the gauge condition has not been satisfied (an example is a cellulation of an annulus, where the non-local 2-cell is a sequence of edges circling the annulus).

\section{Lower bound for classical algorithms}\label{sec:LB}

To prove that the classical algorithm cannot efficiently traverse such random graph ensembles, we proceed by performing an analysis similar to that used in the original welded tree paper~\cite{CCDFGS03}.  The main argument in \cite{CCDFGS03} relied on the fact that (1) a classical random walk would spend most of its time in the large supervertices in the middle of the graph; and (2) an arbitrary classical algorithm would be unlikely to find any cycles, so the subgraph of explored vertices would always be a tree.  This latter point meant that the algorithm would lack any structure that it could exploit to outperform a simple random walk.  Similar arguments were used in the oracle speedups for the adiabatic model~\cite{GSV21}.

In our model we will follow a similar outline but have to deal with the possibility that the classical algorithm could discover some cycles in the small supervertices.  It would seem natural to divide the supervertices into ``small'' ones of size $\ll Q$ and ``large'' ones of size $\gg Q$, where $Q$ is the number of queries. Think of $Q$ as somewhere between $\poly(n)$ and $\exp(\eps n)$ for a small $\eps>0$.  However, we cannot simply give the classical algorithm full knowledge of the small supervertices and hope that it remains completely ignorant of the structure of the large supervertices, because it is possible to find short cycles that start from small supervertices and yield information about nearby large supervertices.  So we will need to also classify supervertices as ``easy'' or ``hard'' based on how easily they can be reached with these cycles. We formalize these ideas as follows.

\begin{definition}[Supergraphs with small and large supervertices]
Let $(V,E)$ define a supergraph where each $v\in V$ corresponds to a set $S_v$
of vertices with $|S_v|=s_v$, and each edge $(v,w)\in E$ corresponds
to $e_{v,w}$ edges between $S_v$ and $S_w$.  Fix parameters $Q$ and
$\delta$, with $Q$ corresponding to the number of queries made by the
classical algorithm, and $O(\delta)$ the desired bound on its success probability.
A supervertex $v$ is {\em small} if $s_v\leq Q/\delta$ or
{\em large} if $s_v>Q/\delta$.  A small supervertex is a {\em boundary} supervertex if it is connected to at least one large supervertex.
\end{definition}

To define easy/hard vertices, we need to formalize a graph exploration
process that captures the ability of classical algorithms to use their
knowledge of small supervertices to learn about large supervertices.

\begin{definition}[Graph exploration walks]\label{def:explore}
  A {\em graph exploration walk} is a random non-backtracking walk
  that starts with some $x\in S_v$ for a supervertex $v$ and ends when
  one of the following conditions is met:
  \bit
\item The walk reaches a small supervertex.  This includes the
  possibility of reaching the original supervertex $v$ as long as
  this happens after one or more steps of the walk.
\item The walk encounters a previously visited vertex.
\item The walk has taken $Q$ steps.
  \eit

  If the walk stops after $Q$ steps then we say it is {\em
    unsuccessful}.  If it stops for one of the other two reasons, then
  it is {\em successful}.
\end{definition}

Graph exploration walks can be used to start from a small supervertex
and learn about (mostly nearby) large supervertices.  It is important
here to precisely specify what is known or unknown to the
algorithm. Here we approximate this knowledge by assuming that the
algorithm ``knows'' everything about small supervertices and nothing
about the edges involving large supervertices.  By being too generous
with the small supervertices, we can still prove lower bounds against
algorithms with partial knowledge of the small supervertices; while
for the large supervertices, we treat any nontrivial knowledge
obtained by the algorithm as an error, and prove that this probability
is negligible.

More concretely, the graph is randomly drawn from an ensemble (as
described in \sec{graph-model}), we allow the starting vertex
$x\in S_v$ to depend arbitrarily on the edges between easy
supervertices, but not the edges involving hard supervertices
(i.e.~(easy, hard) edges or (hard, hard) edges).  We will consider a
graph exploration walk to be a random process that has two sources of
randomness: the random choice of outgoing edge at each step, as well
as the random graph structure that is uncovered by the walk.

To analyze these walks it will be helpful to think of the choice of
the random graph as being made only as it is being revealed by the
algorithm.  More precisely, the {\em hard} part of the random graph
can be thought of as being generated as it is queried, since we will
allow our starting vertex to depend arbitrarily on the easy part of
the graph. An important feature here is that the supergraph structure
should be deterministic, including the sizes of supervertices and the
total number of edges between each pair of supervertices.  Thus, when
we refer to the random choices made in constructing the graph, we mean
the choices of where to draw the edges subject to these constraints.

\begin{figure}[ht]
  \centering
  \includegraphics[scale=0.5]{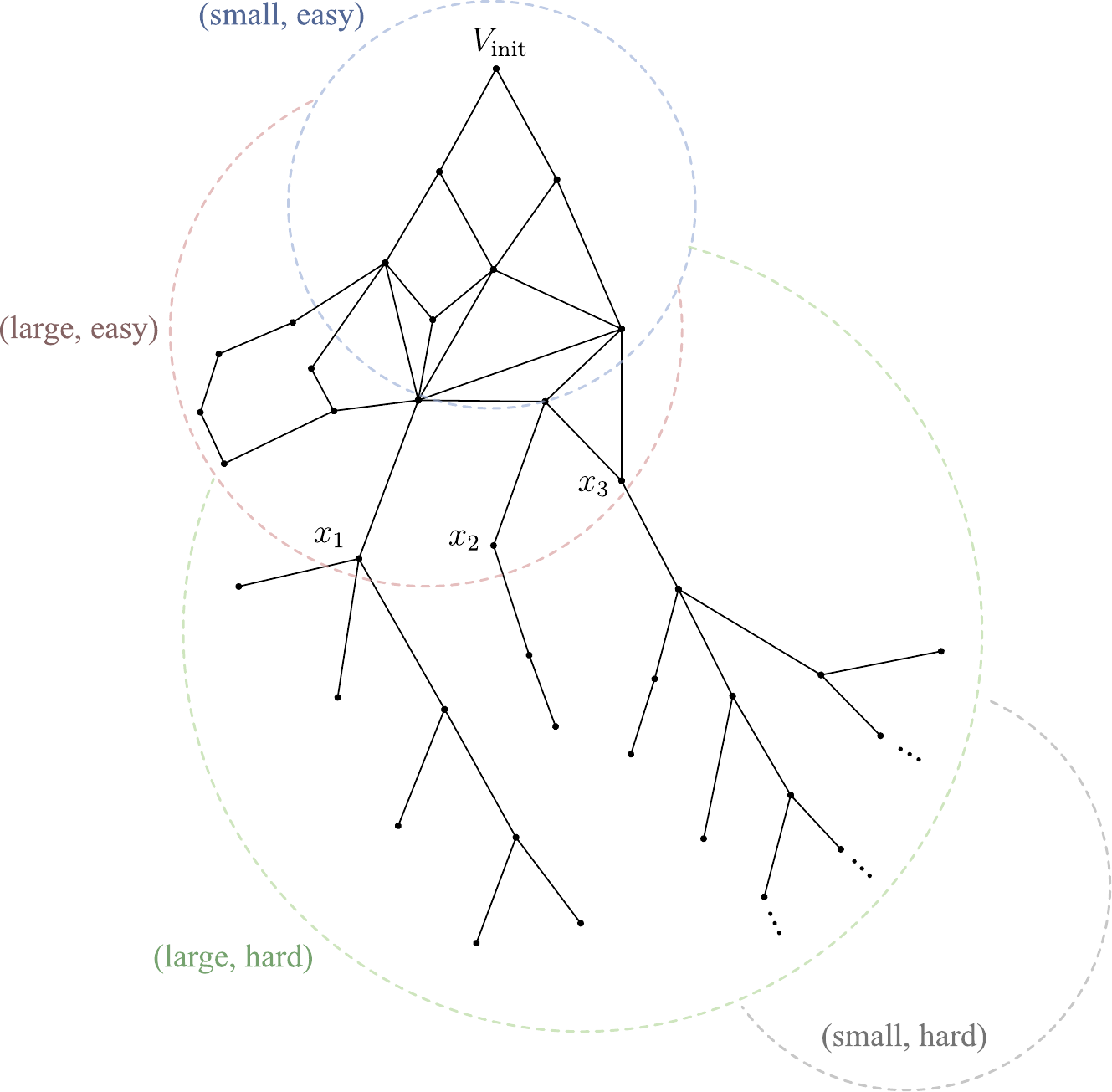}
  \caption{A schematic of the four regions: small easy supervertices, large easy supervertices, large hard supervertices, and small hard supervertices.}
  \label{fig:lower-bound}
\end{figure}

\begin{definition}[Informative superedges and easy supervertices]\label{def:easy-hard}
  A random hierarchical graph $\cG$ can be divided into easy and hard
  supervertices.  We call the vertices contained in these
  supervertices ``easy'' or ``hard'', respectively.  The subgraph
  restricted to the easy supervertices is called $\cG_{\text{easy}}$.
  We can imagine the process of sampling $\cG$ as first drawing
  $\cG_{\text{easy}}$ at random and then drawing the rest of the graph
  conditioned on our choice of $\cG_{\text{easy}}$.

  Given $v,w\in V$, $x\in S_v$ and some choice of $G_{\text{easy}}$,
  define $ P_{\text{reach}}(w|v; x, G_{\text{easy}})$ to be the
  probability that a graph exploration walk starting from $x$ is
  successful and reaches $S_w$ during the walk, conditional on our
  choice of $G_{\text{easy}}$.  This probability is taken over the
  random choices of the walk as well as the random choice of $\cG$,
  outside of $G_{\text{easy}}$.

  Now for $v,w\in V$ define
  \be
  P_{\text{reach}}(w | v) = \max_{G_{\text{easy}}}\max_{x \in S_v} P_{\text{reach}}(w|v; x, G_{\text{easy}})
  \ee
  Use $  P_{\text{reach}}(w | v)$ to define the {\em informative superedges} $E'$ as follows.  The set $E'$ consists of all $(v,w)$ with $v,w\in V$ such that $  P_{\text{reach}}(w | v)  > \delta$.

  Let $v_{\text{init}}$ be the starting supervertex.  We say that our division of $\cG$ into easy/hard supervertices is {\em consistent} if all the superverties reachable by edges in $E'$ from $v_{\text{init}}$ are contained in $\cG_{\text{easy}}$.
\end{definition}

This definition reflects the fact that the algorithm can choose the starting vertex $x$ based on its previous explorations of the easy supervertices.  However, the later steps of the walk can depend only on knowledge of the {\em small} supervertices.

In general $E'$ is neither a subset nor superset of $E$.  It will contain edges $(v,w)\in E$ when at least one of $v,w$ are small, but will usually fail to contain most edges between large supervertices.  $E'$ can also contain edges that are not in $E$ but usually corresponding to $v,w$ that are nearby in the original graph.

\begin{theorem}\label{thm:LB}
  Let $\cA$ be a $Q$-query classical algorithm making queries to the oracle $\cO$ defined in \defn{oracle}.  Suppose further that the graph $\cG$ has a consistent division into easy and hard supervertices.
  Then there exists a distribution over fake oracles $\cF$, depending only on $\cA$ and the information in the easy vertices of the graph, such that $\cA^\cO$ and $\cA^\cF$ produce output distributions within variational distance $O(Q\delta + 1/|V|)$ of each other.
\end{theorem}

As a result, such an algorithm cannot reliably produce outputs depending on hard supervertices.  In particular, it cannot output the identifier of the target vertex with success probability $\gg \delta$.

\begin{proof}[Proof of \thm{LB}]
    Assume for simplicity that we have defined the encoding function in
  \defn{oracle} to force our algorithm to explore the graph
  starting from $v_{\text{init}}$.  This part of the proof is
  identical to Lemma 4 of \cite{CCDFGS03} and we do not repeat it
  here.  If $\Enc$ maps to random strings in a set of size $|V|^2$
  then the probability of the algorithm finding any vertex by any
  means other than exploring the graph is $\leq 1/|V|$.

  As a result, we can assume that every query is of the form $(\Enc(v),s)$ where $v$ is a vertex appearing in a previous oracle output (or else is $v_{\text{init}}$) and $s\in [D]$.  The queries made by $\cA^\cO$ then form a connected subgraph of $\cG$, which we call $\cG(\cA^\cO)$.  Define a subgraph $\cG_{\text{init}}$ to be vertices that can be reached from $v_{\text{init}}$ using only edges in $\cG(\cA^\cO)$ and traversing only easy vertices.  The remainder $\cG(\cA^\cO)-\cG_{\text{init}}$ consists of $k$ connected components, which we call $T_1,\ldots, T_k$.  We will argue that each $T_i$ is very likely to be a tree and to have a single neighboring vertex in $\cG_{\text{init}}$.

  Suppose for now that this is true.  We need to specify how $\cF$ acts on $T_1,\ldots,T_k$. When the algorithm encounters a new vertex $v$ in $T_i$, $\cF$ will either return the true vertex label of $v$ if $v$ is easy, or will return a fresh random label if $v$ is hard.  (Of course, $\cF$ will answer any repeated queries to $v$ in a way that is consistent with these choices.)
  This way, $\cF$ will not depend on any information from $\cO$ evaluated on any hard vertex.  We will also show, by induction, that each $T_i$ is isomorphic to a subgraph of the infinite $D-1$-ary tree.  Thus, neither the connectivity nor the labels will carry any information from $\cO$ about the hard vertices.

  To prove this claim, consider each query, WLOG one that has not already been made.  If the new query is to an easy vertex with only easy neighbors then it will not add any hard vertices, due to our assumption that the easy/hard partition is consistent.  If it queries an easy vertex with a hard large neighbor then it has a chance of choosing this neighbor.  Here we consider the randomness resulting from our choice of random graph.  Specifically we will use the fact that the distribution is invariant under permutations of vertices within any supervertex.
To bound that this hard neighbor has been previously seen, note that $\leq Q$ vertices have been seen anywhere in the graph up until this point.  Also, the ``large'' property means that the hard supervertex in question has size $\geq Q/\delta$.  As a result, the probability that this hard neighbor has been previously seen is $\leq \frac{Q}{Q/\delta} = \delta$.
When a new neighbor has {\em not} been previously seen (which occurs with probability $\geq 1-\delta$), this query creates a new tree $T_{k+1}$.

Similarly suppose we query a vertex at depth $\ell$ of a tree $T_i$.  The new vertex we find is distributed identically to the $(\ell+1)^{\text{th}}$ step of a graph exploration walk (see \defn{explore}), conditioned on the first step of the walk entering the hard vertices. Here again this randomness comes from our random choice of graph as well as the possibly random choices of $\cA$.  Conditioning on the first step involves choosing from among the $D$ outgoing edges so it multiplies the probability of any event by $\leq D$.  The fact that the walk enters the hard supervertices means that the graph exploration walk has $\leq \delta$ chance of being ``successful'' (see \defn{easy-hard}), i.e.~reaching a small supervertex or a previously seen vertex.  Conditioning on the first step changes this bound to $\leq D \delta$.  Taking the union bound over all $Q$ queries in the algorithm, we find that our tree claim fails over the course of the algorithm with probability $\leq DQ\delta$.
\end{proof}

\begin{corollary}\label{cor:1D-LB}
  Let $G$ be a 1D hierarchical graph of the form in \sec{random-graph} with nodes
  $[2n]=\{1,\ldots,2n\}$ and edge ratios $r_i \sim D_+$ for
  $1\leq i\leq n$ and $r_i \sim D_-$ for $n+1\leq i \leq 2n$, such
  that $\bbE[\log(D_+)] = \bbE[-\log(D_-)] = \mu >0$ and
  $|\log(D_-)|, |\log(D_+)|\leq \Delta$ with probability 1 for some
  constant $\Delta$.

  Then with high probability, for a graph taken from this random graph
  ensemble, it takes $\exp(\Omega(n))$ queries for a classical
  algorithm to traverse from 1 to $2n$.
\end{corollary}

\begin{proof}
  Let $\gamma>0$ be a small constant which we will choose later, and take
  $Q = \exp(\gamma n)$ and $\delta = \exp(-2\gamma n)$, so a
  supervertex is large if it has size $\geq \exp(3\gamma n)$.  The
  Chernoff-Hoeffding theorem states that for $v\leq n$,
  \be
  \mathbb{P}\left(s_v < \exp(3 \gamma n)\right) \leq
  \exp(-n\frac{(\frac v n \mu - 3\gamma)^2}{2\Delta^2}).
  \ee
  This is $\leq \delta$ for $v > c_1 \gamma n$, where $c_1>1$ is a constant
  depending on $\mu$ and $\Delta$.  Thus with high probability the
  supervertices $c_1\gamma n \leq v \leq n$ are all large.

  For $v> n$, a similar analysis holds except that one starts at the
  right-most node $2n$ and works backwards.

  Next we need to analyze random non-backtracking walks of length $Q$ starting at $v$
  in the range $c_2\gamma n< v \leq n$ for some $c_2>c_1$ which will correspond to the easy/hard transition.  We claim that these are
  unlikely to reach small supervertices.  First note that the walk on
  supervertices does not depend on which specific vertex we start
  with.  So we can treat the starting vertex as randomly distributed
  over a supervertex $v$ with $c\gamma n< v \leq n$.  Let $\pi$ be the uniform distribution over vertices, or equivalently the distribution which gives supervertex $v$ probability $s_v / |\cV|$.  (Note that the same uniform distribution holds even taking into account the fact that the walk is non-backtracking, using the fact that $G$ is regular and assuming $D\geq 3$; see Section 1.2 of \cite{back06}.)

  Then a basic property of random walks implies that the probability of starting at vertex $v$ and reaching $w$ on the  $t^{\text{th}}$ step is $\leq \frac{\pi(w)}{\pi(v)}$.  Concretely, we  start with the uniform distribution over $v$, given by $\frac{\mathbbm{1}_{S_v}}{s_v}$, and apply $t$ steps of the walk, which we denote with transition matrix $M$.  Then $M\pi=\pi$ (because $\pi$ is stationary) and $\frac{\mathbbm{1}_{S_v}}{s_v} \leq \frac{\pi}{\pi(v)}$ elementwise.  The probability of ending in $S_w$ after $t$ steps can then by bounded by
  \begin{equation}
    \mathbbm{1}_{S_w}^T M^t \frac{\mathbbm{1}_{S_v}}{s_v}
    \leq \frac{1}{\pi(v)} \mathbbm{1}_{S_w}^T M^t {\pi} = \frac{\mathbbm{1}_{S_w}^T {\pi}}{\pi(v)} = \frac{\pi(w)}{\pi(v)} = \frac{s_w}{s_v}.
  \end{equation}
  This is independent of $t$, but for a $Q$-step walk we need to take the union bound over each step.  This means that the probability of reaching $w<c_1\gamma n$ after $Q$ steps is $\leq Q s_w / s_v$.  Next, we can bound $\frac{s_w}{s_v}$ with another Chernoff-Hoeffding bound, since
  \be
  \log \frac{s_w}{s_v} = \kappa + \sum_{w < i \leq v} -\log r_i
  \ee
  is a sum of i.i.d.~bounded random variables with mean $-\mu$, and $\kappa$ is an $O(1)$ constant coming from conversion between the number of edges between supervertices and supervertex sizes.   A similar argument bounds the probability of reaching the small supervertices near $2n$. Finally we can choose $c_2$ to be a large enough constant so make these probabilities $<\delta$, thus establishing that supervertices $>c_2\gamma n$ are hard. (More precisely labeling them hard  and labeling the supervertices $\leq c_2\gamma n$ easy will meet our definition of consistency.)

Our choices of $c_1,c_2$ did not depend on $\gamma$ so we can choose $\gamma$ last to be $1/2c_2$, ensuring that a set of hard vertices exist in the middle of the supergraph.  Using \thm{LB}, these prevent any classical algorithm from traversing from 0 to $2n$ in less than exponential time.
\end{proof}

This allows us to prove the statement from the main text:
\begin{corollary}\label{cor:general-1D-LB}
  With high probability, for a graph taken from this random graph
  ensemble, it takes $(\max_i s_i)^{\Omega(1)}$ queries for a classical
  algorithm to traverse from 1 to $2n$.
\end{corollary}
\begin{proof}
Using a Chernoff bound, it is simple to see that the maximum size of a supervertex is $\exp(c n)$ with high probability.  Given that the lower bound on the classical query complexity is $\exp(\Omega(n))$, the thesis follows.
\end{proof}

\begin{corollary}\label{cor:GFF-LB}
With high probability, it requires $\exp(\Omega(Nd))$ classical queries to traverse a graph taken from the BGFF ensemble from \sec{highdim}.
\end{corollary}

\begin{proof}
  The proof is similar to the 1D case (\cor{1D-LB}) so we only outline it here. As we have introduced in the discussion after \thm{highD-main}, the size of the supernodes grows from $1$ in the corners to $\exp(\Theta(N d))$ in the bulk of the graph.  We can then take $Q = \exp(N d \gamma)$ and $\delta=1/Q$ for small constant $\gamma>0$ and then find that the only small supernodes are in the corners near the entrance and exit, i.e. they have coordinates $(x_1,\ldots,x_d)$ with $\sum_i x_i$ either $<\frac 14 Nd$ or $>\frac 34 Nd$.

  Next, once classical random walks reach the region with $\sum_i x_i \in [\frac 13 Nd, \frac 23Nd]$ they have a negligible chance of returning to a small vertex.  This can be seen just by bounding the ratios of the stationary probabilities, which are $\exp(\Omega(Nd))$.  Thus we have met the conditions for the lower bound of  \thm{LB}.
\end{proof}

\section{Concrete realizations of hierarchical graphs}\label{sec:metropolis}

So far we have discussed the behavior of quantum walks and classical algorithms on hierarchical graphs.  But how can such graphs be constructed? In this section, we introduce general frameworks for doing so.  In the first approach, we propose an explicit construction of such graph utilizing factors of the graph degree $D$, which can be applied to both 1D and higher-dimensional graphs.  The second and more interesting approach generalizes beyond the balanced constraint by using ideas from graph sparsification.  In particular, given Hamiltonian $H$ defined on a supervertex subspace, we use a new kind of graph sparsification to generate a hierarchical graph which with high probability has an approximate supervertex subspace with effective Hamiltonian close to $H$ in operator norm. 

\subsection{Explicit construction of 1D hierarchical graph}\label{sec:construction-1D}

Next, we would like to show that the set of these graphs are nontrivial, i.e., it is possible to construct a large class of random graphs satisfying these constraints.
\begin{example}
  For any positive integers $n,D>1$, we can define a $D$-regular random balanced hierarchical graph on a supergraph $G = (V,E)$ which is a line graph with vertices $\{0,1,\ldots,2n\}$.

  Let $d_1,\ldots,d_k$ be the proper factors of $D$.  For $i=1,\ldots,{n-1}$, randomly select $d_{L,i} \in \{d_1, d_2, \cdots, d_{k}\}$ based on any probability distribution over this finite alphabet, and $d_{R,i} = D - d_{L,i}$, thereby ensuring that $d_{R,i}/d_{L,i} \in \mathbb{N}$.   Here $d_{L,i}$ and $d_{R,i}$ are the left/right degrees of vertex $i$.  We mirror the graph on vertices $n,\ldots,2n$ such that $d_{R,n+i}=d_{L,n-i}$.

  Finally, set $s_0=s_{2n}=1$ and set $e_1=e_{2n-1}=D$.

\end{example}

We claim that this yields a valid graph family with integer $\{s_i\}$ and $\{e_i\}$.
For $1\leq i\leq n$, one can verify that the size of a supervertex and superedge are given by
\be e_{i} = e_1 \prod_{1 \leq j < i} \frac{d_{R,i}}{d_{L,i}}
\qand
s_i =  \frac{e_{i+1}}{d_{R,{i+1}}}\ee
Because $d_{L,i} | D$ and $d_{R,i} = D - d_{L,i}$, it follows that $d_{R,i}/d_{L,i} \in \mathbb{Z}$ and therefore $\{e_i\}$ are integers.  We may write $s_{i}$ as
\begin{equation}
s_i = \frac{e_1}{d_{L, i+1}}\prod_{1 \leq j < i} \frac{d_{R,i}}{d_{L,i}},
\end{equation}
and using the fact that $e_1 = D$ and $d_{L,i} | D$, $\{s_i\}$ are integers as well.  We call such a probability distribution a \emph{good distribution} if it ensures that $\{s_i\}$ and $\{e_i\}$ are integers; the example above shows that good distributions exist.
\begin{lemma}
Given a good distribution over edge-edge ratios, there exists an explicit construction of a $D$-regular random balanced hierarchical graph ensemble on a supergraph $G = (V,E)$ which is a line graph with vertices $\{0,1,\ldots,2n\}$.
\end{lemma}
\begin{proof}
  First, one samples the edge-edge ratios $r_i \sim D_i$.
  Equivalently, one may sample the edge-vertex ratios
  $\kappa_i \sim F_i$ and extract the edge-edge ratios from them using
  $\kappa_i = D/(1 + r_i)$.  Then, one constructs the edge sets
  $\mathcal{E}_i$ from the relation $e_{i+1} = r_i e_i$.
  Finally, one constructs the sets of nodes $S_i$ using the relation
  $D s_i = e_i + e_{i+1}$.  Randomly connect the $e_{i}$ edges between nodes in $S_{i-1}$ and $S_i$ while satisfying the left and right degree constraints.  This generates a random hierarchical graph from the desired ensemble.
\end{proof}

This ensemble of graphs contains an exponentially large number of members and is thus quite generic.  An example of a such a member was  shown in the introduction, in \fig{1D-graph-example}.

\subsection{Explicit construction of higher-dimensional hierarchical graph}\label{sec:construction-highD}

In \defn{liebensemble}, hierarchical graphs on a Lieb lattice are constructed from a probability distribution over the height fields.  As in the 1D case, we would like to show the existence of a good probability distribution, which results in an explicit/concrete construction of such a graph.

As in the previous subsection, define $d_{L,(x,y)}^{(i)}$ and $d_{R,(x,y)}^{(i)}$, where the superscript $i$ denotes one of $d$ directions on the lattice.  We will make the assumption that $d_L^{(i)} + d_R^{(i)} = D$ for all $i$. For the sites with degree 2, we can drop the superscript label.  Below, we provide an explicit probability distribution (albeit certainly not the most general) which ensures that the graph that is constructed has integer values for the edge and supervertex sizes.  This graph is constructed by creating a mountain shaped landscape and superimposing this landscape with fluctuations that have a Gaussian free field scaling limit.

We first require the following assumption on the degree sizes:
\begin{assumption}\label{asm:factor}
Let $F = \{d_1,\ldots,d_f\}$ be the proper factors of $D$. Pick an $f \in F$ and subset $S_f \subseteq F$ such that for any $d_i \in S_f$, $\exists f_{+}, f_{-} \in F$ satisfying
  \begin{equation}
  \log \left(\frac{D}{f} - 1\right) \pm \log \left(\frac{D}{d_i} - 1\right) = \log \left(\frac{D}{f_{\pm}} - 1\right).
  \end{equation}
For large enough $D$, there exists a choice of $f$ so that the set $S_f$ can have large cardinality.
\end{assumption}

For any positive integers $N,D,d>1$ with $D$ obeying the above assumption, we can define a $D$-regular random balanced hierarchical graph on a supergraph $G = (V,E)$ which is a $d$-dimensional Lieb lattice with linear size $N$.  We first describe the construction of a mountain-type hierarchical graph with these parameters.

\begin{figure}[t]
    \centering
    \includegraphics[scale=0.55]{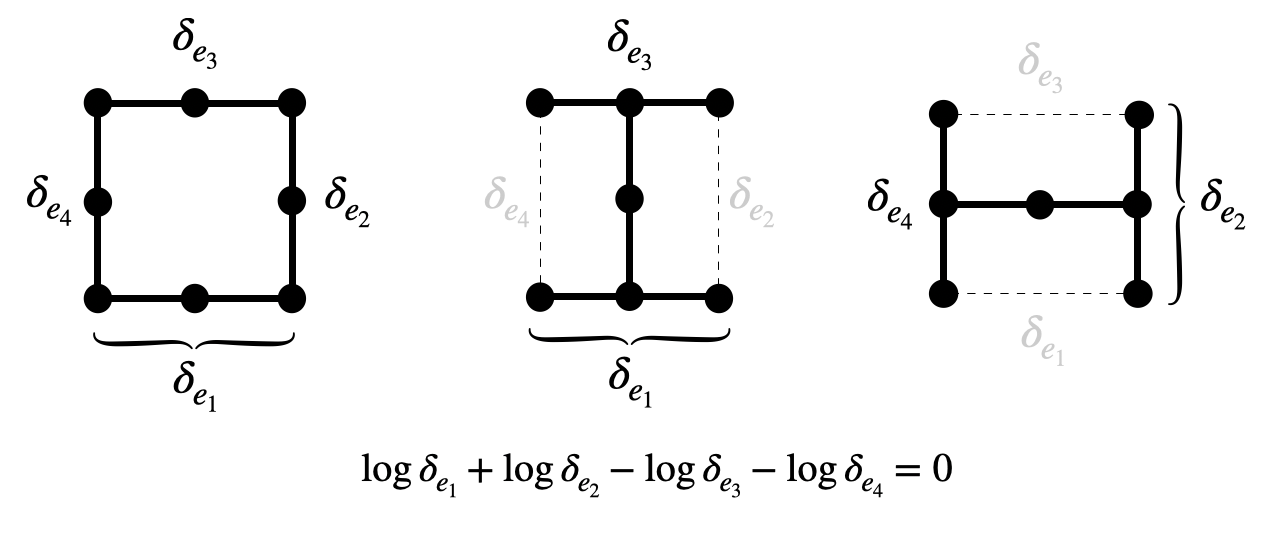}
    \caption{The 2-cells of each of the 3 square sublattices of the 2D Lieb lattice.  The assignments of the factors are shown as well as the gauge constraint.  Each of these assignments are drawn from a distribution satisfying the gauge constraint separately on each of the 3 sublattices.  The light gray labels and edges indicate that there are no assignments to these edges.}
    \label{fig:2D-graph-example}
\end{figure}

\begin{example}[Mountain hierarchical graph]

For each site $v = (x_1,x_2,\ldots,x_d) \in \Delta_1$ ($\Delta_1$ corresponds to the edge sublattice) with orientation pointing along direction $i$ of the Lieb lattice, set $d_{R,v}^{(i)}/d_{L,v}^{(i)} = D/f - 1$ if $x_i \leq N/2$ and $d_{R,v}^{(i)}/d_{L,v}^{(i)} = (D/f - 1)^{-1}$ if $x_i > N/2$.  This creates a mountain-type landscape with a uniform bias.
\end{example}

For the mountain hierarchical graph, we may observe the following fact:
\begin{proposition}
    The mountain hierarchical graph satisfies the gauge constraint and thus admits a zero mode supported on $\Delta_0$.
\end{proposition}
Next, we want to add fluctuations to this mountain hierarchical graph so that the gauge constraint is still satisfied and the limit distribution describing fluctuations is a Gaussian free field.

\begin{example}
    Split the $d$-dimensional Lieb lattice into $d+1$ hypercubic sublattices.

    Assign a ratio $\delta_e \in S_f$ to each edge $e$ of each hypercubic sublattice.  These are assigned subject to the constraint that taken counterclockwise around each 2-cell of each hypercubic sublattice, $\delta_{e_1}\delta_{e_2}\delta_{e_3}^{-1}\delta_{e_4}^{-1} = 1$ ($e_1$ and $e_3$ are the bottom and top edges, $e_2$ and $e_4$ are the right and left edges).  Alternatively, $\log \delta_{e_1} + \log \delta_{e_2} - \log \delta_{e_3} - \log \delta_{e_4} = 0$ around each 2-cell.

    Having randomly selected a configuration $\delta_e$, modify the values of $d_{R,v}^{(i)}/d_{L,v}^{(i)}$ of the mountain landscape by multiplying them by the corresponding $\delta_e$.

    Finally, set $s_{(0,0,\cdots,0)}=s_{(N,N,\cdots,N)}=1$ and $e^{(i)}_{(0,0,\cdots,0)}=e^{(i)}_{(N,N,\cdots,N)}=2D$, and use the edge-edge ratios constructed above to determine the superedge sizes and supervertex sizes.
\end{example}

\begin{figure}
    \centering
    \includegraphics[scale=0.5]{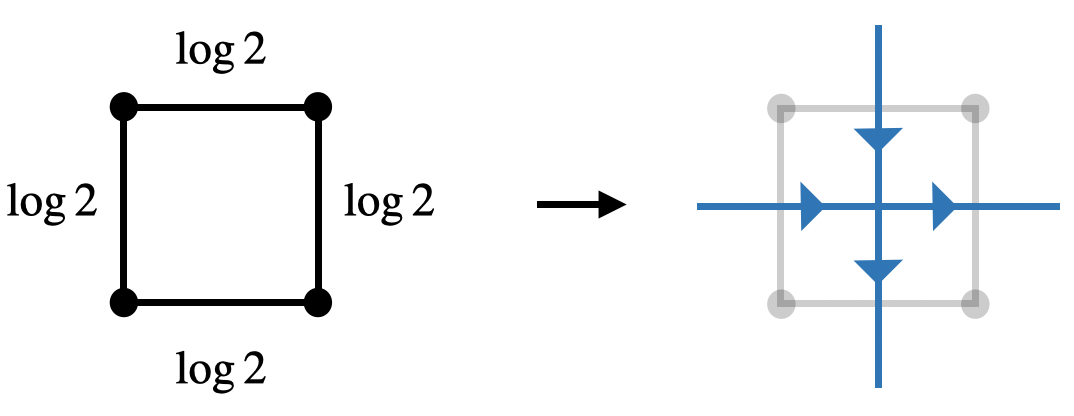}
    \caption{The ``reference'' configuration mapping a configuration of log edge-edge ratios to a configuration of arrows in square ice.  Negating $\log 2 \to - \log 2$ on an edge reverses the direction of the arrow crossing the edge.  The gauge condition is mapped onto configurations of arrows satisfying a zero flux condition.}
    \label{fig:icerule}
\end{figure}

As an explicit example, we consider a 2D Lieb lattice, where the 2-cells and the assignment of the ratios $\delta_e$ for each of the 3 square sublattices is depicted in \fig{2D-graph-example}.  For simplicity, we can choose $D = 30 = 2\cdot3\cdot 5$.  If $f = 10$, then the corresponding edge-edge ratio is $2$ in the rising part of the mountain landscape and $1/2$ in the falling part.  Next, we consider the set $S_{10} = \{2, 1/2\}$.  We want to assign either $\log 2$ or $-\log 2$ to the log edge-edge ratios subject to the gauge constraint.  If all edges are assigned log ratio $\log \delta_e = \log 2$ on all edges of a 2-cell, we identify this assignment with a reference configuration of ``arrows'' on a dual lattice, which is depicted in \fig{icerule}.  Negating a log ratio reverses the direction of the arrow.  Any configuration of log edge-edge ratios satisfying the gauge condition around a 2-cell is equivalent to the zero net flux of arrows into the 2-cell.  Thus, we therefore randomly select a configuration of arrows satisfying a zero-flux condition on all 2-cells, convert this to an assignment of edge-edge ratios, and modify the mountain hierarchical graph with this configuration.  The set of all zero flux configurations define a statistical mechanics model known as \emph{square ice}.  Viewed as a probability distribution which is uniform over all zero-flux configurations, there is significant analytical and numerical evidence that square ice has a well-behaved scaling limit to a Gaussian free field \cite{sheffield2005random}.  The square ice model also admits an exact solution due to Lieb \cite{lieb1967residual}.

Another example corresponds to $D = 99450 = 2\cdot 5\cdot 9\cdot 17 \cdot 65$ and $f = 11050 = 2\cdot 5 \cdot 17 \cdot 65$.  Then, $S_{f} = \{8,2, 1/2,1/8\}$.  The gauge condition can be satisfied from sampling uniformly from configurations satisfying a ``dimer'' constraint on the dual lattice (i.e. dimers are assigned to edges and a single dimer touches each site).  We illustrate a mapping from the dimer model to a configuration of edge-edge ratios in \fig{dimer}.   The dimer model (also known as the domino tiling model) is a famous model of statistical mechanics which admits an exact solution for the partition function and correlation functions (see Kasteleyn, Temperley, and Fisher \cite{kasteleyn1961statistics,temperley1961dimer}).   Kenyon has also rigorously proved that the dimer model has a scaling limit to a Gaussian free field \cite{kenyon}. Thus, this assignment of edge-edge ratios is also suitable to use for fluctuations about the mountain hierarchical graph.

\begin{figure}\label{fig:dimer}
    \includegraphics[scale=0.45]{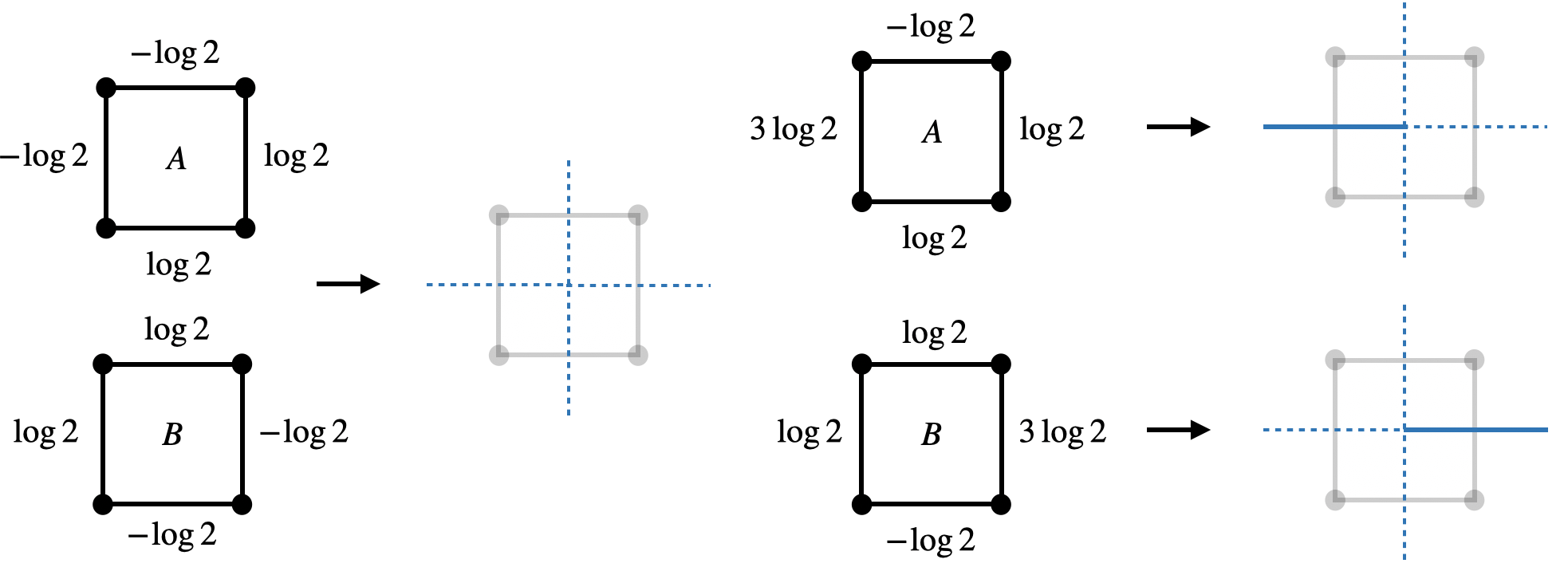}
    \caption{On the left, a ``reference'' configuration denoting a site with no dimer and the corresponding edge-edge ratios on the A and B sublattice 2-cells.  Of course, this reference configuration \emph{does not ever appear} because of the single dimer per site constraint, but is helpful for constructing the height fields in the presence of dimers.  On the right, we discuss how the edge-edge ratios change in the presence of a dimer.  Adding a dimer to an edge whose edge-edge ratio is $-\log 2$ changes it to $3\log 2$; adding a dimer to an edge whose edge-edge ratio is $\log 2$ changes it to $-3\log 2$.  This change preserves the gauge condition.}
\end{figure}

Note that, by an analogous discussion to the 1D case, all of the superedges and supervertices will have an integer number of edges and vertices, respectively.  We expect large classes of fluctuations realized over the mountain hierarchical graph to have a Gaussian free field scaling limit \cite{sheffield2005random}.

\subsection{Hierarchical graphs from sparsification}\label{sec:construction-sparse}
In this section, we describe a different and more general framework for constructing a hierarchical
random graph ensemble.  Unlike the previous sections in which the supergraph is a regular lattice, this approach will allow an arbitrary choice of supergraph.  The idea is to construct a densely connected hierarchical graph with the desired structure, and then to ``sparsify'' it, namely to construct a lower-degree graph that has approximately the same spectrum.  The resulting graphs will no longer have an exact supervertex subspace, and this allows us to generalize beyond the more rigid balanced constraint required in the previous sections.

We will start with a set of vertices $V$ which correspond to the
supervertices in the hierarchical graph we will construct.  Let $\bbC^V \cong \bbC^{|V|}$ denote the corresponding Hilbert space.
We also have a matrix of hopping amplitudes $t_{uv}\geq 0$ for $u,v \in V$, which
corresponds to the effective Hamiltonian we will obtain in the supervertex
subspace.  (The matrix $t$ should be symmetric but we do allow
non-zero diagonal entries; this latter case will be discussed further
in \sec{Anderson}.)  We claim that from this data it is possible to
construct a weighted hierarchical graph with the desired effective
Hamiltonian.

\begin{lemma}
  Given an effective Hamiltonian $t$ acting on $\bbC^V$, we can construct a dense hierarchical graph $\mathcal{G}=(\cV,\cE)$, with corresponding weighted adjacency matrix $A_{\cG}$ such that:
  \bit
\item $\cG$ is $\lambda$-regular for some constant $\lambda$.  In other words, each row and column of $A_{\cG}$ sums to $\lambda$;
\item  There exists an exact supervertex subspace such that $A_{\cG}$ restricted to this subspace is isospectral with $t$.  In terms of the projector $\cP$ onto this subspace this means that $\cP A_\cG \cP$ is isospectral with $t$ and $[\cP,A_\cG]=0$.
  \eit
\end{lemma}
\begin{proof}
We first choose supervertex sizes in $\mathcal{G}$ to be $s_v = N \pi(v)$ where $\pi$ is a probability distribution and $N$ is the total number of vertices.   We determine $\pi$, and thereby $s$, according to the eigenvalue equation
\be t \sqrt{\pi} = \lambda \sqrt{\pi} \label{eq:t-pi-relation}\ee
where $\lambda>0$ is the top eigenvalue which is guaranteed to be unique by the Perron-Frobenius theorem.
From $s$ we can also choose
\be
e_{uv} \triangleq t_{uv} \sqrt{s_us_v} = t_{uv} N \sqrt{\pi(u)\pi(v)}.
\ee

Note that Eqn.~\eq{t-pi-relation} guarantees that
\be
\sum_{v\in V} e_{uv} = \lambda s_u\label{eq:reg-mu}.
\ee
Thus we can create $\mathcal{G}$ by adding an edge between each $(u,x)\in S_u$ and each $(v,y)\in S_y$ with weight $\frac{e_{uv}}{s_u s_v}$.  The weighted degree of a vertex $(u,x)$ is
\be
\sum_{v\in V}\sum_{y\in S_v}\frac{e_{uv}}{s_u s_v}
= \sum_{v\in V}\frac{e_{uv}}{s_u}
 = \lambda,
\ee
by Eqn.~\eq{reg-mu}.
If $t_{uu}\neq 0$ then we add edges between all pairs $(u,x),(u,y)\in S_u$, each with weight $\frac{e_{uu}}{s_u(s_u-1)}$.
Let $A_\mathcal{G}$ (we will drop the subscript $\mathcal{G}$ in what follows) denote the weighted adjacency matrix of the hierarchical graph described above.  Each row and column of $A$ sums to $\lambda$, the entries are nonnegative with zeroes on the diagonal, and the supervertex subspace is an invariant subspace.
When we restrict $A$ to the supervertex subspace we obtain the matrix $t$.  While $t$ may have had diagonal entries, $A$ does not, since the graph does not have any self-loops.
\end{proof}

We now discuss the performance of the quantum walk on the sparsified graph relative to that of the quantum walk on the dense graph:
\begin{proposition}
Suppose the quantum walk implemented with Hamiltonian $H = -A$ requires time $T$ and satisfies
\begin{equation}
\mathbb{P}(\text{EXIT}) = \frac{1}{T} \int_0^T dt\,|\mel{S_{\text{exit}}}{e^{-i H t}}{S_{\text{init}}}|^2 \geq p.
\end{equation}
Then the quantum walk implemented with Hamiltonian $\widetilde{H} = -\tilde{A}$ corresponding to the sparsified graph satisfies 
\begin{equation}
\mathbb{P}(\text{EXIT}) = \frac{1}{T} \int_0^T dt\,|\mel{S_{\text{exit}}}{e^{-i \widetilde{H} t}}{S_{\text{init}}}|^2 \geq \frac{p}{2}.
\end{equation}
if $\norm{A - \tilde A} \leq \frac{p}{4T}$.
\end{proposition}
\begin{proof}
We write
\begin{align}
\frac{1}{T} \int_0^T dt\,|\mel{S_{\text{exit}}}{e^{-i \widetilde{H} t}}{S_{\text{init}}}|^2 &= \frac{1}{T} \int_0^T dt\,\abs{\mel{S_{\text{exit}}}{e^{-i H t}}{S_{\text{init}}} + \mel{S_{\text{exit}}}{e^{-i \widetilde{H} t} - e^{-i H t}}{S_{\text{init}}}}^2 \nonumber \\
&\geq \frac{1}{T} \int_0^T dt\,\abs{\mel{S_{\text{exit}}}{e^{-i H t}}{S_{\text{init}}}}^2 - \frac{2}{T} \int_0^T dt\,\abs{\mel{S_{\text{exit}}}{e^{-i \widetilde{H} t} - e^{-i H t}}{S_{\text{init}}}} \nonumber \\
&\geq \frac{1}{T} \int_0^T dt\,\abs{\mel{S_{\text{exit}}}{e^{-i H t}}{S_{\text{init}}}}^2 - 2 \norm{ H - \tilde H} T.
\end{align}
Thus, with $\norm{ H - \tilde H} = \norm{A - \tilde A} \leq \frac{p}{4T}$, the proof is finished.
\end{proof}

\subsubsection{Two approaches to sparsification}
We will now present two approaches for sparsifying to obtain $\widetilde{A}$.  Sparsification is necessary because implementing $e^{-iHt}$ is not efficient on a quantum computer due to the large degree of the unsparsified graph.  The above proposition shows that so long as $\norm{A - \tilde A} \leq \frac{p}{4T}$, then the exit probability of the sparsified graph is similar to that of the unsparsified graph.  We will see that this can be achieved by sparsifying to a graph with degree that is polynomial in $T / \sqrt{p}$.

{\em 1. Poisson sparsification.}  One method for sparsification results
in graphs that are {\it not} regular, but where each vertex's degree is drawn
from the same distribution.  We will choose a parameter $D>0$ as the {\em average} degree.
Between large supervertices, include edge
$(u,x)$--$(v,y)$ with probability $ \frac{e_{u,v}}{s_us_v}\cdot \frac{D}{\lambda} = \frac{t_{uv}}{\sqrt{s_us_v}}\cdot \frac{D}{\lambda}$.  Each vertex
then has a degree distribution that is close to Poisson with parameter
$ D$. We will see below that the accuracy of this approximation increases with $D$.  With high probability all of the nodes will have degree $\leq 3D$ (we provide a more rigorous justification of this later in the subsection).  Therefore, Hamiltonian simultation will still be efficient for Poisson sparsification.

In order to show that the quantum algorithm can run efficiently, we will need an upper bound on the degree of any vertex.  (In principle, we could accept a small fraction of vertices with very high degree, but the Poisson distribution concentrates well enough that this won't be necessary.)  We discuss this along with the rest of the mathematical analysis of this technique below.

{\em 2. Birkhoff-von Neumann sparsification.}  Another approach to sparsifying will construct a regular graph with a desired degree $D$.
By the Birkhoff-von Neumann theorem, we can write $A = \lambda \sum_P \mu(P) P$ where $P$ ranges over permutations and $\mu$ is a probability distribution over permutations.  Since $A$ is symmetric we can assume that $\mu(P) = \mu(P^{-1})$, or equivalently, write $A$ as
\be
A = \lambda \sum_P \mu(P) \frac{P + P^{-1}}{2}.
\label{eq:Birkhoff-A}\ee
Now we sample $P_1,\ldots,P_D\sim \mu$ and define the $2D$-regular graph with adjacency matrix
\be
\widetilde{A} =\frac{\lambda}{2D} \sum_{i=1}^D P_i + P_i^{-1}.
\ee
Up to the overall scaling $\frac{\lambda}{2D}$ this is an unweighted graph.
This gives an ensemble of random graphs satisfying our requirement.

However, this graph can have multiple edges between the same vertices; we refer to such graphs as ``multi-graphs''.  This can happen  due to two possible scenarios:
\begin{itemize}
\item if $P_i(v)=P_j(v)$ for some $i\neq j$ (or $P_i(v)=P_j^{-1}(v)$, etc.)
\item if any of the $P_i$ contains a 2-cycle, since $P_i$ and $P_i^{-1}$ each contribute to the same edge
\end{itemize}
The presence of these multi-edges does not affect the convergence of $\tilde A$ to $A$ in operator norm (as we will discuss below) but it might affect the classical lower bound.  However, we want to argue that the multi-edges are sufficiently rare in the large supervertices and only prevalent inside small supervertices.  As a result, by paying only a small change in operator norm we can get rid of multi-edges in large supervertices, thereby resulting in a random graph ensemble that classical algorithms cannot traverse.  As a result, we will perform a process we call {\it restructuring} to obtain a graph with no multi-edges in large supervertices and will pay a small cost in operator norm.

{\em Restructuring in Birkhoff-von Neumann sparsification. }First, we consider the case where some of the $P_i$ have one or more
2-cycles.  These edges have doubled amplitudes as compared to other
edges not involved in 2-cycles.  The first step in addressing this is
to replace $\mu$ with a symmetrized distribution $\mu_{\text{sym}}$.   Let $\pi$ be a uniformly random permutation of $\cV$ subject to the condition that $\pi(S_v)=S_v$ for  each $v\in V$; in other words, it permutes the vertices within each supervertex.  Then we obtain a sample from $\mu_{\text{sym}}$ by
drawing $\pi$ in this way,  $P$ from $\mu$ and outputing $\pi P \pi^{-1}$.
For simplicity, we can take $\mu
\triangleq \mu_{\text{sym}}$, since $\mu_{\text{sym}}$ is also a valid
solution to Eqn.~\eq{Birkhoff-A}.

Each permutation $P$ can be written in terms of its cycle
decomposition as $(P_2,P_{>2})$ where $P_2$ denotes the 2-cycles and
$P_{>2}$ denotes the remaining cycles.  Since $A_{x,x}=0$ we must have
$\mu(P)=0$ for any $P$ containing a 1-cycle, so we do not need to
consider any permutations with 1-cycles.  Let $\cV(P_2)$ denote the vertices acted on by $P_2$.

Next, we will construct a bijection $f$ on the set of permutations such that $\cV(P_2)=\cV(f(P)_{2})$ and $P_{>2}=f(P)_{>2}$ but $P_2$ and $f(P)_2$ share very few edges.  To construct this matching, we will produce a permutation $\pi$ acting on $\cV(P_2)$ and depending only on $\cV(P_2)$.  Then we will define $f(P)= \pi P \pi^{-1}$. This notation leaves implicit the dependence of $\pi$ on $\cV(P_2)$; however, since $\pi$ does not change $\cV(P_2)$, the map is indeed a bijection.  Defining $f$ in this way will guarantee  that $\mu_{\text{sym}}(f(P)) = \mu_{\text{sym}}(P)$.  This follows from the fact that the symmetrized distribution $\mu_{\text{sym}}$ is invariant under permuting nodes within a supervertex.  As a result, we may rewrite the adjacency matrix as
\begin{equation}
A = \lambda \sum_{P} \mu(P) \frac{P + P^{-1} + f(P) + f(P)^{-1}}{2}
\end{equation}
Again the RHS is a mixture of weighted  adjacency matrices of graphs that are 2-regular for all vertices except at most one pair per superedge.

It remains to construct $\pi$ from $P$.  For each superedge $e = (u,v)$, we define the vertices involved in 2-cycles between $u$ and $v$ to be $\cV(e,P_2)$.  We note that $\cV(e,P_2)$ and $\cV(e',P_2)$ for $e \neq e'$ are disjoint due to the 2-cycle property.  As a result of the disjointness, we can separately analyze $\cV(e,P_2)$ for each superedge $e$.   We will take $\pi$ to act as the identity on $S_u \cap \cV(e,P_2)$ and to permute $S_v \cap \cV(e,P_2)$ arbitrarily, as long as it has no fixed points.  If $|S_v \cap \cV(e,P_2)|>1$ then this will guarantee that the restrictions of $P$ and $\pi P \pi^{-1}$ to $S_u \times S_v$ have no edges in common.  If $|S_v \cap \cV(e,P_2)|=1$ then $P$ has a single 2-cycle between $S_u$ and $S_v$, and this will not be changed by the action of $\pi$. Define $c(P)$ to be the set of these unmatched 2-cycles; it will contain at most one 2-cycle per superedge. Thus this construction guarantees that for each superedge $(u,v)$ and each permutation $P$, there will be at most one common 2-cycle between $P$ and $f(P)$.   If there were no shared 2-cycles then
\be \frac{1}{2}(P + P^{-1} + f(P) + f(P)^{-1})  \label{eq:almost-2-regular}\ee
would equal the normalized adjacency matrix of a 2-regular undirected graph.  We will use this to obtain an almost regular graph that is close to our original graph in operator norm.

The second source of multi-edges is from sampling.  If two different samples from $\mu$ share an edge, then their sum will have a multi-edge.  At this point, it is too late to change the distribution we are sampling from, but we can simply argue that these multi-edges are rare.

We can bound the number of multi-edges from each source by using the symmetry of $\mu$.  Suppose we sample $P_1,\ldots,P_D$ from $\mu$.  Let $\widetilde P_i \triangleq (P_i + P_i^{-1} + f(P_i) + f(P_i)^{-1})/2$.  We will argue that $\sum_i \widetilde P_i$ is $4D$-regular almost everywhere.
For simplicity, we restrict our attention to a single superedge $e=(u,v)$.  The expected number of multi-edges in $S_u\times S_v$ has a contribution of $\leq D$ from the sets $c(P_1),\ldots,c(P_D)$.
Let the expected number of multi-edges in $S_u\times S_v$ from a single pair $\tilde P_i$ and $\tilde P_j$ be $C_{uv}(i,j)$.  When restricted to $S_u\times S_v$, $\widetilde P_i$ and $\widetilde P_j$ have degree $\leq 2$ everywhere, so each one has $\leq 2 \min(s_u,s_v)$ edges total.   Since there are $s_us_v$ possible edges, and $\mu$ is symmetric under permutations within each supervertex, the expected number of collisions is $\leq \frac{4 \min(s_u,s_v)^2}{s_us_v} \leq  4$.  Across all pairs $\tilde P_i,\tilde P_j$ this is $\leq \binom{D}{2}\cdot 4 = 2 D(D-1)$.  Together with the multi-edges from the $c(P_i)$, we have $\leq 2 D^2$ multi-edges per superedge in expectation.  By Markov's inequality and the union bound, the probability that {\em any} superedge has $\geq \frac{2 D^2 |E|}{\eps}$ multi-edges is $\leq \eps$.  Since this probability is over our random choice of graph, we can choose $\eps=1/2$ and just condition our choice of random graph on not having too many multi-edges within any superedge.  This conditioning will magnify the probability of any other event by at most a factor of 2, so it will not significantly change our analysis of rare events elsewhere, such as in the discussion of classical lower bounds.

Therefore, the hierarchical graph $\mathcal{G}$ has $O(D^2|E|)$ multi-edges per pair of supervertices with high probability.  We cannot rule out the existence of a classical algorithm that can identify these multi-edges, which will yield additional advice that our current lower bound is too weak to handle\footnote{Of course, since the density of multi-edges is $D^2 |E|/|\mathcal{E}_{uv}|$ between supervertices $u$ and $v$ and this is generically exponentially small, no classical algorithm is expected to benefit from knowing multi-edges and we conjecture that an exponential classical lower bound still holds.  However, we do not prove this and instead leave this to future work.}.  We will now perform a \emph{rewiring} of the hierarchical graph $\mathcal{G}$ which results in a new hierarchical graph $\mathcal{G}'$ with no multi-edges in the large supervertices.  Recall that if a classical algorithm makes $Q$ queries, the small supervertices are $u \in V$ such that $|\mathcal{V}_u| \leq Q/\delta$ where $\delta$ is the desired success probability of the classical algorithm.  Therefore, we only need to concern ourselves with removing multi-edges from supervertices $u$ satisfying $|\mathcal{V}_u| > Q/\delta$, without affecting the analysis of the classical lower bound.  In these large supervertices, we additionally want to argue that with high probability, all multi-edges are double-edges and not $k$-edges for $k > 2$.  To show this, we calculate that the probability of choosing a particular edge in $S_u\times S_v$ $k$ times is upper bounded by
\begin{equation}
\binom{D}{k} \left(\frac{2 \min(s_u, s_v)}{s_u s_v}\right)^k \leq \left(\frac{2 D \min(s_u, s_v)}{s_u s_v}\right)^k.
\end{equation}
 By a union bound, the probability that this holds for any edge in the superedge is $\leq s_u s_v \left(\frac{2 D \min(s_u, s_v)}{s_u s_v}\right)^k$.  If $D \ll \max(s_u, s_v)$ (i.e. $D \ll O(Q/\delta)$), then the probability that $k \geq 3$ within any supervertex goes like $O(\delta D^3/Q)$ which is small.  Union bounding over all supervertices, this probability is  $O(\delta D^3 |E|/Q)$, and with high probability the graph exclusively contains single or double-edges in large supervertices.  
 
The rewiring procedure will thus work as follows. For each double-edge connecting between supervertices $u$ and $v$, we can always find a single-edge (as there are only $O(D^2 |E|)$ such edges with high probability and we assume that the supergraph has significantly fewer nodes than the original graph) connecting between supervertices $v$ and $w$ \footnote{Additionally, we also require that the 4 nodes involved in the double and single edge do not also form a rectangle.}.  Call the vertices connected by the weight-2 edge $\alpha$ and $\beta$ and the vertices connected by the weight-1 edge $\gamma$ and $\delta$.  We will (a) remove the $(\gamma, \delta)$ edge, (b) change the $(\alpha, \beta)$ double-edge to be a single-edge, and (c) add single-edges $(\alpha, \gamma)$ and $(\beta, \delta)$.  One can verify that the degrees of the vertices do not change, but the double-edge disappears.  We repeat this for each of the double-edges making sure we choose single-edges that have not already been used for rewiring.  After rewiring is complete, we obtain a new graph $\mathcal{G}'$, whose spectral properties we need to compare to the old graph.  A single rewiring event changes the graph Hamiltonian by
\begin{equation}
A \to A' = A + \frac{\lambda}{2 D} \cdot \, \begin{blockarray}{ccccc}
\alpha & \beta & \gamma & \delta \\
\begin{block}{(cccc)c}
0 & -1 & 1 & 0 & \alpha\\
-1 & 0 & 0 & 1 & \beta\\
1 & 0 & 0 & -1 & \gamma\\
0 & 1 & -1 & 0 & \delta\\
\end{block}
\end{blockarray}
\end{equation}
and thus $\norm{A' - A} \leq \frac{\lambda}{D}$.  All other rewirings have similar corrections with disjoint supports, so after all rewirings, $\norm{A' - A} \leq \frac{\lambda}{D}$ still holds.  (Here we use the assumption that there are no triple or higher collisions.  Of course the argument could be extended to accommodate $o(D)$ collisions if necessary.)  So long as $\lambda/D$ is much smaller than the gap away from the zero-energy eigenstate in $t$, then (as we will argue), the quantum algorithm will still remain efficient.

\subsubsection{Quantum and classical walk performance on sparse graphs}
Recall from earlier in this subsection that if $A$ is the adjacency matrix of the dense graph and $\tilde A$ is the sparse approximation and we run the algorithm for time $T$ with success probability at least $p$ then the algorithm will still succeed on the sparsified graph if $\norm{A - \tilde A } \leq \frac{p}{4T}$.  For instance, in the 1D hierarchical graph models, we had $T = \exp(O(\sqrt{N}))$ and $p = \exp(-O(\sqrt{N}))$.  Thus, when sparsifying, we need to choose the degree $D$ large enough so that the $\norm{A - \tilde A } \leq \exp(-O(\sqrt{N}))$.  For higher dimensional hierarchical graphs (where for simplicity the dimension is $O(1)$), the same statement holds with $T$ and $1/p$ both $\text{poly}(N)$.  To determine the value of $D$ to pick, we will use the matrix Bernstein inequality, which is Theorem 1.4 of \cite{Tropp-LD}.

\begin{theorem}[matrix Bernstein inequality~\cite{Tropp-LD}]
  Let $X_1,\ldots,X_n$ be independent $d$-dimensional Hermitian matrices satisfying
  \ba \bbE[X_k] = 0, \hspace{0.5cm} \lambda_{\max}(X_k) \leq R, \hspace{0.5cm} \norm{\sum_{k=1}^n \bbE[X_k^2]} =   \sigma^2
  \ea
  Then for all $\delta\geq 0$,
  \be
 \mathbb{P}\left[\lambda_{\max}\qty(\sum_k X_k) \geq \delta\right] \leq d \cdot \exp(-\frac{\delta^2/2}{\sigma^2 + R\delta/3}).
  \ee
\end{theorem}

We can use the matrix Bernstein inequality to obtain the following lemma:
\begin{lemma}
For Birkhoff-von Neumann sparsification of the supervertex Hamiltonian $A$, the sparsified Hamiltonian $\widetilde{A}$ satisfies $\norm{A - \widetilde{A}} \leq \frac{p}{4T}$ with high probability if the degree of $\widetilde{A}$ satisfies $D \geq k T^2 \lambda^2 p^{-2} \log |\mathcal{V}|$ for $k$ a large enough constant and $\lambda$ the largest eigenvalue of $A$.

For Poisson sparsification of the supervertex Hamiltonian $A$, the sparsified Hamiltonian $\widetilde{A}$ satisfies $\norm{A - \widetilde{A}} \leq \frac{p}{4T}$ with high probability if the degree of $\widetilde{A}$ satisfies $D \geq k' T^2 p^{-2} \log |\mathcal{V}|$ for $k'$ a large enough constant.
\end{lemma}

As an example, when $T = \exp(O(\sqrt{N}))$ and the supergraph is a 1D line, we require that $D = N \exp(O(\sqrt{N}))$.  For higher dimensional supergraphs with $d = O(1)$, $T = N^{\gamma(d)}$ for some $\gamma(d)$ and we find that $D = O(N^{\gamma(d) + 1})$.

\begin{proof}
For the Birkhoff-von Neumann sparsification we have $d = |\mathcal{V}|$, $n=D$, $X_i = \frac{\lambda \widetilde{P_i} - A}{D}$ with $\tilde P_i = (P_i + P_i^{-1} + f(P_i) + f(P_i)^{-1})/2$.  Then $\tilde A = A + \sum_{i=1}^D X_i$ and $\bbE[X_i]=0$.  We can bound
\begin{equation}
  \lambda_{\max}(X_i) \leq \frac{\lambda}{2D}\norm{P_i + P_i^{-1} + f(P_i) + f(P_i)^{-1}} + \frac{\norm{A}}{D} \leq \frac{2 \lambda}{D}
  \qand \sigma^2 \leq \frac{4 \lambda^2}{D}
\end{equation}

If the quantum algorithm runs for time $T$ and we require $ \norm{\sum_i X_i} = \norm{A - \tilde A } \leq \frac{p}{4T}$, then we will choose $\delta = \frac{p}{4 T}$.  Substituting into the matrix Bernstein bound and assuming $D \gg \lambda$, we find that
\be
\mathbb{P}\left[\lambda_{\max}\qty(\sum_k X_k) \geq  \frac{p}{4 T}\right]
\leq |\mathcal{V}| \cdot \exp(-c \frac{p^2 D}{\lambda^2 T^2}).
  \ee
for some constant $c$.  If $D$ is large relative to $T^2\lambda^2a^{-2} \log(|\cV|)$ then most of the sparsified matrices drawn from this distribution will be close in norm to the true adjacency matrix.

For the Poisson sparsification, recall that the probability of placing an edge between vertices $(u,x)$ and $(v,y)$ is
\be
p_{u,v} \triangleq \frac{t_{uv}}{\sqrt{s_us_v}}\cdot \frac{D}{\lambda}
\ee
and that $\sum_v \sum_{y \in S_v} p_{u,v} = \sum_v p_{u,v}s_v=D$ for all $u$.  Let $\tilde A$ be the adjacency matrix of the graph that we construct via Poisson sparsification.  Since each vertex has expected degree $D$ in $\tilde A$ but weighted degree $\lambda$ in $A$, we will need to compare $\frac \lambda D \tilde A$ with $A$.
To analyze the spectrum of $\tilde A$, define the independent random variables
\be
X_{(u,x),(v,y)} \triangleq ( \ket{u,x}\bra{v,y} + \ket{v,y}\bra{u,x})  (\text{Bern}(p_{u,v}) - p_{u,v}),
\ee
where $\text{Bern}(\cdot)$ denotes a Bernoulli random variable.  We can write
\be
\tilde A - \frac{D}{\lambda} A = \sum_{(u,x),(v,y)} X_{(u,x),(v,y)}
\ee
This  ensures that $\bbE[X_{(u,x),(v,y)}]=0$ and $\lambda_{\max}(X) \leq 1$.  For $\sigma^2$ we calculate
\ba
\sum_{(u,x),(v,y)} \bbE[X_{(u,x),(v,y)}^2]  & =
\sum_{(u,x),(v,y)} (\dyad{u,x} + \dyad{v,y}) p_{u,v}(1-p_{u,v}) \nonumber \\
& \leq
\sum_{u,x} \sum_v 2 s_v p_{u,v} \dyad{u,x}
\nonumber \\ & = 2 D I
\ea
so $\sigma^2 \leq 2D$.
Matrix Bernstein now implies that for $\eps<1$,
\be
\mathbb{P}\left[\norm{\frac{\lambda}{D} \tilde A - A}\geq \eps\lambda\right]
\leq
|\cV| \cdot \exp( -\frac{D^2\eps^2/2}{2D + D\eps/3})
\leq
|\cV| \cdot \exp( -cD \eps^2).
\ee
Setting $\eps = p/4T$ we find a similar  bound to the Birkhoff-von-Neumann-based sparsification case.
\end{proof}

Note that to handle random sums of permutations (or more generally unitary matrices), it often suffices to use a simpler concentration method, such as the matrix Chernoff or Hoeffding bounds (see Thms 1.1-1.3 of \cite{Tropp-LD}).  These bounds would work equally well for the permutation-based sparsification.  However, these methods often perform badly on sums of low-rank matrices, which we encounter in  Poisson sparsification.  For example, the corresponding value of $\sigma^2$ arising in the matrix Hoeffding bound would scale with $\max_v s_v$.  In this case, the matrix Bernstein inequality is necessary, and for ease of exposition, we use it in both cases.

In the Birkhoff-von Neumann sparsification, we additionally had to perform a restructuring procedure to obtain matrix $A'$ which is $D$-regular and spectrally close to $\widetilde{A}$.  We therefore have
\begin{corollary}
After the restructuring procedure of the Birkhoff-von Neumann sparsification, the $D$-regular matrix $A'$ that is obtained satisfies $\norm{A' - A} \leq p/4 T$ with high probability.
\end{corollary}
\begin{proof} 
Follows from $\norm{A' - A} \leq \norm{A' - \widetilde{A}} + \norm{\widetilde{A} - A} \leq p/4T + \lambda/D  \lesssim p/4 T$ with high probability since $D \gg T^2 p^{-2} \log |\mathcal{V}|$.
\end{proof}
For the case of Poisson sparsification, we know that the degree distribution is random.  For the quantum algorithm to work, we need all of the nodes to have relatively low degree.  To ensure this, we need the following concentration property:
\begin{lemma}
Given a Poisson sparsified graph $\widetilde{A}$ with mean degree $D \gg \log |\mathcal{V}|$, with high probability no vertex has degree higher than $3D$.
\end{lemma}
\begin{proof}
We start with the well known fact that the moment-generating function of the Poisson distribution is
\be \bbE[e^{tX}] = e^{D\qty(e^t-1)}\ee
where $X$ is a Poisson-distributed random variable with mean $D$.  However, in the Poisson sparsified graph, the degrees are not precisely Poisson distributed.  A vertex in $S_u$ has degree that is a sum of $s_v$ copies of $\text{Bern}(p_{u,v})$ for each $v\in V$.  If $X$ denotes a random variable corresponding to the resulting degree, then the moment generating function of $X$ is
\ba
\bbE[e^{tX}] & =  \prod_v (p_{u,v} e^t + (1-p_{u,v}))^{s_v} \nonumber\\
& =  \prod_v (1 + p_{u,v} (e^t-1))^{s_v} \nonumber\\
& \leq \exp( \sum_v s_v p_{u,v} (e^t-1)) \nonumber\\
 & =  \exp( D (e^t-1))
\ea
This distribution is thus stochastically dominated by the Poisson distribution, which implies that we can upper bound the probability of a large deviation using the corresponding bound from the Poisson distribution. Taking $t=1$ we  have
\be \mathbb{P}[X \geq \alpha D]  \leq
e^{D(e-1 - \alpha)}\ee
If $D \gg \log|\cV|$ (as we have assumed above) then we can choose $\alpha=3$ and find that the probability that $X\geq 3D$ is $\ll 1/|\cV|$, meaning that with high probability no vertex has degree higher than $3D$.  Our assumption that $D \gg \log |\cV$ is only for convenience, and we could still find an efficient quantum algorithm when $D$ is smaller by choosing a larger value of  $\alpha$.
\end{proof}
This fact allows us to provide a runtime analysis of the quantum walk algorithm:
\begin{theorem}
Consider a dense hierarchical graph with supervertex Hamiltonian $t$ and graph adjacency matrix $A$.  Suppose the exit probability satisfies
\begin{equation}
\mathbb{P}(\text{EXIT}) = \frac{1}{T} \int_0^T dt\,|\mel{S_{\text{exit}}}{e^{i A t}}{S_{\text{init}}}|^2 \geq p.
\end{equation}
Then,
\begin{itemize}
    \item a Birkhoff-von Neumann sparsified graph with adjacency matrix $\widetilde{A}$, which when restructured to $A'$, is $D$-regular with high probability, with $D > c T^2 \lambda^2 p^{-2} \log |\mathcal{V}|$.  Furthermore, the restructured $A'$ satisfies $\frac{1}{T} \int_0^T dt\,|\mel{S_{\text{exit}}}{e^{i A' t}}{S_{\text{init}}}|^2 \geq p/2$ with high probability, and the quantum walk can be simulated with $\mathrm{poly}(D, T/p)$ queries to the oracle.
    \item a Poisson sparsified graph with adjacency matrix $\widetilde{A}$ has maximum degree $\leq 3D$ with high probability, with $D > c T^2 p^{-2} \log |\mathcal{V}|$.  Furthermore, $\widetilde{A}$ satisfies $\frac{1}{T} \int_0^T dt\,|\mel{S_{\text{exit}}}{e^{i \widetilde{A} t}}{S_{\text{init}}}|^2 \geq p/2$ with high probability, and the quantum walk can be simulated with $\mathrm{poly}(D, T/p)$ queries to the oracle.
\end{itemize}
\end{theorem}

Next, we consider whether classical algorithms fail to efficiently traverse these graphs.  The analysis is nearly identical to that of Section~\ref{sec:LB}, so we will state the main result for brevity:
\begin{conjecture}[informal]
If the supergraph is a $d$-dimensional lattice with $d = O(1)$, then the number of queries that any classical algorithm needs to make to traverse the sparsified graph is $\geq (\max_i |s_i|)^{\Omega(1)}$.
\end{conjecture}
\begin{proof}[Proof idea:]
  We will sketch the essential elements of a potential proof.  First, while the Birkhoff-von Neumann based sparsification procedure yielded a $D$-regular graph, the last rewiring step was technically not needed, since before then multi-edges are rare between supervertices. Thus if the multi-edges are left intact, a classical algorithm is unlikely to find them.  (The classical lower bound in \sec{LB} does not handle the case of rare multi-edges, but we believe it could be adapted to do so.)
  If we average over the random graph obtained by sparsification, the resulting graph has an exact supervertex subspace; therefore, we might expect the classical algorithm to take exponential time to traverse this graph on average.  A similar argument holds for the Poisson-based sparsification.  The key property in both cases is that if we condition on the part of the graph seen by the classical algorithm, newly seen vertices in large supervertices should be nearly uniformly drawn from that supervertex.  This will ensure that the classical algorithm never sees a cycle involving the hard supervertices and therefore cannot learn useful information about that part of the graph, along the lines of the proof in \sec{LB}.  Making these  claims rigorous may require a careful analysis which we leave to future work.  

\end{proof}

We now briefly provide an explicit example where the sparsification procedure results in hierarchical graphs which are efficiently traversable by a quantum walk.  The example we will choose is simply a sparsification of the original welded trees graph, where the supergraph is a line graph with $2n$ nodes.  

\begin{example}
Define the matrix $t$ to be a hopping Hamiltonian where $\mel{S_i}{t}{S_{i+1}} = \mel{S_{i+1}}{t}{S_{i}} = 1$ for $i \neq n$, $\mel{S_i}{t}{S_{i+1}} = \mel{S_{i+1}}{t}{S_{i}} = \sqrt{2}$ for $i = n$, and all other matrix elements are zero.  We solve for the spectrum of this Hamiltonian, and find that $2n-2$ of the $2n$ eigenstates are scattering states and are delocalized.  However, the bottom and top eigenvectors are bound states \cite{Childs}.  In particular, the amplitudes of the top eigenvector $\ket{\psi}$ satisfy $\braket{j}{\psi} \propto \sinh (k j)$ for $1 \leq j \leq n$ and $\braket{j}{\psi} \propto \sinh (k (2n+1-j))$ for $n+1\leq j \leq 2n$.  The solution for $k$ satisfies
\begin{equation}
\frac{\sinh((n+1)k)}{\sinh nk} = \sqrt{2}.
\end{equation}
Therefore, since the top eigenvector is $\sqrt{\pi}$, the size of each supervertex $N \pi$ grows and shrinks like $1,2,4,\cdots,2^{n-1}, 2^n, 2^{n}, 2^{n-1},\cdots, 4,2,1$.  We connect edges in an all-to-all manner between supervertices, and place weights $1/\sqrt{s_{j} s_{j+1}}$ on all edges where $j \neq n$ and $\sqrt{2/(s_{j} s_{j+1})}$ for $j = n$.  This dense graph can then be sparsified either using the Poisson or Birkhoff methods.
\end{example}


\section{Hierarchical graphs with diagonal disorder}\label{sec:Anderson}
For the class of random graph models we have considered in previous sections, we assume that nodes in supervertex $S_i$ are connected to nodes in $S_{i-1}$ and $S_{i+1}$.  It turns out that we may obtain an exact analysis in the case for which nodes within $S_i$ are connected to each other. For one-dimensional hierarchical graphs, following the model in \sec{random-graph} where there are $e_i$ edges connecting nodes in supervertex $S_i$ and $S_{i+1}$, we now consider a new random graph model where there are $\mathcal{F}_i$ nodes connecting nodes in supervertex $S_i$ to other nodes in supervertex $S_i$.  The total degree of each node is still $D$, and the number of edges from a node in one supervertex to a node in a consecutive supervertex is the same for each node in the former supervertex.

Due to the construction of the random graph model, the system still evolves in the subspace of supernode states.  The effective Hamiltonian in the supernode space has the elements
\begin{equation}
    \mel{S_i}{H}{S_j} = \begin{cases}
\frac{e_i}{\sqrt{s_i s_{i+1}}}, & \text{for } j=i+1\\
\frac{e_{i-1}}{\sqrt{s_{i-1} s_{i}}}, & \text{for } j=i-1\\
\frac{2 \mathcal{F}_{i}}{ s_{i} }, & \text{for } j=i\\
0, & \text{otherwise }\\
\end{cases}
\end{equation}
This effective Hamiltonian resembles an Anderson model, where both the hopping and onsite potentials can be potentially random.  We can analyze the dynamics of the above Hamiltonian by making some assumptions on the structure of the random graph.  First, we will denote $f_i$ to be the number of edges connecting a node in supervertex $i$ to other nodes in supervertex $i$.  We immediately obtain the following generalization of our results:

\begin{theorem}\label{thm:onsite-superpolynomial}
If $f_i$ is constant for all $i$, then this class of random graphs will exhibit a superpolynomial quantum-classical separation.
\end{theorem}
\begin{proof}
If $f_i = f$ for all $i$, then $2\mathcal{F}_i/s_i = 2 f$ is a constant.  The Hamiltonian has purely off-diagonal disorder and is shifted by $2f$ times identity.  The zero mode is located at energy $E = 2f$ and the gap to the zero mode has an unchanged behavior.  Therefore, the analysis in \sec{random-graph} holds and one finds a superpolynomial quantum-classical separation.
\end{proof}

Therefore, for particular cases of onsite potentials, we can find vast generalizations of the results we derived.  However, for generic onsite potential, the eigenstates may localize, preventing a speedup.  An example is illustrated below.

\begin{theorem} \label{thm:Andersonspeedup}
Restrict $e_i = c \sqrt{s_i s_{i+1}}$ for all $i$, i.e., the number of edges going from supervertex $i$ to $i+1$ is a constant times the geometric mean of the number of nodes in supervertex $i$ and the number of nodes in supervertex $i+1$.  Allow $2 \mathcal{F}_i/s_i \sim U[0,1]$ to be uniformly random.  Define such a class of random graphs to be $A_1$, with $A_d$ for $d > 1$ an analogously defined random graph on a $d$-dimensional hypercube.  This ensemble of random graphs does not exhibit a superpolynomial classical-quantum speedup.
\end{theorem}

To satisfy the constant degree constraint, we restrict the random graph ensemble to those satisfying $D s_i = 2 \mathcal{F}_i + c \sqrt{s_i} \left(\sqrt{s_{i-1}} + \sqrt{s_{i+1}}\right)$.  Then, the effective Hamiltonian for this random graph model has diagonal (onsite) disorder and constant off-diagonal elements.

The proof of the theorem above requires a well-known result known as the Furstenberg's theorem:
\begin{theorem}[Furstenberg~\cite{furstenberg1963noncommuting}, from~\cite{gorodetski2021parametric}]
Let $\{X_k,k \geq 1\}$ be independent and identically distributed random variables, taking values in $SL(d, \mathbb{R})$.  Let $G_X$ be the smallest closed noncompact subgroup of $SL(d, \mathbb{R})$ containing the support of the distribution of $X_k$, and assume that
$\mathbb{E}[\log \norm{X_k}] < \infty$.  Further assume that there exists no $G_X$-invariant finite union of proper subspaces of $\mathbb{R}^d$. Then there exists a positive constant $\lambda_F$ such that with
probability one
\begin{equation}
    \lim_{n \to \infty } \frac{1}{n} \log \norm{X_n \cdots X_2 X_1} = \lambda_F > 0.
\end{equation}
\end{theorem}
The use of Furstenberg's theorem can be seen through the following representation of the eigenstates.  The eigenvalue equation satisfies
\begin{equation}
    \psi_{n-1} + u_n \psi_n + \psi_{n+1} = E \psi_n
\end{equation}
and can be rearranged in the form
\begin{equation}
    \begin{pmatrix}
           \psi_{n-1} \\
           \psi_{n}
         \end{pmatrix} = \begin{pmatrix}
           E - u_n & -1 \\
           1 & 0
         \end{pmatrix} \begin{pmatrix}
           \psi_{n} \\
           \psi_{n+1}
         \end{pmatrix} = X_n \begin{pmatrix}
           \psi_{n} \\
           \psi_{n+1}
         \end{pmatrix}.
\end{equation}
In general, the amplitudes at site $k$ can be determined by analyzing the Lyapunov exponents of a product of $X_i$ from $i = 0$ to $i = k$.  As $\det X_i = 1$, the Furstenberg theorem can be applied and thus all eigenfunctions decay exponentially.

The proof of \thm{Andersonspeedup} can then be seen as follows:
\begin{proof}[Proof of \thm{Andersonspeedup}]
  Under the time evolution protocol, we will need to compute an upper bound rather than a lower bound.  This can be done by 
  \begin{align}
    \mathbb{P}(\text{EXIT}) &= \frac{1}{\tau} \int_0^\tau \dd t\,\left|\mel{S_{2n}}{e^{-iHt}}{S_0}\right|^2 \nonumber \\
    &= \frac{1}{\tau} \sum_{E,E'} \int_{0}^{\tau} \dd t\, e^{-i(E-E') t} \braket{S_{2n}}{E} \braket{E}{S_0} \braket{S_{2n}}{E'} \braket{E'}{S_0}, \nonumber \\
    &\leq \sum_{E,E'} \left|\braket{S_{2n}}{E}\right| \left|\braket{E}{S_0}\right| \left|\braket{S_{2n}}{E'}\right| \left|\braket{E'}{S_0}\right|.
\end{align}
A given energy eigenstate decays exponentially, and assuming that the Lyapunov exponent is $\lambda_F > 0$, we know that $\left|\braket{S_{2n}}{E}\right| \left|\braket{E}{S_0}\right| \leq C e^{-\lambda_F n}$ for some constant $C$.  Therefore,
 \begin{align}
    \mathbb{P}(\text{EXIT}) &\leq \sum_{E,E'} e^{-2 \lambda_F n} \leq C' n^2 e^{-2 \lambda_F n}
\end{align}
and thus at least an exponential number of repetitions are necessary to boost this probability to a constant.
\end{proof}

The Furstenberg theorem is powerful because it tells us that as long as the onsite potentials are independent and randomly distributed, there cannot be a quantum-classical superpolynomial separation.  However, the theorem is only suitable in one dimension.  A popular conjecture, which was first understood in the seminal paper of Abrahams, Anderson, Licciardello, and Ramakrishnan~\cite{AALR} using a renormalization group procedure, is the following:

\begin{fact}[Anderson transition]\label{cor:Anderson-transition}
For a Hamiltonian with nearest neighbor hopping and random onsite potential in a $d$-dimensional lattice, a nonzero fraction of eigenstates are delocalized when $d \geq 3$.
\end{fact}
While this result is (overwhelmingly) expected to be true, there is still no completely rigorous mathematical proof.  However, if there were, this immediately allows us to claim:
\begin{corollary}
Given that the Anderson transition occurs at $d\geq 3$, the class random graphs $A_d$ for $d \geq 3$ will exhibit a superpolynomial quantum-classical separation.
\end{corollary}
This result follows from the analyses throughout the paper.  The only subtlety is to check whether the density of states (spectral measure) is not divergent in the mobility edge.  However, this is true from well-known results due to Wegner studying the density of states in disordered systems, see \cite{Wegner}.

It is somewhat pessimistic to conclude that all one-dimensional systems with onsite disorder cannot have a mobility edge (which describes the transition from localized to delocalized eigenstates as a function of energy).  In fact, if we assume the disorder is not random but rather quasiperiodic, then one can construct models corresponding to unnatural random graph ensembles, but which have a mobility edge and thus a region of delocalized states.  An example of this is the generalized Aubry-Andre model, constructed in \cite{GPD}, where the onsite potentials are chosen to take the following form:
\begin{equation}
    u_n = 2 \lambda \frac{\cos(2\pi n b + \phi)}{1-\alpha \cos(2\pi n b + \phi)}
\end{equation}
where $b = \frac{1}{2}(\sqrt{5}+1)$.  The equations for the boundaries of the mobility edges are $\alpha E = 2 \text{ sgn}(\lambda) (1 - |\lambda|)$.  Although we expect a superpolynomial quantum speedup in traversing such graphs, the irrationality of the onsite potential makes an explicit construction of such a random graph ensemble challenging.

There are also classes of random graph ensembles which are equivalent to an effective Anderson model with onsite disorder.  An example of such graphs were studied in \cite{HDC}, which also bears similarity to the quantum walk algorithm for the NAND trees problem in \cite{v004a008} as well as the construction proving that quantum walks are capable of universal quantum computation \cite{Childs}.  For the models discussed in \cite{HDC}, the authors proved the existence of quantum scar states, which were exact eigenstates with distinct localization properties from the rest of the states, causing the Hamiltonian to fail to thermalize.  Although the scar states in the aforementioned models were localized, it might be possible to construct models with delocalized scar states, which would be another route to proving a quantum-classical separation.

\section{Conclusions}\label{sec:concl}
In this paper, we studied superpolynomial speedups for quantum walks in random landscapes. In particular, we studied a wide class of one-dimensional chain random graph models with superpolynomial quantum-classical separation in hitting time. For higher dimensions, we studied $d$-dimensional Lieb lattices where we also prove a random graph ensemble with superpolynomial quantum speedup. We also proved classical query lower bounds for these graph problems and provided more general graph models with quantum-classical separations using graph sparsification. 

Our work leaves several natural open questions for future investigation:
\begin{itemize}
\item For one-dimensional hierarchical graphs, can we achieve superpolynomial quantum speedup for graphs with an onsite potential?  This seems possible if one allows both hopping and onsite potential to be random.

\item Can we achieve superpolynomial quantum speedups for other higher-dimensional hierarchical graphs?  This question contains three perspectives. In this paper, we assumed that the underlying randomness was generated via a GFF, which is the limiting distribution for a wide range of combinatorial/statistical mechanics problems, such as configurations of classical dimers on a lattice \cite{kenyon}.  However, if we consider distributions different from Gaussian free field, this may imply different scaling behavior of correlations in low dimensions and potentially a different quantum-classical separation.  Ideas from combinatorics, statistical mechanics, and field theory will give us insight into correlation decay for such distributions.  Second, can we generalize beyond the gauge constraint? In general, even when the gauge constraint is violated, we still have a manifold of delocalized zero modes, but the analysis required to prove that this implies a subexponential time quantum walk algorithm becomes considerably more tricky.  Lastly, we mainly considered a $d$-dimensional cubic Lieb lattice.  If one considers other lattices, could there be a drastic difference in runtimes?  For instance, on a honeycomb lattice, the spectrum in the absence of disorder contains Dirac cones. This has previously been used for quantum algorithms for spatial search~\cite{ChildsGoldstoneDirac, ChildsGe}.  Analyzing this in the presence of disorder remains an open challenge.

\item Do our results have implications to superpolynomial quantum-classical separation of other problems? For example, the welded tree has been applied to constructing the superpolynomial quantum speedup of graph property testing with the adjacency list model~\cite{BCG+20} as well as adiabatic quantum computation with no sign problem~\cite{GSV21}, and it is natural to seek for further algorithmic applications of our random graph ensembles.  Furthermore, it would be interesting if there existed an efficient adiabatic algorithm for the particular random graph ensembles we consider.

\item
  What are other general conditions on achieving exponential speedups if we break the supervertex subspace or regularity conditions imposed on our random graph models?  We argued that sparsification provide a general framework for finding graphs with efficient quantum runtimes that do not satisfy regularity.  However, we wonder whether these constructions can be further generalized.  Powerful results from extremal graph theory such as Szemeredi's regularity lemma imply that large graphs can be clustered into supernodes that obey an approximate regularity.  However, for such graphs, approximate regularity may only be obtained in the regime where one still requires an exponentially large number of supernodes; this unfortunately would elude a superpolynomial speedup.   Separately, finding other natural graph ensembles with a small (exact or approximate) supervertex subspace would also be another interesting avenue.

\item One important aspect crucial to the quantum speedup is the existence of a small Krylov subspace where the effective dynamics evolves in.  This is reminiscent of many-body scars \cite{bernien2017probing, turner2018weak}, where a small subspace of states of a quantum many body system appear athermal.  It would be interesting to find other algorithmic applications where scarred subspaces play a role in finding speedups.

\item There exists a mapping from a hierarchical graph to the configuration space of a many-body Hamiltonian.  Given some many-body Hamiltonian $H = \sum_{k} \mathcal{O}_k$ on $n$ qubits where $\mathcal{O}$ denotes a k-local operator, we may consider the configuration space where nodes are labelled by strings $\vec{s} \in \{0,1\}^n$, and two nodes $\vec{s}, \vec{s}'$ are connected if $\mel{\vec{s}}{H}{\vec{s}'} \neq 0$.  We wonder when these configuration graphs resemble hierarchical graphs and what constraints on the Hamiltonian ensure an efficient quantum walk runtime on these graphs.  This may provide a more natural setting where one can achieve a black-box speedup.  For the examples we provide, $H$ is very non-local.  However, it might be possible to use perturbative gadgets to realize $H$ (or some modification of $H$) in terms of sums of $k$-local terms.

\end{itemize}


\section*{Acknowledgements}
We thank Andrew Childs and Eddie Farhi for valuable discussions, as
well as Hari Krovi for pointing out~\cite{krovi2007quantum}.  SB was
supported by the National Science Foundation Graduate Research
Fellowship under Grant No.~1745302 and NTT (Grant AGMT DTD 9/24/20). TL was supported by the National
Science Foundation Grant PHY-1818914 and a Samsung Advanced Institute
of Technology Global Research Partnership.  AWH was funded by NSF
grants CCF-1729369, PHY-1818914, the NSF QLCI program through grant
number OMA-2016245 and the DOE Co-design Center for Quantum Advantage
(C2QA) under contract number DE-SC0012704.  There is no data associated with this manuscript.


\bibliographystyle{plain}
\bibliography{welded}

\end{document}